\documentclass[journal,10pt,twocolumn,twoside]{IEEEtran} 
\IEEEoverridecommandlockouts
\usepackage{amsmath}
\usepackage{amsfonts}
\usepackage{cases}
\usepackage{setspace}
\usepackage{fancybox}
\usepackage{subfigure}
\usepackage{epsfig}
\usepackage{graphicx}
\usepackage{epstopdf}
\usepackage{float}
\usepackage{multirow}
\usepackage{color}%
\usepackage{amsmath}
\usepackage{multirow}

\usepackage{indentfirst}
\usepackage{dsfont}
\usepackage{amsfonts}
\usepackage{times,amsmath,color,amssymb,epsfig,cite,subfigure,algorithm,algorithmic}

\newenvironment{proof}[1][Proof]{\begin{trivlist}
\item[\hskip \labelsep {\bfseries #1}]}{\end{trivlist}}

\newcommand{\qed}{\nobreak \ifvmode \relax \else
      \ifdim\lastskip<1.5em \hskip-\lastskip
      \hskip1.5em plus0em minus0.5em \fi \nobreak
      \vrule height0.40em width0.6em depth0.25em\fi}

\newtheorem{lemma}{Lemma}
\newtheorem{theorem}{Theorem}

\newtheorem{proposition}{Proposition}

\begin{document}
%
\title{Efficient QP-ADMM Decoder for Binary LDPC Codes and Its Performance Analysis}

\author{Jing Bai,
        Yongchao Wang,
        Qingjiang Shi}

%

\markboth{}%
{}

\maketitle

\begin{abstract}
This paper presents an efficient quadratic programming (QP) decoder via the alternating direction method of multipliers (ADMM) technique, called QP-ADMM, for binary low-density parity-check (LDPC) codes. Its main contents are as follows: first, we relax maximum likelihood (ML) decoding problem to a non-convex quadratic program. Then, we develop an ADMM solving algorithm for the formulated non-convex QP decoding model.
In the proposed QP-ADMM decoder, complex Euclidean projections onto the check polytope are eliminated and variables in each updated step can be solved analytically in parallel.
Moreover, it is proved that the proposed ADMM algorithm converges to a stationary point of the non-convex QP problem under the assumption of sequence convergence.
We also verify that the proposed decoder satisfies the favorable property of the \emph{all-zeros assumption}.
Furthermore, by exploiting the inside structures of the QP model, the complexity of the proposed algorithm in each iteration is shown to be linear in terms of LDPC code length. Simulation results demonstrate the effectiveness of the proposed QP-ADMM decoder.
\end{abstract}


\IEEEpeerreviewmaketitle

\section{Introduction}
Low-density parity-check (LDPC) codes were first introduced by Gallager in 1962 \cite{Gallager}.
Later, owing to the rediscovery of LDPC codes by MacKay and Neal in the 1990s, significant attention was drawn to LDPC codes due to their excellent near-Shannon performance \cite{MacKay-1} \cite{MacKay-2}.
{The decoders most used for LDPC codes are based on belief propagation (BP) algorithm and its variants \cite{BP-decoder1,channel-code,BP-decoder,BP-implement-TSP1,BP-implement-TSP2,BP-implement-TSP3}.}
However, BP decoding usually suffers from {\it error floor effects} in high signal-to-noise ratio (SNR) regions.
Moreover, analyzing the performance of general LDPC codes using BP decoding in theory is very challenging because of its message update rules and the graphical structure of the code such as stopping sets
\cite{stopping-sets}, trapping sets \cite{trapping-sets} and graph-cover pseudo-codewords \cite{graph-cover-pseudocodewords}.

In the past decade, mathematical programming (MP) techniques, such as linear programming (LP), have attracted increasing attention for decoding LDPC codes due to their strong theoretically-guaranteed
decoding performance \cite{FeldmanLP} \cite{MP-decoding}. In MP decoding, the maximum-likelihood (ML) decoding problem is modeled as a linear integer programm and then relaxed to an MP problem, which can be solved by free MP solvers \cite{Sedumi} \cite{CVX} in polynomial computational complexity.
Therefore, MP decoding offers some desirable advantages in comparison with BP decoding. For example, LP decoding has the \emph{ML-certificate} property, i.e., if the output of the LP decoder is an integer solution, then it is guaranteed to be an ML codeword. Moreover, some important decoding performance, such as convergence and complexity, can be analyzed by the existing mathematical tools \cite{FeldmanLP} \cite{LP-error-floor2}.
However, there still exist two major drawbacks in MP decoding: high-computational complexity and poor error-correction performance in low signal-to-noise
(SNR) regions \cite{FeldmanLP}\cite{MP-decoding}\cite{LP-error-floor1,LP-error-floor2,LP-error-floor3}.
To overcome the drawbacks, authors in \cite{adaptive-LP} proposed an adaptive small-sized LP decoding problem by starting from a simple LP model and adaptively adding constraints.
In \cite{check-node-decompetition1} and \cite{check-node-decompetition2}, two new LP decoding models, whose variables and constraints grew linearly with the check node degree, were independently established.
Authors in \cite{cut-plane-algorithm,separation-cut,Adaptive-cut} improved the error-correction performance of LP decoding by using different cut-generating algorithms to find redundant parity-check cuts that can tighten the feasible region of the corresponding optimization model and cut off undesirable fractional solutions.
Other approaches toward improving the performance of LP decoding include facet guessing \cite{Guessing-facets} and branch-and-bound algorithms \cite{branch-and-bound1}\cite{branch-and-bound2}.
However, all of the above MP decoders are based on general-purpose LP solving algorithms, such as interior point method \cite{interior-point} and simplex method \cite{revised-simplex}, whose high-computational complexities limit their practical applications.

In recent years, alternating direction method of multipliers (ADMM) technique was introduced into the MP decoding area \cite{Barman-ADMM}. Although the corresponding ADMM-based decoder has a lower complexity than general LP decoders, its complexity is still more expensive than the classical BP decoders because time-consuming Euclidean projection onto the check polytope is involved in the decoding procedure.
To address this challenge, authors in \cite{efficient-projection1} proposed a different check projection algorithm based on the cut-search algorithm of \cite{Adaptive-cut}.
In \cite{efficient-projection2}, authors transformed the projection operation from the check polytope to a simplex.
{In \cite{Draper-ACSSC2015}, authors designed a hardware compatible projection onto the check polytope by combing the methods in  \cite{efficient-projection1} and \cite{efficient-projection2}.}
In \cite{efficient-projection3},  an iterative check projection algorithm to reduce the complexity of the projection operation was proposed.
{Authors in \cite{look-up-SPL} simplified the check projection operation by jointly using uniform quantization and look-up tables.}
Moreover, a projection reduction method \cite{jiao-zhang} and two different approaches based on node-wise scheduling
and horizontal-layered scheduling in \cite{Node-Wise-Scheduling} and \cite{Horizontal-Layered-Scheduling} respectively were proposed to reduce the decoding complexity by decreasing the number of projections.
Furthermore, authors in \cite{our-admm-lp} introduced an efficient ADMM-based LP decoding algorithm, which can eliminate the expensive projection operation onto the check polytope.
On the other hand, to improve the error-correction performance of the LP decoder in low SNR regions, authors in \cite{penalty-decoder} designed ADMM-based penalized decoders by adding the penalty term into the objective function of the LP model \cite{Barman-ADMM}.
Authors in \cite{improve-penalized-irregular} and \cite{piece-penalized} optimized the penalized decoders by modifying the penalty functions.
Authors in \cite{lp-box-decoder} added an $\ell_{p}$-box constraint into the ML decoding model and then proposed a different ADMM-based decoder.

As a whole, although ADMM-based decoders show advantages over conventional MP decoders, they still need to be improved from a practical viewpoint.  For example, state-of-the-art ADMM-based decoders, such as in \cite{efficient-projection1}--\!\!\cite{our-admm-lp}, can reduce computation complexity per iteration to be linear to the length of LDPC code length. However, their error correction performance is inferior to BP decoders. Meanwhile, some others, such as in \cite{penalty-decoder}--\!\!\cite{lp-box-decoder}, can improve error correction performance, but their expensive computational complexity limits their practical applications.
{In addition, of particular relevance are the works \cite{Draper-ICASSP2017} \cite{Draper-TSP} where authors investigated an ADMM-based penalized decoder in hardware implementation while the check ploytope projection still dominates the overall hardware resources and power consumption. }

In this paper, we focus on designing a new {check-polytope-free} ADMM decoder. By combining ideas of relaxation of the three-variables parity-check equation in \cite{check-node-decompetition1}, quadratic penalty in \cite{penalty-decoder} and {LP decoder in \cite{our-admm-lp}} respectively, we obtain a new decoder, called QP-ADMM, for LDPC codes, which has favorable error correction performance meanwhile providing cheap linear complexity in terms of LDPC code length. The main technical contributions of this paper are twofold:
{
\begin{enumerate}
\item Efficient implementation: compared with the state-of-the-art ADMM-based MP decoders \cite{efficient-projection1}--\!\!\cite{Horizontal-Layered-Scheduling} and \cite{penalty-decoder}--\!\!\cite{Draper-ICASSP2017}, {the main novelty of the proposed QP-ADMM decoder is the reduction of the Euclidean projections onto check polytopes to simple Euclidean projections onto positive quadrant.} Moreover, all the variables are solved analytically and updated in parallel in each ADMM step.
    {Furthermore, an important improvement in comparison with the LP decoder \cite{our-admm-lp} is that our proposed QP-ADMM decoder achieves better error correction performance than the conventional BP decoder by adding penalty term into the linear objective of LP decoding.}
\item Theoretical analysis: {first, we show that the proposed ADMM algorithm converges to a stationary point of the formulated QP problem if it is convergent. Second, we verify that the decoder satisfies the property of the \emph{all-zeros assumption} and its decoding output can be easily tested whether it is an ML codeword if it is integral.} At last, by exploiting the inside structures of the formulated model, we show that the complexity of the proposed algorithm in each iteration is linear in terms of LDPC code length.
\end{enumerate}}

The rest of this paper is organized as follows.
Section \ref{problem-formulation} shows preliminaries of the formulated non-convex quadratic program for the ML decoding problem of LDPC codes.
In Section \ref{admm-qp-section}, we propose an efficient ADMM algorithm to solve the quadratic program by exploiting its inside structures.
Section \ref{Analysis-admm-qp-decoding} presents the detailed analysis, such as convergence and complexity, of the proposed ADMM-based QP decoder.
Simulation results, which show the effectiveness of our proposed decoder, are presented in Section \ref{simulation-result}.
Section \ref{Conclusion} concludes this paper.

\emph{Notations:} In this paper, we use bold uppercase and lowercase letters to denote matrices and vectors respectively;
$\hat{\mathbf{b}}_{i}$ and $\mathbf{b}_{j}^{T}$ denote the $ith$ column vector and $jth$ row vector of a matrix $\mathbf{B}$ respectively;
the $ith$ entry of a vector $\mathbf{x}$ is denoted by $x_{i}$ and $\|\mathbf{x}\|_{2}$ represents the 2-norm of vector $\mathbf{x}$;
$(\cdot)^{T}$ symbolizes the transpose operation and $[\cdot]_{2}$ denotes the modulus-2 operation;
symbol $\mathbb{R}$ represents the Euclidean space, whose nonnegative orthant is denoted by $\mathbb{R}_{+}$;
$\underset{[a,b]}\Pi(\cdot)$ denotes the Euclidean projection operator onto the interval $[a,b]$;
$\otimes$ denotes the Kronecker product; $\textrm{diag}(\cdot)$ indicates the operation extracting the diagonal vector of the matrix; and the derivative of
a function $g$ with respect to variable $\mathbf{v}$ is denoted by $\nabla_{\mathbf{v}} g$.

\section{ML Decoding Problem and Its equivalent linear integer program }\label{problem-formulation}
   We consider a binary LDPC code $\mathcal{C}$ of length $n$ defined by an $m \times n$ parity-check matrix $\mathbf{H}$.
   Each column of $\mathbf{H}$ corresponding to a codeword symbol, is indexed by $\mathcal{I} := \{1,\cdots,n\}$. Similarly, each row of $\mathbf{H}$ corresponding to a parity check, is indexed by $\mathcal{J}: = \{1,\cdots,m\}$.
   Suppose a codeword $\mathbf{x} \in \mathcal{C}$ is transmitted over a noisy memoryless binary-input output-symmetric channel, resulting in a corrupted signal $\mathbf{r}$.
   Assuming that all of the codewords are transmitted with equal probability, then the ML decoding problem can be formulated as the following optimization problem
   {\setlength\abovedisplayskip{4pt}
     \setlength\belowdisplayskip{4pt}
      \setlength\jot{1pt}
       \begin{subequations}\label{ML formulation}
      \begin{align}
         &\hspace{0cm} \mathop {\min }\limits_\mathbf{x} \hspace{0.5cm} \pmb{\gamma}^{\mathrm{T}}\mathbf{x} \\\
         &\hspace{0.15cm} {\rm{s.t.}} \hspace{0.55cm} \bigg[\sum_{i=1}^{n}{H_{ji}x_{i}}\bigg]_{2} = 0, \hspace{0.6cm} j\in\mathcal{J}, \label{ML formulation_b}\\\
         &\hspace{1.15cm}  \mathbf{x} \in \{0,1\}^n, \hspace{1.55cm} i\in\mathcal{I}, \label{ML formulation_c}
      \end{align}
   \end{subequations}
where} $\pmb{\gamma}$ is a length-$n$ vector of log-likelihood ratios (LLR) defined as
{\setlength\abovedisplayskip{2pt}
     \setlength\belowdisplayskip{2pt}
      \setlength\jot{0.5pt}
     \begin{equation}\label{gamma}
       \gamma_{i} = {\rm{log}}{\left(\frac{{\rm{Pr}}(r_{i}|x_{i}=0)}{{\rm{Pr}}(r_{i}|x_{i}=1)}\right)}.
     \end{equation}
The} difficulty of solving the ML decoding problem \eqref{ML formulation} comes from non-convex parity-check constraints \eqref{ML formulation_b} and integer constraint \eqref{ML formulation_c}. To address these challenges, we introduce auxiliary variables to re-formulate the ML problem in the following.

    First, we consider the three-variables {parity-check} equation
{\setlength\abovedisplayskip{3pt}
     \setlength\belowdisplayskip{3pt}
      \setlength\jot{0.5pt}
    \begin{equation}\label{check-3}
       \big[x_{1}+ x_{2} + x_{3}\big]_2 = 0,\ \ x_{i} \in \{0,1\},\ \ i \in \{1,2,3\},
    \end{equation}
    which} can be equivalent to
{\setlength\abovedisplayskip{2pt}
   \setlength\belowdisplayskip{2pt}
    \setlength\jot{1pt}
    \begin{equation}\label{degree-3 linear inequalities}
       \begin{split}
         & x_{1} \leq x_{2} + x_{3}, ~~x_{2} \leq x_{1} + x_{3},\\
         & x_{3} \leq x_{1} + x_{2}, \hspace{0.29cm} x_{1} + x_{2} + x_{3} \leq 2 ,  \\
         & x_{1},x_{2},x_{3} \in \{0,1\},
       \end{split}
    \end{equation}
    in} the sense that \eqref{check-3} and \eqref{degree-3 linear inequalities} have the same solutions.
    Define
{\setlength\abovedisplayskip{2pt}
   \setlength\belowdisplayskip{2pt}
    \setlength\jot{1pt}
   \begin{equation}\label{t F matrix}
      \begin{split}
         \mathbf{w}=\begin{bmatrix}\ 0\ \\ \ 0\ \\ \ 0\ \\ \ 2\ \end{bmatrix}, \ \
         \mathbf{T} = \begin{bmatrix}
                           ~~1  &-1  & -1 \\
                          -1  &~~1 & -1 \\
                          -1  &-1  &~~1 \\
                          ~~1  &~~1 &~~1
                       \end{bmatrix}.
      \end{split}
   \end{equation}
   Then,} \eqref{degree-3 linear inequalities} can be further expressed
{\setlength\abovedisplayskip{5pt}
     \setlength\belowdisplayskip{5pt}
      \setlength\jot{1pt}
   \begin{equation}\label{check-3_2}
       \mathbf{Tx} \preceq \mathbf{w},\ \mathbf{x}\in\{0,1\}^3,
   \end{equation}
   where} $\mathbf{x}=[x_1\ x_2\ x_3]^T$. This idea can be applied to the general parity-check equation with more than three variables and  construct its equivalent expression.
   Specifically, we consider the $jth$ parity-check equation in \eqref{ML formulation_b}. Without loss of generality, we assume it involves $d_i$ variables and they are noted by $x_{\sigma_1},\dotsb, x_{\sigma_{d_i}}$, i.e., their corresponding $H_{j,i}$ are 1 and others are 0. Introduce an auxiliary variable for every two variables and let them satisfy the three-variables parity-check equation. If the number of auxiliary variables is larger than three, we introduce extra auxiliary variables for the first introduced auxiliary variables and formulate new three-variables parity-check equations. Continuing with this procedure, we can decompose any multi-variables parity-check equation into finite three-variables ones.
   Moreover, it can be shown that the numbers of the introduced auxiliary variables and the corresponding three-variables parity-check equations are calculated by $d_j-3$ and $d_j-2$ respectively.
   Then, the total numbers of the three-variables parity-check equations and the introduced auxiliary variables corresponding to the constraints \eqref{ML formulation_b} are
{\setlength\abovedisplayskip{6pt}
\setlength\belowdisplayskip{6pt}
\setlength\jot{2pt}
\begin{equation}\label{gamma_a_c}
  \begin{split}
   &\Gamma_{c} = \sum_{j=1}^{m}(d_j-2),  ~~ \Gamma_{a} = \sum_{j=1}^{m}(d_j-3), 
  \end{split}
\end{equation} respectively.}

Define
$\mathbf{v} =[\mathbf{x}^{T}, \mathbf{u}]^{T}$, where auxiliary variable $\mathbf{u} \in\{0,1\}^{\Gamma_a}$. Let $\mathbf{Q}_{\tau}\in\{0,1\}^{3\times(n+\Gamma_a)}$ be a variable-selecting matrix corresponding to the variables in the $\tau th$ three-variables parity-check equation. Then, it can be expressed as
$\mathbf{T} \mathbf{Q}_{\tau}\mathbf{v}\preceq {\mathbf{t}}$. Moreover, we define
    \begin{subequations}\label{Abq}
     \begin{align}
     &\mathbf{q}=[\pmb{\gamma}^{T}, \mathbf{0}]^{T},  \label{Abq-a}\\
     &\mathbf{A}=[\mathbf{T} \mathbf{Q}_1;\cdots;\mathbf{T}  \mathbf{Q}_\tau;\cdots;\mathbf{T} \mathbf{Q}_{\Gamma_c}], \label{Abq-b}\\
     &\mathbf{b}=\mathbf{1} \otimes {\mathbf{w}},                \label{Abq-c}
    \end{align}
   \end{subequations}
where symbols ``$\mathbf{1}$'' and ``$\mathbf{0}$'' are $\Gamma_c$-length all-ones vector and $\Gamma_a$-length all-zeros vector respectively. Then, the ML decoding problem \eqref{ML formulation} is equivalent to the following linear integer program\footnote{Here, we just give a brief review of transforming the ML decoding problem to a linear integer program. More detailed description can be found in \cite{check-node-decompetition1} and \cite{our-admm-lp}.}
{\setlength\abovedisplayskip{5pt}
   \setlength\belowdisplayskip{5pt}
    \setlength\jot{3pt}
     \begin{subequations}\label{C-LP model}
      \begin{align}
       &\hspace{0.0cm} \underset{\mathbf{v}}\min \hspace{0.25cm} \mathbf{q}^{T}\mathbf{v} \label{C-LP model a} \\
       &\ \textrm{s.t.} \hspace{0.3cm} \mathbf{Av} \preceq \mathbf{b}, \label{C-LP model b} \\
       &\hspace{0.90cm} \mathbf{v} \in \{0,1\}^{n+\Gamma_a}. \label{C-LP model c}
      \end{align}
    \end{subequations}

{\it Remarks:}} {
   In \cite{check-node-decompetition1}, authors relaxed integer constraint \eqref{C-LP model c} into a box constraint, and then customized the Lagrangian dual method to solve the resulting linear relaxation model. In \cite{our-admm-lp}, authors exploited matrix $\mathbf{A}$'s structures\footnote{Since $\mathbf{T}$ is a column-wise orthogonal matrix, $\mathbf{A}$ is also orthogonal in columns. Moreover, it can be found that elements in $\mathbf{A}$ are either $0$, $1$ or $-$1.}, such as column-wise orthogonality, integer elements and sparsity, and developed an efficient ADMM decoding algorithm to solve the LP relaxation. However, their error correction performance in low SNR regions for LDPC codes is worse than the classical BP decoder. In the following sections, we exploit the penalty idea in \cite{penalty-decoder} and develop a new LDPC decoder to improve error correction performance. Detailed analysis on ADMM iteration convergence, {\it all-zeros assumption}, and computational complexity, are also given. Moreover, we  show that any integral output can be verified easily whether or not it is an ML solution.}

\section{Non-convex QP formulation and its ADMM Solving Algorithm}\label{admm-qp-section}

\subsection{Non-convex QP formulation}
Deploy linear constraints $\mathbf{v} \in [0,1]^{n+\Gamma_a}$ to take the place of the binary integer constraint \eqref{C-LP model c} and then tighten the relaxation by adding the quadratic penalty function into objective \eqref{C-LP model a}, which leads to our proposed QP decoder
{\setlength\abovedisplayskip{3pt}
   \setlength\belowdisplayskip{3pt}
    \setlength\jot{1pt}
     \begin{subequations}\label{QP}
      \begin{align}
       &\hspace{0.0cm} \underset{\mathbf{v}}\min \hspace{0.2cm} \mathbf{q}^{T}\mathbf{v} - \frac{\alpha}{2}\|\mathbf{v}-0.5\|_2^2 \label{QP a}\\
       &\hspace{0.1cm} \textrm{s.t.} \hspace{0.3cm} \mathbf{Av} \preceq \mathbf{b}, \label{QP b} \\
       &\hspace{0.85cm} \mathbf{v} \in [0, 1]^{n+\Gamma_a}. \label{QP c}
      \end{align}
      \end{subequations}

    Comparing} to \cite{check-node-decompetition1} and \cite{our-admm-lp}, one notable improvement of the formulated model \eqref{QP} is that the non-convex quadratic penalty function is added into the linear objective of \eqref{C-LP model}, which can make the integer solutions more favorable. Moreover, by exploiting the inherent structures in model \eqref{QP}, we propose an efficient solving algorithm based on the ADMM technique whose computational complexity per iteration is linear in terms of LDPC code length.  Finally, we prove that the proposed ADMM-based QP decoder satisfies the favorable property of the {\it all-zeros assumption}.

\vspace{-0.3cm}
\subsection{ADMM algorithm framework}\label{ADMM-penalized-decoding}

By adding auxiliary variable $\mathbf{z}\in\mathbb{R}^{4\Gamma_c}_{+}$ into the constraint \eqref{QP b}, the optimization problem \eqref{QP} is equivalent to
{\setlength\abovedisplayskip{2pt}
   \setlength\belowdisplayskip{2pt}
    \setlength\jot{2pt}
    \begin{subequations}\label{QP 2}
     \begin{align}
       &\hspace{0.0cm} \underset{\mathbf{v}}\min \hspace{0.2cm} \mathbf{q}^{T}\mathbf{v} - \frac{\alpha}{2}\|\mathbf{v}-0.5\|_2^2, \label{QP 2 a}\\
       &\hspace{0.1cm} \textrm{s.t.} \hspace{0.3cm} \mathbf{Av}+\mathbf{z} = \mathbf{b}, \label{QP 2 b} \\
       &\hspace{0.87cm} \mathbf{v} \in [0, 1]^{n+\Gamma_a}, \ \mathbf{z}\succeq\mathbf{0}, \label{QP 2 c}
     \end{align}
    \end{subequations}
whose} augmented Lagrangian function can be expressed as
     \begin{equation}\label{augmented-Lagrangian-penalty}
      \begin{split}
        \emph{L}_{\mu}(\mathbf{v},\mathbf{z},\mathbf{y}) & =  \mathbf{q}^{T}\mathbf{v}- \frac{\alpha}{2}\|\mathbf{v}-0.5\|_2^2 \\
        & +  \mathbf{y}^{T}(\mathbf{Av+z-b})+ \frac{\mu}{2} \| \mathbf{Av+z-b} \|_2^{2},
      \end{split}
      \end{equation}
where $\mathbf{y} \in \mathbb{R}^{4\Gamma_c}$ denotes the Lagrangian multiplier and $\mu$ is some positive penalty parameter.
Then, the framework of the proposed ADMM solving algorithm for problem \eqref{QP} can be described as
     \begin{subequations}\label{ADMM update_QP}
      \begin{align}
        & \mathbf{v}^{k+1} = \mathop{\arg \min}\limits_{\mathbf{v} \in [0, 1]^{n+\Gamma_a}} \emph{L}_{\mu}(\mathbf{v},\mathbf{z}^{k},\mathbf{y}^{k}), \label{d-update-ori-penalized}  \\
        & \mathbf{z}^{k+1} = \ \ \mathop{\arg \min}\limits_{\mathbf{z}\succeq \mathbf{0}} ~ \emph{L}_{\mu}(\mathbf{v}^{k+1},\mathbf{z},\mathbf{y}^{k}), \label{mu-update-ori-penalized} \\
        & \mathbf{y}^{k+1} = \mathbf{y}^{k} + \mu(\mathbf{Av}^{k+1}+\mathbf{z}^{k+1}-\mathbf{b}), \label{lamda-update-ori-penalized}
      \end{align}
     \end{subequations}
where $k$ denotes the iteration number.

Observing \eqref{ADMM update_QP}, we can see that the main computation cost lies in solving \eqref{d-update-ori-penalized} and \eqref{mu-update-ori-penalized}. In the following, we show that both  \eqref{d-update-ori-penalized} and \eqref{mu-update-ori-penalized} can be solved efficiently.

\subsubsection{Solving subproblem \eqref{d-update-ori-penalized}}\label{subp-1}
 Due to the fact that matrix $\mathbf{A}$ has the property of orthogonality in columns, $\mathbf{A}^T\mathbf{A}$ is a diagonal matrix, which means that variables in problem \eqref{d-update-ori-penalized} are separable. Therefore, solving problem \eqref{d-update-ori-penalized} is equivalent to solving the following $n+\Gamma_a$ subproblems in parallel
     \begin{subequations}\label{vsp}
     \begin{align}
      & \underset{v_i}\min \hspace{0.15cm} \frac{1}{2}\big(\mu e_i\!-\!\alpha\big)\!v_i^2 \!+\! \Big(q_i\!+\!\frac{1}{2}\alpha\!+\!\hat{\mathbf{a}}_i^T\big(\mathbf{y}\!+\!\mu(\mathbf{z}-\mathbf{b})\big)\Big)v_i,  \label{vsp a}\\
      &\hspace{0.1cm} \textrm{s.t.}\hspace{0.3cm} v_i\in [0,1], \label{vsp b}
    \end{align}
   \end{subequations}
where $\mathbf{e}={\rm diag}(\mathbf{A}^T\mathbf{A})=[e_{1},\cdots,e_{n+\Gamma_a}]^{T}$.
Let ${e}_{\rm min}$ denote the minimum value of the elements in vector $\mathbf{e}$. In order to guarantee that \eqref{vsp a} is strongly convex, we choose proper $\alpha$ and $\mu$ and let them satisfy $\mu e_{\rm min}>\alpha$. Then, the solution of problem \eqref{vsp} can be performed as follows: setting the gradient of objective \eqref{vsp a} to zero, then projecting the resulting solution to the interval $[0,1]$, i.e.,
     \begin{equation}\label{vi-update-2fanshu}
      v_{i}^{k+1}\!\!=\!\!\underset{[0,1]}\Pi\bigg(\frac{\hat{\mathbf{a}}_i^T\Big(\mathbf{b}\!-\! \mathbf{z}^k\!-\!\frac{\mathbf{y}^k}{\mu}\Big)\!-\!\phi_i}{\theta_i}\bigg),
     \end{equation}
where  $i\in\{1,\ldots,n\!+\!\Gamma_{a}\}$, $\phi_i=\frac{2q_i+\alpha}{2\mu}$, and $\theta_i = e_i-\frac{\alpha}{\mu}$. The projection $\underset{[0,1]}\Pi(\cdot)$ can be easily accomplished by the following rule: when the value of the element is larger than 1, set it to be 1; when the value of the element is less than 0, set it to be 0; otherwise, leave it unchanged.

\subsubsection{Solving subproblem \eqref{mu-update-ori-penalized}}
Observing subproblem \eqref{mu-update-ori-penalized}, we can see that the elements (variables) in $\mathbf{z}$ are also separable in both the objective and constraints. Thus, the optimal solution of problem \eqref{mu-update-ori-penalized} can be obtained by setting the gradient of the function $\emph{L}_{\mu}(\mathbf{v}^{k+1},\mathbf{z},\mathbf{y}^{k})$ to zero, then projecting the resulting solution of the corresponding linear equation to region $[0, +\infty)^{4\Gamma_c}$, i.e.,
\begin{equation}\label{z-update-dual}
     \mathbf{z}^{k+1} = \underset{[0,+\infty)^{4\Gamma_c}}\Pi\Big((\mathbf{b}-\mathbf{Av}^{k+1})-\frac{\mathbf{y}^{k}}{\mu}\Big).
     \end{equation}
More specifically, similar to computing $\mathbf{v}^{k+1}$, all of the elements in $\mathbf{z}^{k+1}$ can also be obtained in parallel by
 \begin{equation}\label{zi update}
      {z}_j^{k+1}\!=\!\underset{[0,+\infty)}\Pi \Big(b_j\!-\!\mathbf{a}_j^T\mathbf{v}^{k+1}\!-\!\frac{y_j^{k}}{\mu}\Big),
     \end{equation}
where $j\in\{1,\ldots 4\Gamma_c\}$ and $\underset{[0,+\infty)}\Pi(\cdot)$ denotes the projection operator onto the positive quadrant $[0,+\infty)$.

Observing the variable $\mathbf{y}$ in \eqref{vi-update-2fanshu} and \eqref{zi update}, one can find that if updating its scaled form $\frac{\mathbf{y}}{\mu}$, then the corresponding multiplications can be saved. Therefore, we let $\tilde{z}_j^{k+1} = b_j-\mathbf{a}_j^T\mathbf{v}^{k+1}-\frac{y_j^{k}}{\mu}$.
Then, each element in $\mathbf{y}$ can be obtained through
     \begin{equation}\label{y-update-simplify}
      \frac{y_{j}^{k+1}}{\mu} = \left\{
                \begin{array}{ll}
                  0, &  \textrm{if}~\tilde{z}_j^{k+1} \geq 0,\\
                  -\tilde{z}_j^{k+1}, &  \textrm{otherwise}.
                \end{array}
              \right.
    \end{equation}

To be clear, we summarize the proposed ADMM algorithm for solving decoding model \eqref{QP} in \emph{Algorithm \ref{ADMM-penalized-alg}}. In the following, we make some remarks on the proposed \emph{Algorithm \ref{ADMM-penalized-alg}}:
\begin{itemize}
  \item Due to the orthogonality of matrix $\mathbf{A}$, all of the variables in $\mathbf{v}$ and $\mathbf{z}$ have analytical expressions and can be updated in parallel. This favorable property comes from auxiliary variables introduced to the three-variables parity-equations.
  \item Compared with the state-of-the-art works \cite{Barman-ADMM,efficient-projection1,efficient-projection2,efficient-projection3,jiao-zhang,Horizontal-Layered-Scheduling,Node-Wise-Scheduling} and \cite{penalty-decoder,improve-penalized-irregular,piece-penalized,lp-box-decoder}, the proposed ADMM algorithm is free of the complex Euclidean projection on each check polytope. From \eqref{zi update}, it is easy to see that updating the variables in $\mathbf{z}$ only requires a simple Euclidean projection onto the positive quadrant.
  \end{itemize}

\begin{algorithm}[!t]
\caption{The proposed decoding algorithm}
\label{ADMM-penalized-alg}
\begin{algorithmic}[1]
\STATE Compute log-likelihood ratio $\pmb{\gamma}$ via \eqref{gamma}. Construct $\mathbf{q}$, $\mathbf{A}$, and $\mathbf{b}$ via \eqref{Abq-a}, \eqref{Abq-b}, and \eqref{Abq-c} respectively.
\STATE Initialize $\mathbf{y}^0$ and $\mathbf{z}^0$ as the all-zeros vectors. Set $k=0$. \\
\STATE  \textbf{Repeat}  \\
\STATE \hspace{0.2cm} S.1: Update $\{v_i^{k+1}|i=1,\ldots,n\!+\!\Gamma_{a}\}$ in parallel by \eqref{vi-update-2fanshu}.
\STATE \hspace{0.2cm} S.2: Update $\{{z}_j^{k+1}|j=1,\ldots, 4\Gamma_c\}$ in parallel by \eqref{zi update}.
\STATE \hspace{0.2cm} S.3: Update $\{\frac{y_{j}^{k+1}}{\mu}|j=1,\ldots, 4\Gamma_c\}$ in parallel by \eqref{y-update-simplify}.
\STATE\textbf{Until} $\parallel \mathbf{A}{\mathbf{v}^{k}}+{\mathbf{z}^{k}}-\mathbf{b} \parallel_{2}^{2}$ is less than some preset $\epsilon$.
\end{algorithmic}
\end{algorithm}

\vspace{-0.1cm}
\section{Performance Analysis}\label{Analysis-admm-qp-decoding}
In this section we make a detailed analysis of \emph{Algorithm \ref{ADMM-penalized-alg}} from the following three aspects: convergence analysis, decoding performance analysis, and computational complexity analysis.
\vspace{-0.3cm}
\subsection{Convergence property}\label{conver-analysis-alg1}
We have the following theorem to show the convergence property of the proposed ADMM algorithm. The detailed convergence proof is presented in Appendix \ref{convergence-proof}.

\begin{theorem}\label{conver-theorem}
  suppose $\mu e_{min} > \alpha$ holds. Let $\{\mathbf{v}^{k},\mathbf{z}^{k},\mathbf{y}^{k}\}$ be the sequence of iterations generated by the proposed \emph{Algorithm \ref{ADMM-penalized-alg}}. If sequence $\{\mathbf{v}^{k},\mathbf{z}^{k},\mathbf{y}^{k}\}$ converges, i.e., $\underset{k \rightarrow \infty}\lim \{\mathbf{v}^{k},\mathbf{z}^{k},\mathbf{y}^{k}\} = (\mathbf{v}^{*},\mathbf{z}^{*},\mathbf{y}^{*})$, then $(\mathbf{v}^{*},\mathbf{z}^{*})$ is some feasible point of  model \eqref{QP 2}. Moreover, $\mathbf{v}^{*}$ is a stationary point of problem \eqref{QP}, i.e,
     \begin{equation}\label{stationary}
       (\mathbf{v}-\mathbf{v}^{*})^{T}\nabla_{\mathbf{v}}f(\mathbf{v}^{*})\geq 0,~~\forall \mathbf{v}\in \mathcal{X},
     \end{equation}
where $f(\mathbf{v})=\mathbf{q}^{T}\mathbf{v}- \frac{\alpha}{2}\|\mathbf{v}-0.5\|_2^2$ and $\mathcal{X} = \{\mathbf{v}\mid \mathbf{Av} \preceq \mathbf{b}, \mathbf{v} \in [0,1]^{n+\Gamma_{a}}\}$.
\end{theorem}

{\it Remarks:} The above theorem states that the proposed Algorithm 1 could reach some stationary point of problem \eqref{QP} under the assumption of sequence convergence\footnotemark. Although sequence convergence is not theoretically proven so far, it is indeed observed in extensive simulations under different parameter settings (see Sec. V). Furthermore, the solution obtained by the proposed algorithm yields favorable decoding performance in comparison with several state-of-the-art approaches.

 \footnotetext{Actually, to prove convergence of ADMM algorithm for solving non-convex optimization problem is a very challenging research topic.
 Therefore, such an assumption is often used in establishing the convergence of nonconvex ADMM. Although there are a few works \cite{Mingyi-hong, wotao-yin} addressing the convergence issue of nonconvex ADMM, the proof techniques shown in these works are only suitable for problems with very special structures.}

\subsection{Decoding performance analysis}

\subsubsection{ML test}
 {in general, we cannot guarantee that an integral output of the proposed decoding \emph{Algorithm \ref{ADMM-penalized-alg}} is an ML solution because it is difficult to determine whether or not this integral output is a global minimizer of the non-convex QP problem \eqref{QP}.
 However, the proposed decoding \emph{Algorithm \ref{ADMM-penalized-alg}} has a similar merit of the decoding algorithm in \cite{penalty-decoder} where one can easily test whether its integral output is an ML solution.} Specifically, the test can be described as the following proposition.

\begin{proposition}\label{ML-test}
 if the output of the proposed \emph{Algorithm \ref{ADMM-penalized-alg}} is an integral solution, we plug the integral output back into \emph{Algorithm \ref{ADMM-penalized-alg}}, but set the penalty parameter $\alpha=0$, and just do one more iteration.
 If the value of objective function \eqref{QP a} doesn't decrease, then this integral solution is the ML solution.
\end{proposition}

\begin{proof}
 first, observing models \eqref{C-LP model} and \eqref{QP}, we can see that no integral feasible points are introduced when integer constraint \eqref{C-LP model c} is relaxed to linear constraint \eqref{QP c}. This means that any integral solution of model \eqref{QP} must be a valid codeword. Second, we need to verify whether this integral solution is an ML solution.
 By setting the penalty parameter $\alpha=0$, we can obtain a linear program. Then, we do one more iteration by applying \emph{Algorithm \ref{ADMM-penalized-alg}} by taking this integral solution as the initial point. If the value of objective function \eqref{QP a} doesn't decrease, i.e., the new output is equal to this integral solution, and we can say that the integral solution is the minimizer of the LP decoding problem since the LP problem is convex. By exploiting the ML certificate of LP decoding, the integral solution must be an ML codeword.
\end{proof}

\subsubsection{All-zeros assumption}
 our proposed QP decoder also satisfies the favorable property of the \emph{all-zeros assumption}, which is described as follows. The detailed proof for \emph{Theorem \ref{all zero assumption}} is presented in Appendix \ref{all-zero-assumption-proof}. Moreover, two lemmas and one corollary related to the proof are presented in Appendices \ref{proof-lemma1}, \ref{proof-lemma2}, and \ref{proof-corollary1} respectively.

\begin{theorem}\label{all zero assumption}
  assume that the noisy channel is symmetrical. Then, the probability that \emph{Algorithm \ref{ADMM-penalized-alg}} fails is independent of the transmitted codeword.
\end{theorem}

\subsection{Computational complexity analysis}\label{complexity-sec}

  Now, we focus on analyzing the complexity of the proposed \emph{Algorithm \ref{ADMM-penalized-alg}}.
  Since the entries in matrix $\mathbf{A}$ are either 0, 1, or $-$1, its corresponding matrix multiplications, such as $\mathbf{A}^T\mathbf{b}$ et al., can be performed via only addition operations.
  Observing \eqref{vi-update-2fanshu}, one can find that $\phi_i$ and $\theta_i$ can be calculated before we start the ADMM iterations. This means that we only need one multiplication to compute $v_i^{k+1}$.
  Thus, the computational complexity of $\mathbf{v}^{k+1}$ \eqref{vi-update-2fanshu} in each iteration is $\mathcal{O}(n+\Gamma_{a})$.
  The properties of matrix $\mathbf{A}$ can also be applied to computing $\mathbf{z}^{k+1}$, which results in calculating it via only addition operations. Observing \eqref{y-update-simplify}, we can find that it takes no addition and multiplication operations to compute $\frac{\mathbf{y}^{k+1}}{\mu}$ since $\tilde{\mathbf{z}}^{k+1}$ has already been obtained in \eqref{zi update}.
  From the above analysis, the total computational complexity of \emph{Algorithm \ref{ADMM-penalized-alg}} in each iteration is roughly $\mathcal{O}(n+\Gamma_{a})$.
  Moreover, observing \eqref{gamma_a_c}, one can find $\Gamma_a \leq m(d-3)=n(1-R)(d-3)$, where $R$ denotes the code rate and $d$ is the largest check node degree.
  This implies that $\Gamma_a$ is comparable to code length $n$ since $d\ll n$ in the case of the LDPC code.
  Hence, we conclude that the computational complexity of our proposed decoding algorithm in each iteration is linear in terms of LDPC code length.

  In Table \ref{Complexity-Comparison}, we compare computational complexity per iteration of \emph{Algorithm \ref{ADMM-penalized-alg}} with several state-of-the-art ADMM-based MP decoders.
  {Moreover, we also present the complexity of the conventional sum-product BP decoder \cite{channel-code} as a comparison.}
  From it, we have the following observations:
  first, the total complexity of the ADMM-based penalized decoding (PD) method \cite{penalty-decoder} per iteration is roughly $O(n+m2d\textrm{log}d)$. In it, ``two-slice'' (sorting) operations are involved.
  Second, the ADMM-based PD method with the check polytope projection (CPP) algorithm \cite{efficient-projection1} in each iteration has a complexity of roughly $O(n+md\textrm{log}d)$, where implementing projection operations onto check polytope needs one sorting operation.
  Third, the ADMM-based PD method with the CPP algorithm \cite{efficient-projection2} in each iteration has an average complexity of $O(n+md\textrm{log}d)$ since a partial sorting operation, which has an average linear complexity and a worst case quadratic complexity to the check degree, has to be involved.
  Fourth, the ADMM-based PD method with the CPP algorithm \cite{efficient-projection3} in each iteration has $O(n+mI_{\rm max})$ complexity per iteration because it involves an iterative projection procedure.
  Fifth, the $\ell_{p}$-box ADMM-based decoder \cite{lp-box-decoder} is expensive in comparison with state-of-the-art ADMM-based penalized decoders since it has to handle the extra $\ell_{p}$-sphere constraint.
  {In the end, as for the conventional sum-product BP decoder, it costs $d(2d-2)$ multiplications and $d(d-1)$ hyperbolic tangent operations for updating each check node and hence its per-iteration complexity is about $\mathcal{O}(md^2)$.}
  From the above analysis, it is clear that computational complexity per iteration of the proposed ADMM-QP decoder is cheaper than state-of-the-art MP decoders \cite{efficient-projection1}\cite{efficient-projection2}\cite{efficient-projection3}\cite{penalty-decoder}\cite{lp-box-decoder} and {is comparable to the conventional sum-product BP decoder in \cite{channel-code}.}

\begin{table}\small
\caption{Comparisons of Per-iteration Computational Complexity.}
\label{Complexity-Comparison}
\begin{center}
\begin{tabular}{|c|c|c|}
\hline
\hline
\textbf{Decoding algorithm}&  \textbf{Computational complexity}\\\hline
{Sum-product BP decoder \cite{channel-code}} & {$\mathcal{O}(md^2)$\footnotemark} \\\hline
ADMM PD \cite{penalty-decoder} & $\mathcal{O}(n+m2d \log d)$\\\hline
ADMM PD with CPP \cite{efficient-projection1} & $\mathcal{O}(n+md\log d)$\\\hline
ADMM PD with CPP \cite{efficient-projection2}  & $\mathcal{O}(n+md)$\\\hline
ADMM PD with CPP \cite{efficient-projection3}  & $\mathcal{O}(n+mI_{\rm max})$\footnotemark\\\hline
$\ell_{p}$-box ADMM decoder \cite{lp-box-decoder} & $\mathcal{O}(n+md\log d)$\\\hline
Proposed QP-ADMM decoder & $\mathcal{O}(n+m(d-3))$ \\\hline\hline
 \end{tabular}
 \end{center}
\end{table}
 \footnotetext[4]{Symbol $d$ denotes the largest check-node degree. {Since, in the case of the LPDC code, check node number $m$ scales linearly with code length $n$ and check-node degree $d$ is small, the computational complexity per iteration of the sum-product BP decoder is linear to code length $n$.}}

\footnotetext{$I_{\rm max}$ denotes the maximum number of iterations in the projection algorithm \cite{efficient-projection3}. Usually, $I_{\rm max}>d$ in the LDPC case. }

\section{Simulation results}\label{simulation-result}

  \begin{figure*}[tp]
\subfigure[$\mathcal{C}_{1}$ -- (2640,1320).]{
    \begin{minipage}{8.5cm}
    \centering
        \includegraphics[width=3.5in,height=2.7in]{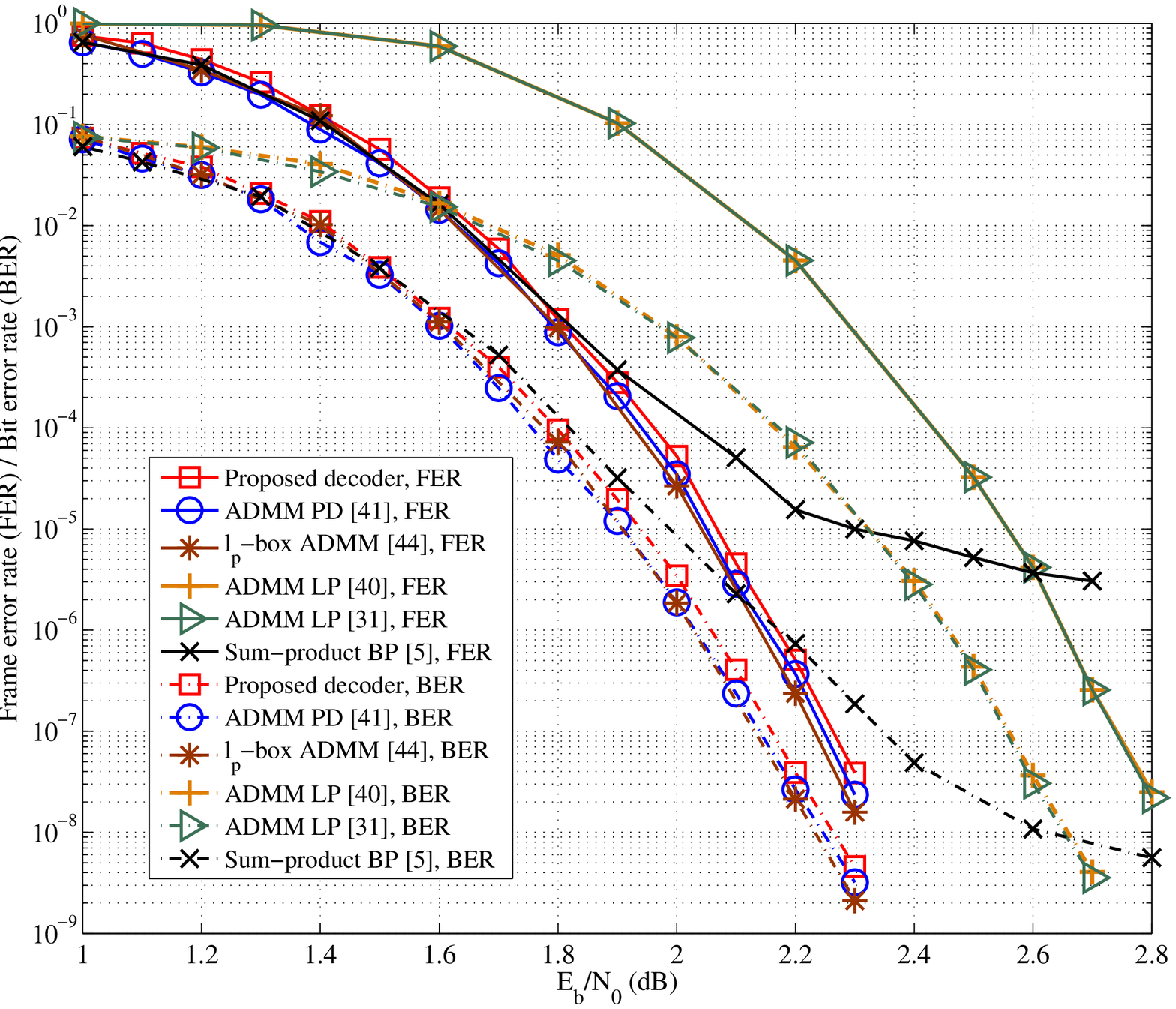}
            \label{fer2640}
    \end{minipage}%
    }
    \subfigure[$\mathcal{C}_{2}$ -- (13298, 3296).]{
    \begin{minipage}{8.5cm}
    \centering
        \includegraphics[width=3.5in,height=2.7in]{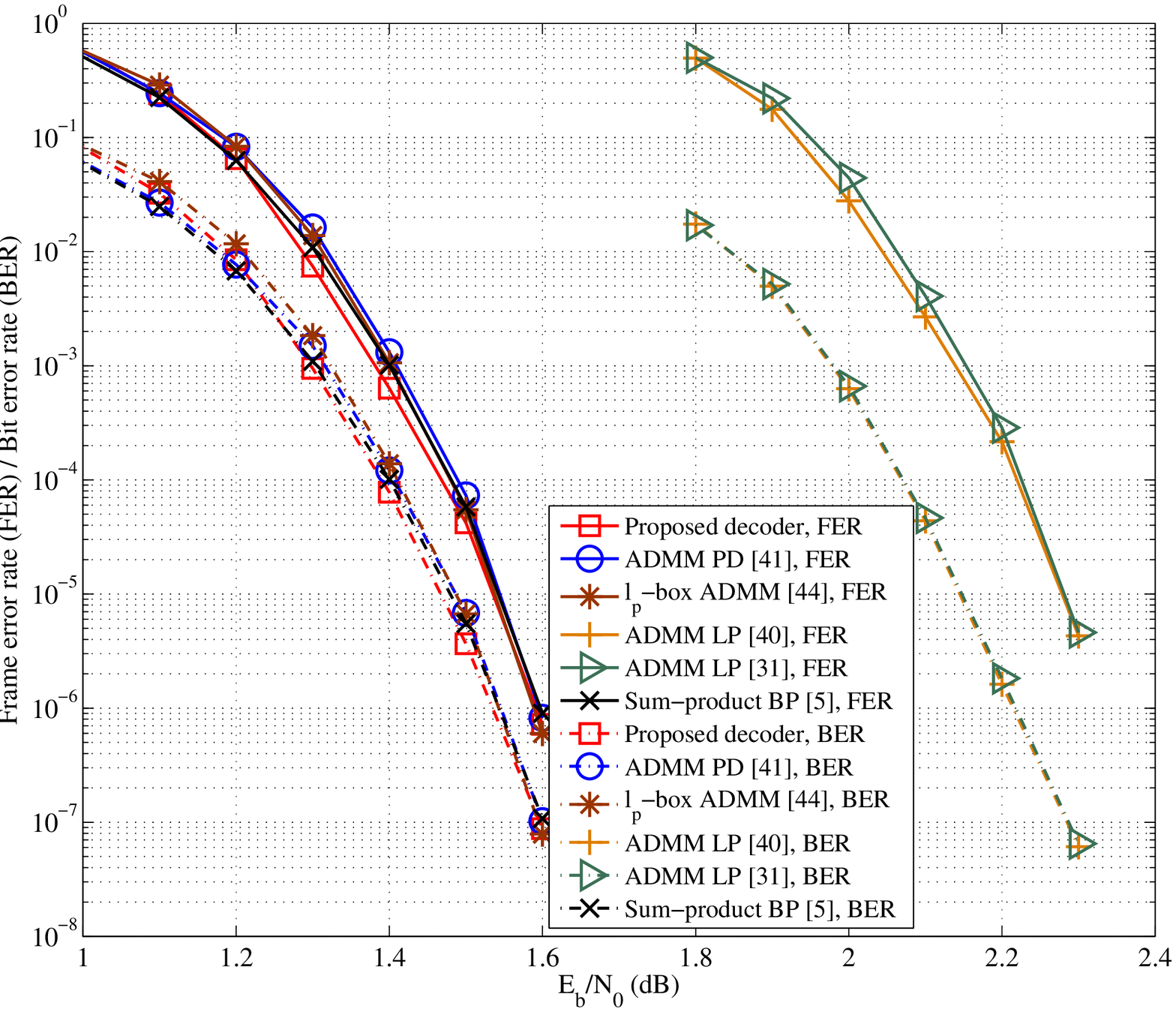}
            \label{fer13298}
    \end{minipage}%
    }\
\subfigure[$\mathcal{C}_{3}$ -- (999,888).]{
    \begin{minipage}{8.5cm}
    \centering
        \includegraphics[width=3.5in,height=2.7in]{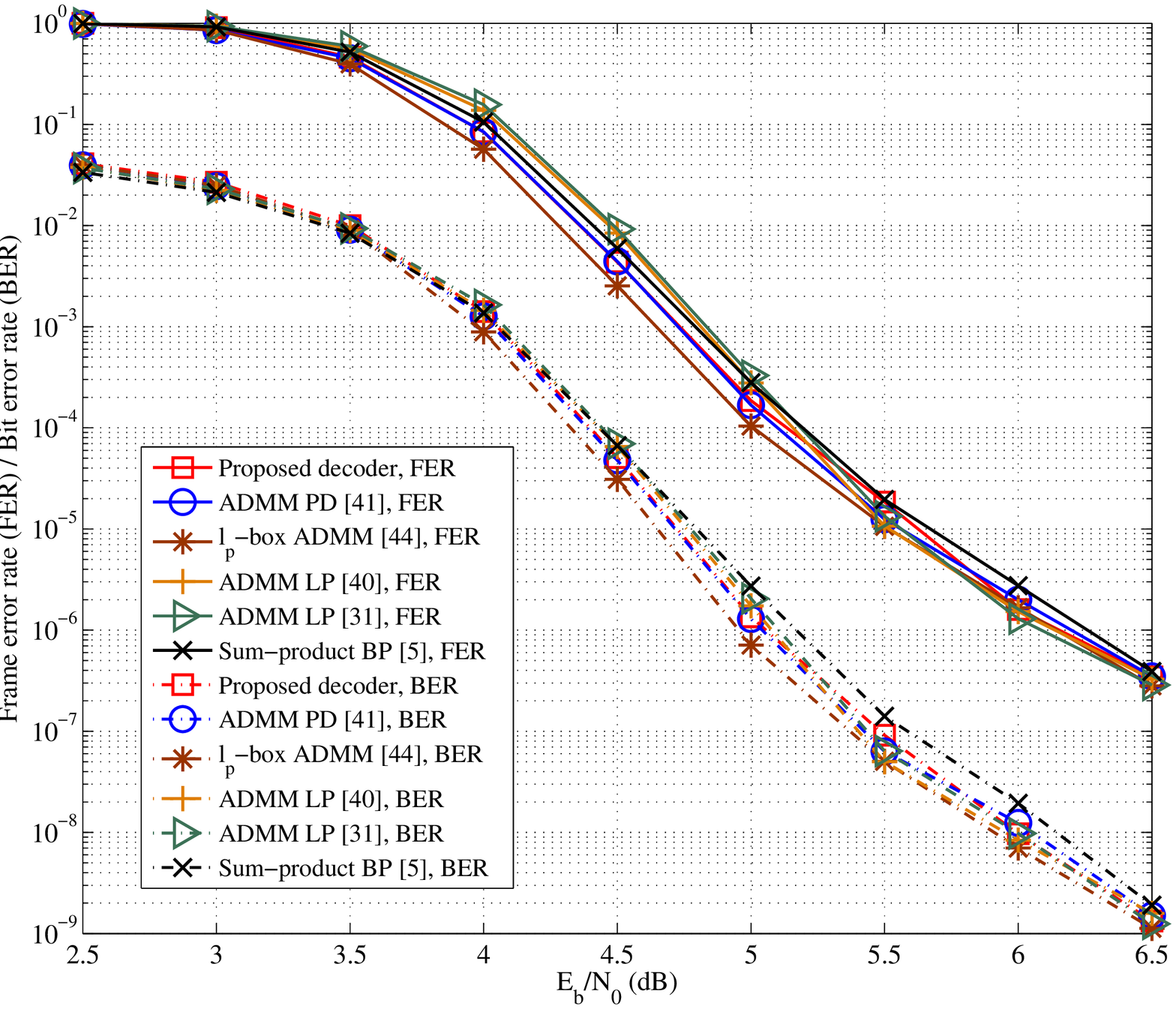}
            \label{fer999}
    \end{minipage}%
    }
    \centering
    \subfigure[$\mathcal{C}_{4}$ -- (576, 288).]{
    \begin{minipage}{8.5cm}
    \centering
        \includegraphics[width=3.5in,height=2.7in]{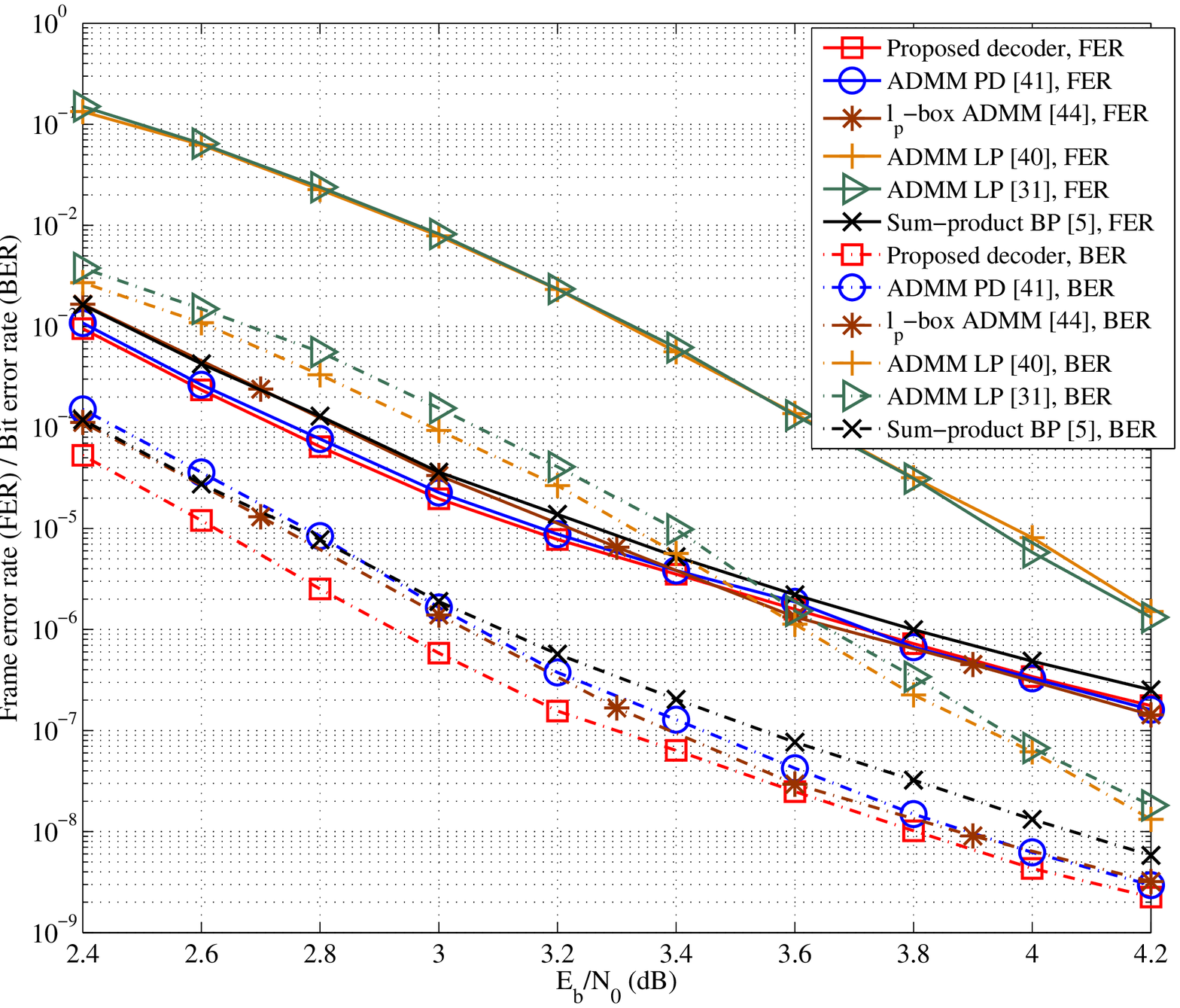}
            \label{fer576}
    \end{minipage}%
    }\
    \subfigure[$\mathcal{C}_{5}$ -- (1152,864).]{
    \begin{minipage}{8.5cm}
    \centering
        \includegraphics[width=3.5in,height=2.7in]{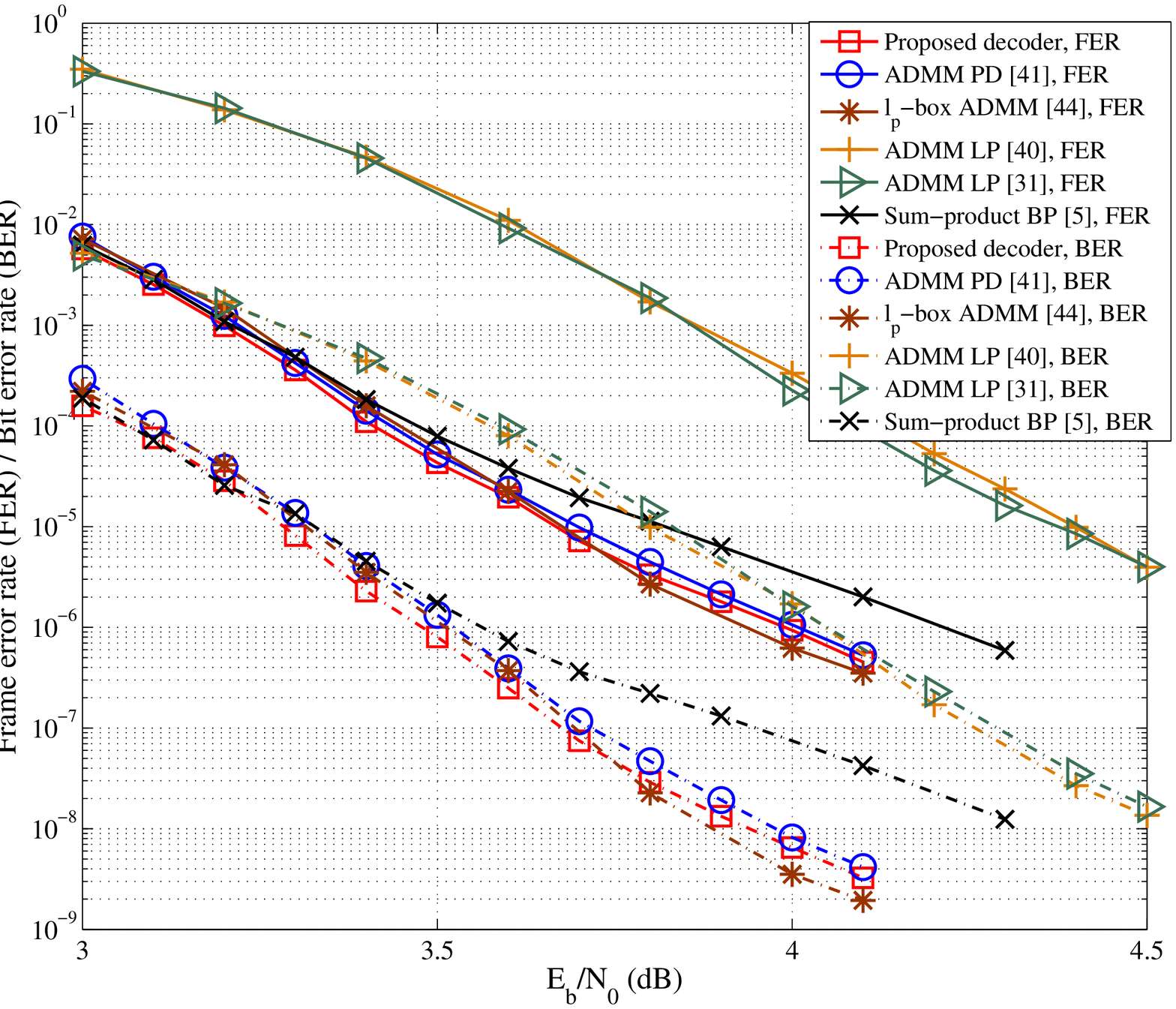}
            \label{fer1152}
    \end{minipage}%
    }\
    \subfigure[$\mathcal{C}_{6}$ -- (648,432).]{
    \begin{minipage}{8.5cm}
    \centering
        \includegraphics[width=3.5in,height=2.7in]{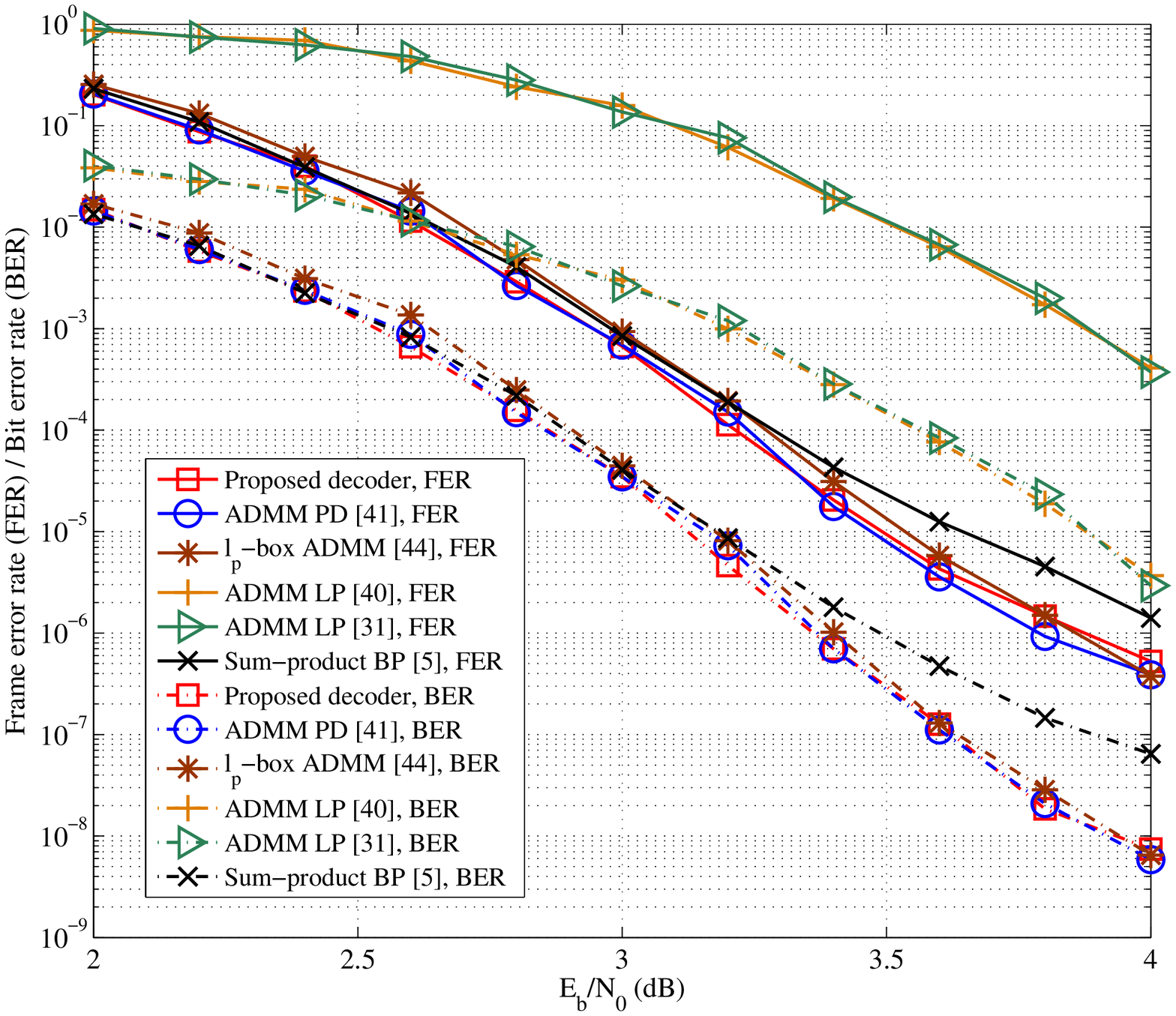}
            \label{fer1296}
    \end{minipage}%
    }
     \centering
    \caption{Comparisons of FER/BER performance for six LDPC codes from \cite{Mackay-code} and \cite{802.16-code} using different decoders, where $\mathcal{C}_{1}$ denotes (2640,1320) regular ``Margulis'' LDPC code, $\mathcal{C}_{2}$ and $\mathcal{C}_{3}$ denote MacKay (13298, 3296) irregular LDPC code and rate-0.89 MacKay (999,888) LDPC code, $\mathcal{C}_{4}$ and $\mathcal{C}_{5}$ are from IEEE 802.16e standard and denote  $(576, 288)$ irregular LDPC code and $(1152,864)$ irregular LDPC code respectively, {and $\mathcal{C}_{6}$ denotes $(648,432)$ regular LDPC code from IEEE 802.11n standard}.}
    \label{fer_ber}
 \end{figure*}

  In this section, several simulation results are presented to show the effectiveness of the proposed ADMM-based QP decoder.
  Specifically, in Section \ref{performance-simulation-result} we demonstrate the error correction performance and decoding efficiency of the proposed QP-ADMM decoder and also compare our decoder with the classical sum-product BP decoder and other ADMM-based MP decoders.
  In Section \ref{parameter-choices-alg1-simulation-result}, we focus on how to choose parameters $\mu$ and $\alpha$ to improve the error-correction performance and decoding efficiency of our proposed decoder.

\subsection{Performance comparison of the proposed QP-ADMM decoder}\label{performance-simulation-result}

  In this subsection, we present simulation results for six binary LDPC codes. The considered LDPC decoders are our proposed QP-ADMM decoder,  classical sum-product BP decoder  \cite{channel-code}, ADMM-based penalized decoder \cite{penalty-decoder}, another two ADMM-based penalized decoders with more efficient CPPs in \cite{efficient-projection1} and \cite{efficient-projection3} respectively, and $\ell_{p}$-box ADMM-based decoder \cite{lp-box-decoder}.
  The first LDPC code, named by $\mathcal{C}_{1}$, is (2640,1320), (3,6)-regular ``Margulis'' LDPC code.
  The other two codes, named by $\mathcal{C}_{2}$ and $\mathcal{C}_{3}$, are  MacKay (13298, 3296) irregular LDPC code and rate-0.89 MacKay (999,888) LDPC code.
  The above three codes $\mathcal{C}_{1}$, $\mathcal{C}_{2}$ and $\mathcal{C}_{3}$ are from \cite{Mackay-code}.
  Moreover, we consider another two LDPC codes $\mathcal{C}_{4}$ and $\mathcal{C}_{5}$ from the IEEE 802.16e standard \cite{802.16-code} and {one wifi code $\mathcal{C}_{6}$ from the IEEE 802.11n standard \cite{802.11n-code}}.
  They are $(576, 288)$ irregular LDPC code, $(1152,864)$ irregular LDPC code and {$(648,432)$ regular LDPC code} respectively.
  The simulation parameters are set as follows: parameter $\mu$ is set to be 1 for all of the six LDPC codes; penalty parameter $\alpha$ is chosen as 0.6, 0.8, {0.3}, 0.9, 0.4 and {0.6} for the six codes respectively; and we stop iterations when the residual $\|\mathbf{Av}^k+\mathbf{z}^k-\mathbf{b}\|_2^2<10^{-5}$ or the maximum number of iterations, 1000, is reached.
  The information bits are transmitted over the additive white Gaussian noise (AWGN) channel with binary phase shift keying (BPSK). Simulations are performed under the Microsoft visual C++ 6.0 environment and carried out on a computer whose specifications are i5-3470 3.2GHz CPU and 16 GB RAM.

 \begin{figure}[htbp]
   \subfigure[Average decoding time of $\mathcal{C}_{1}$.]{
    \begin{minipage}{9cm}
    \centering
        \includegraphics[width=3.5in,height=2.7in]{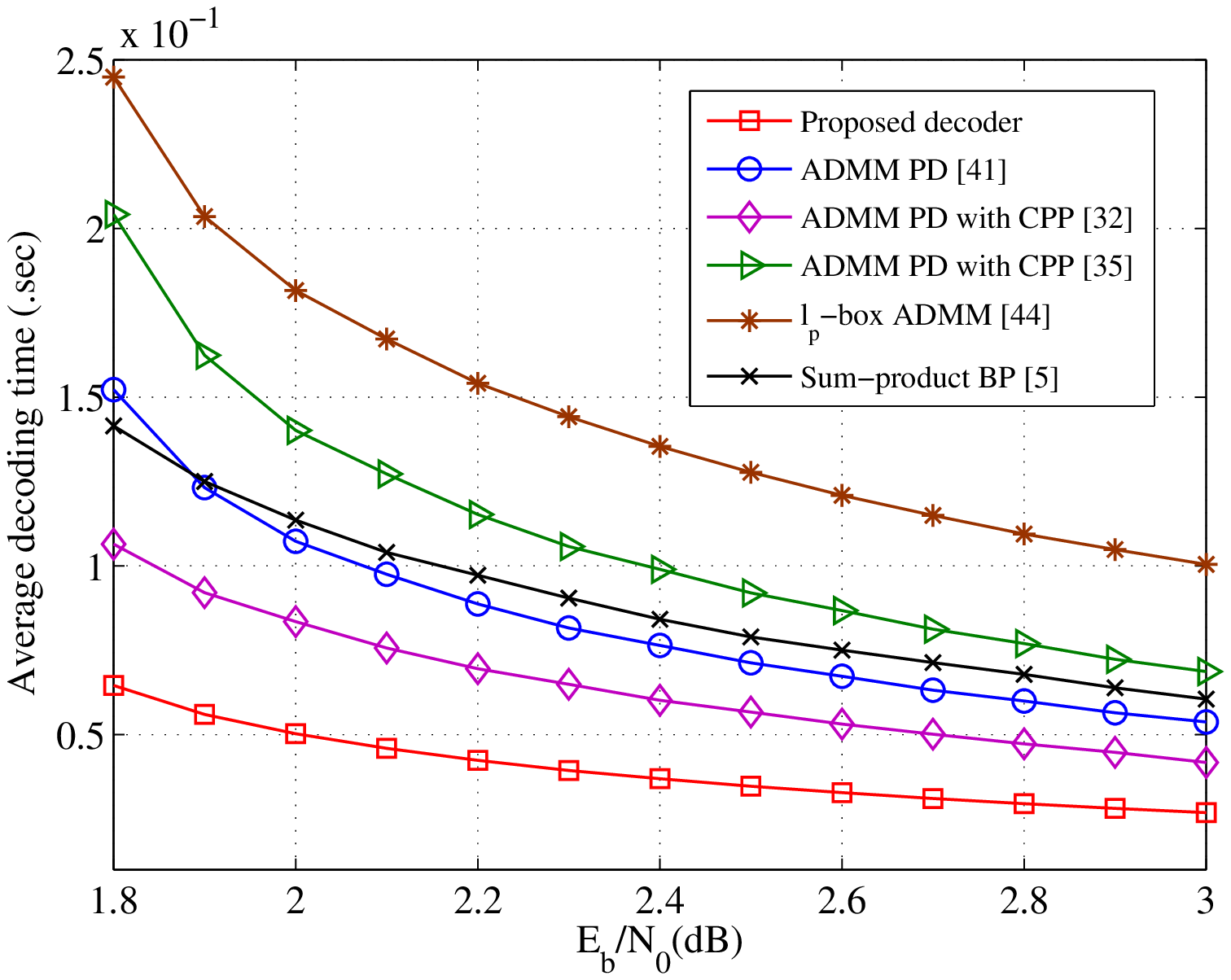}
            \label{time2640}
    \end{minipage}%
    }
    ~~~~~~~~~~~~~
    \subfigure[Convergence characteristic of $\mathcal{C}_{1}$.]{
    \begin{minipage}{9cm}
    \centering
        \includegraphics[width=3.5in,height=2.7in]{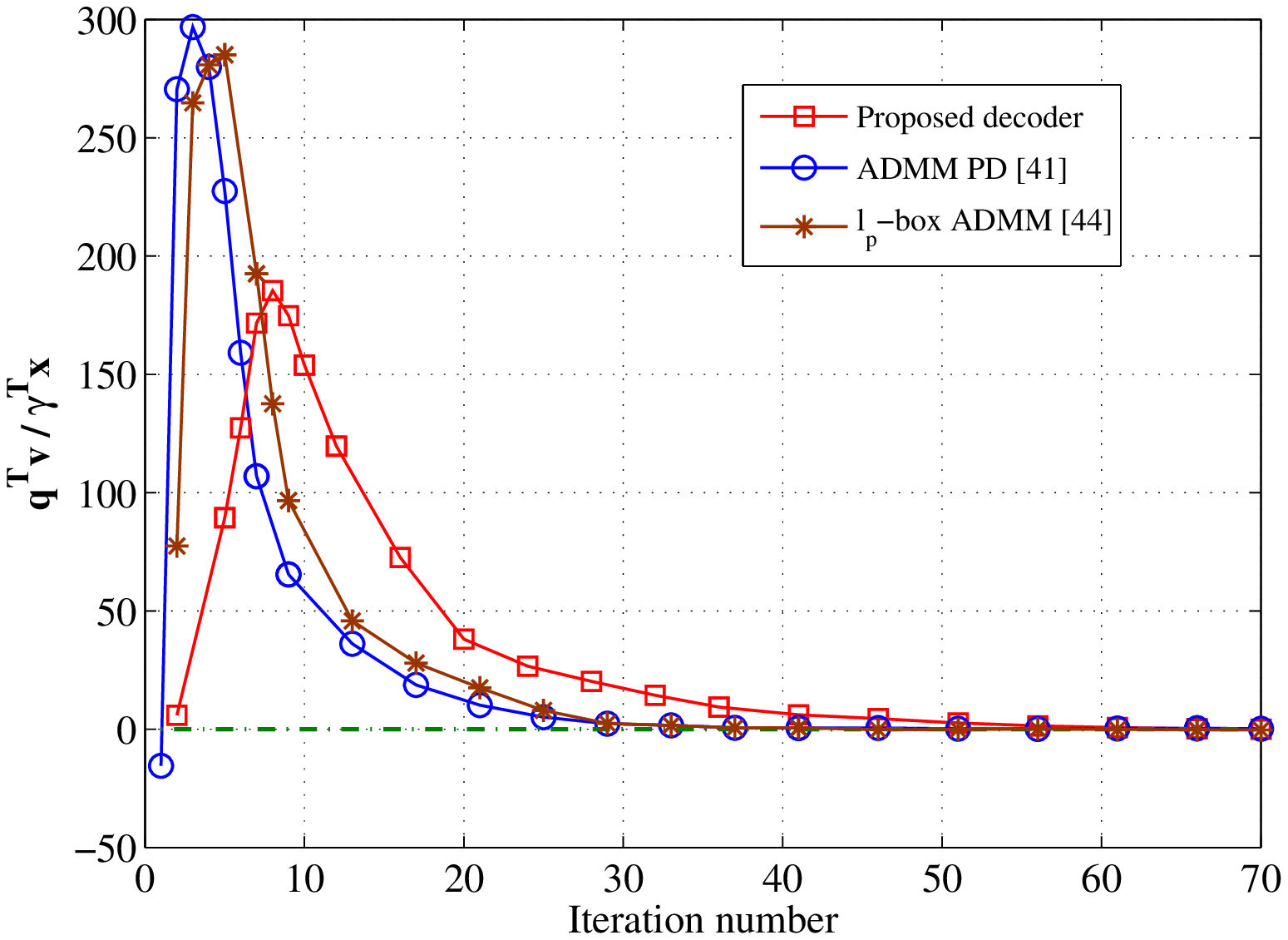}
            \label{conver2640}
    \end{minipage}%
    }
    ~~~~~~~~~~~~~
        \subfigure[Iteration number distribution of $\mathcal{C}_{1}$.]{
    \begin{minipage}{9cm}
    \centering
        \includegraphics[width=3.5in,height=2.7in]{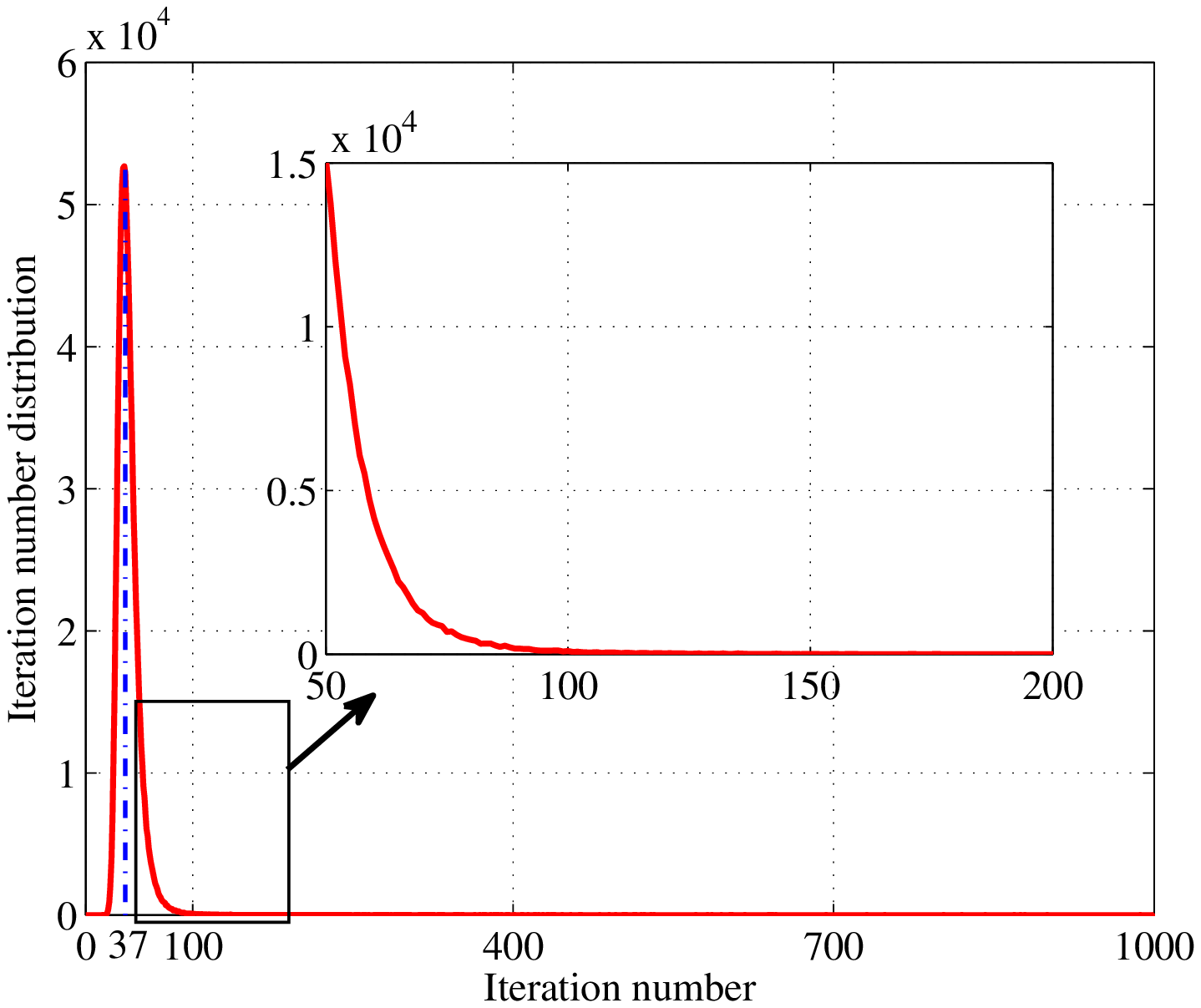}
            \label{ite2640}
    \end{minipage}%
    }\
 \centering
 \caption{ Average decoding time comparison, convergence characteristic comparison at $E_b/N_0$=2dB and iteration number distribution at $E_b/N_0$=2dB of $\mathcal{C}_{1}$, where $\mathcal{C}_{1}$ denotes the (2640,1320) regular ``Margulis'' LDPC code. All-zeros codeword is transmitted.}
 \label{time3}
 \end{figure}
  \begin{figure}[htbp]
 \subfigure[Average decoding time of $\mathcal{C}_{2}$.]{
    \begin{minipage}{9cm}
    \centering
        \includegraphics[width=3.5in,height=2.7in]{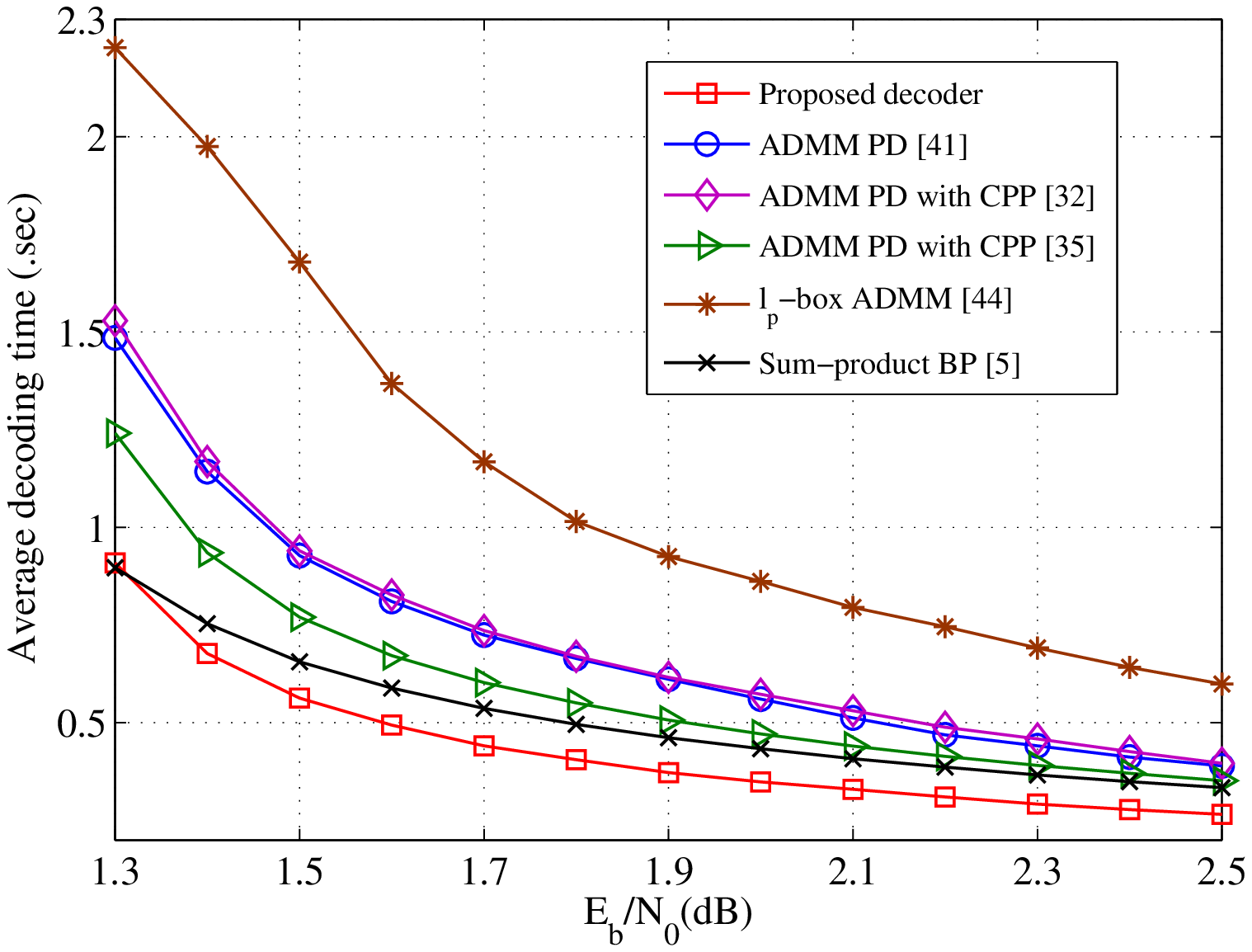}
            \label{time13298}
    \end{minipage}%
    }
    \subfigure[Convergence characteristic of $\mathcal{C}_{2}$.]{
    \begin{minipage}{9cm}
    \centering
        \includegraphics[width=3.5in,height=2.7in]{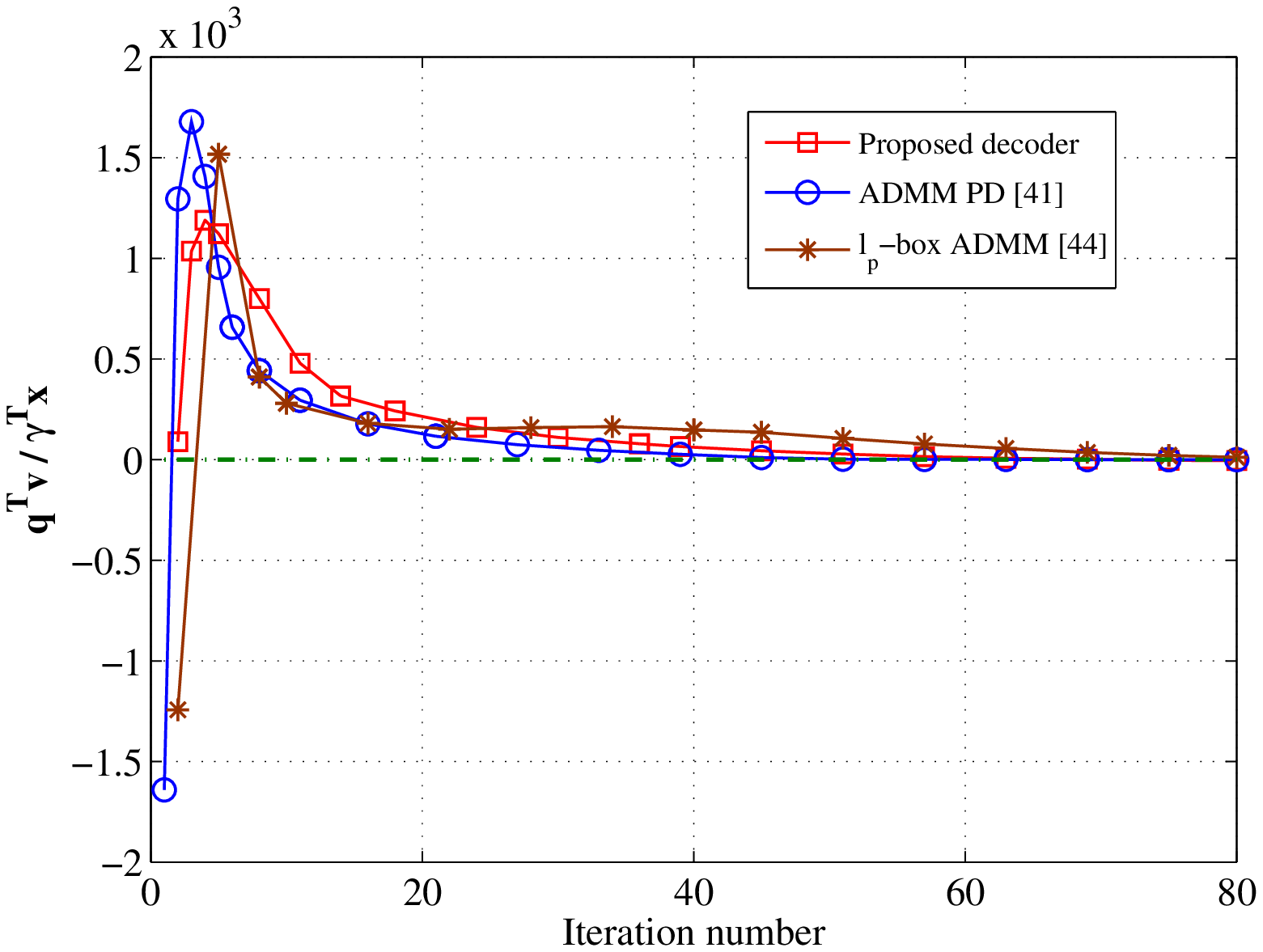}
            \label{conver13298}
    \end{minipage}%
    }
        \subfigure[Iteration number distribution of $\mathcal{C}_{2}$.]{
    \begin{minipage}{9cm}
    \centering
        \includegraphics[width=3.5in,height=2.7in]{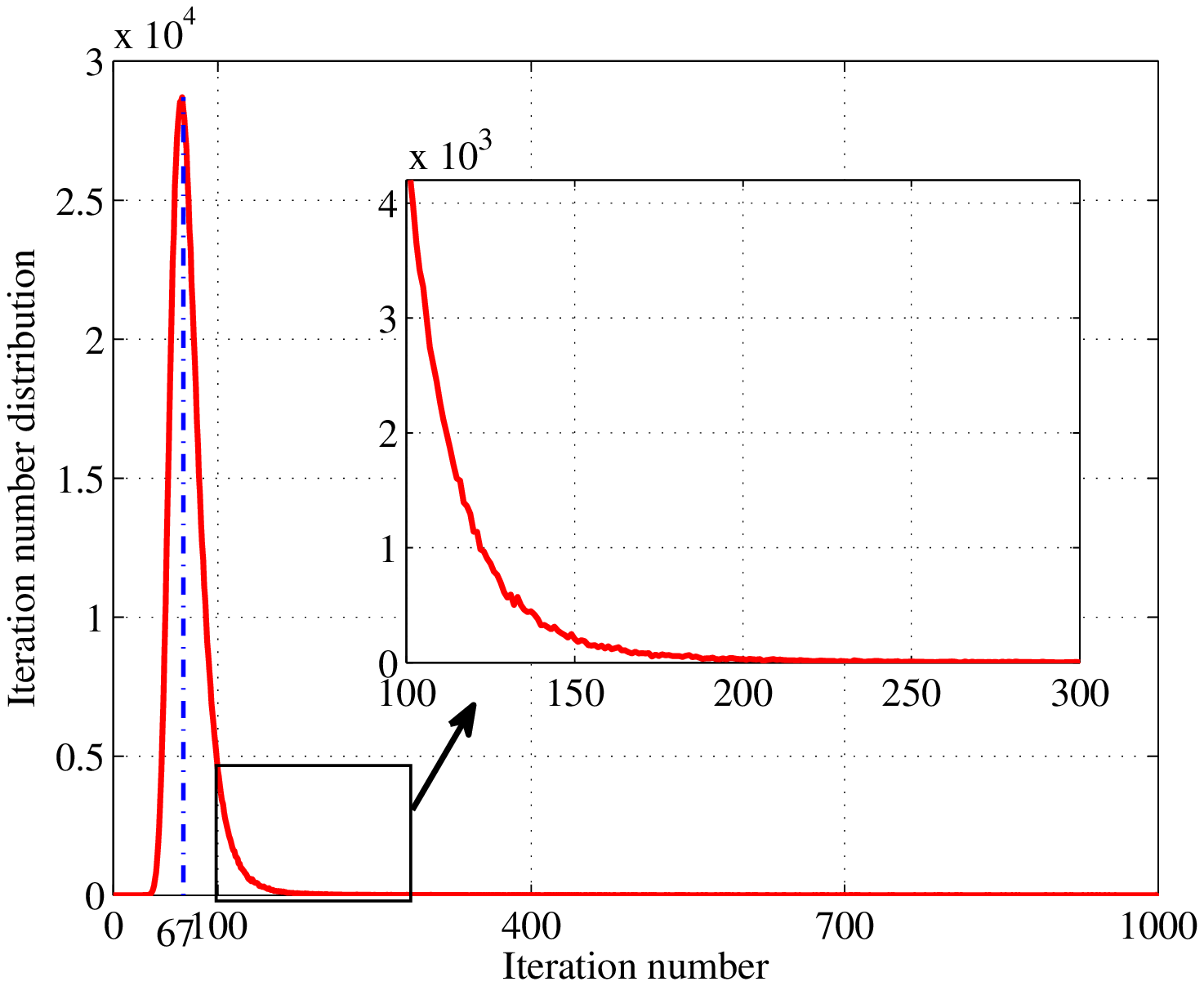}
            \label{ite13298}
    \end{minipage}%
    }\
 \centering
 \caption{ Average decoding time comparison, convergence characteristic comparison at $E_b/N_0$=1.5dB and iteration number distribution at $E_b/N_0$=1.5dB of $\mathcal{C}_{2}$, where $\mathcal{C}_{2}$ denotes the MacKay (13298, 3296) irregular LDPC code. All-zeros codeword is transmitted.}
 \label{time3}
 \end{figure}
 \begin{figure}[htbp]
   \subfigure[Average decoding time of $\mathcal{C}_{3}$.]{
    \begin{minipage}{9cm}
    \centering
        \includegraphics[width=3.5in,height=2.7in]{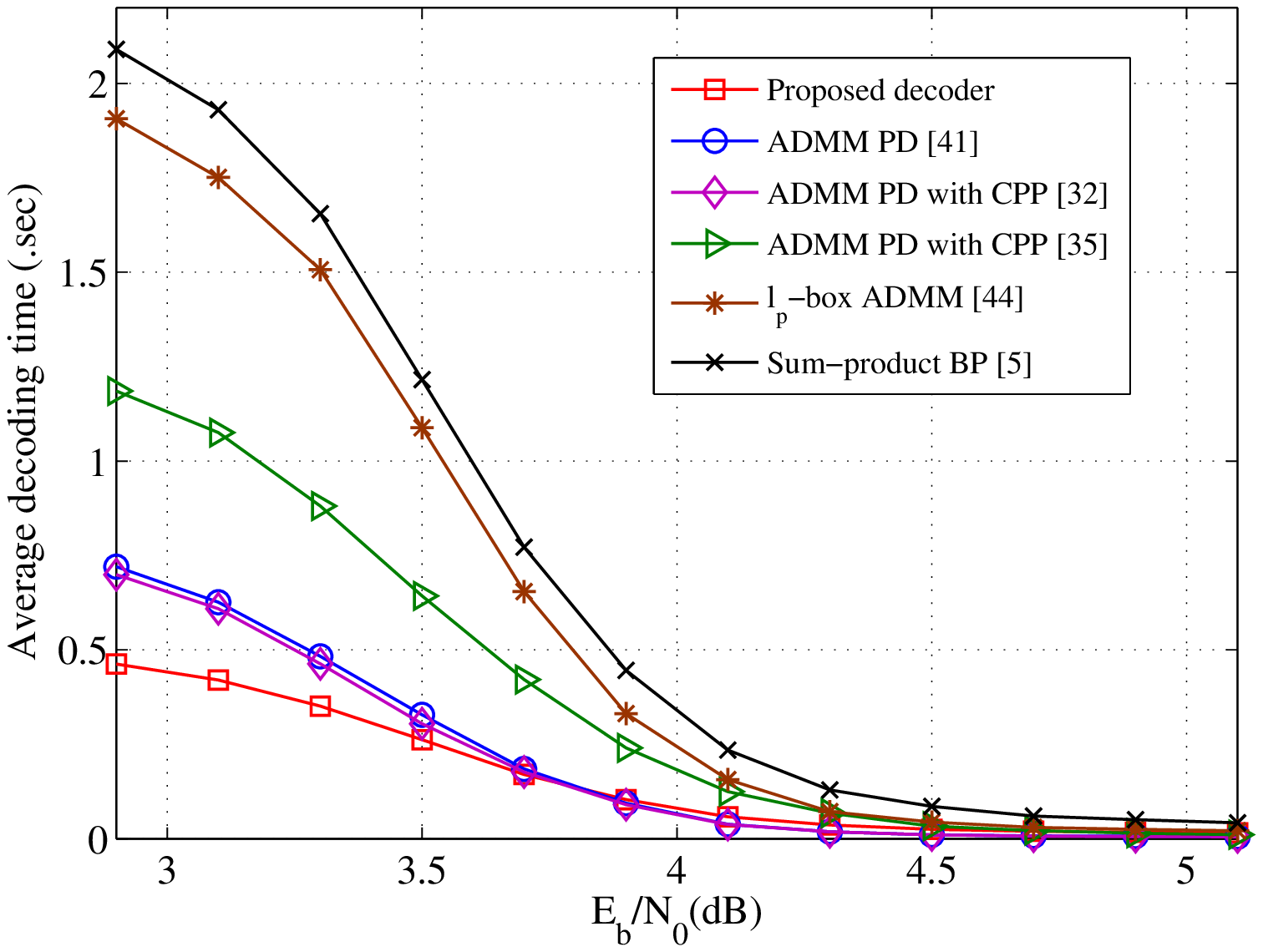}
            \label{time999}
    \end{minipage}%
    }
    \subfigure[Convergence characteristic of $\mathcal{C}_{3}$.]{
    \begin{minipage}{9cm}
    \centering
        \includegraphics[width=3.5in,height=2.7in]{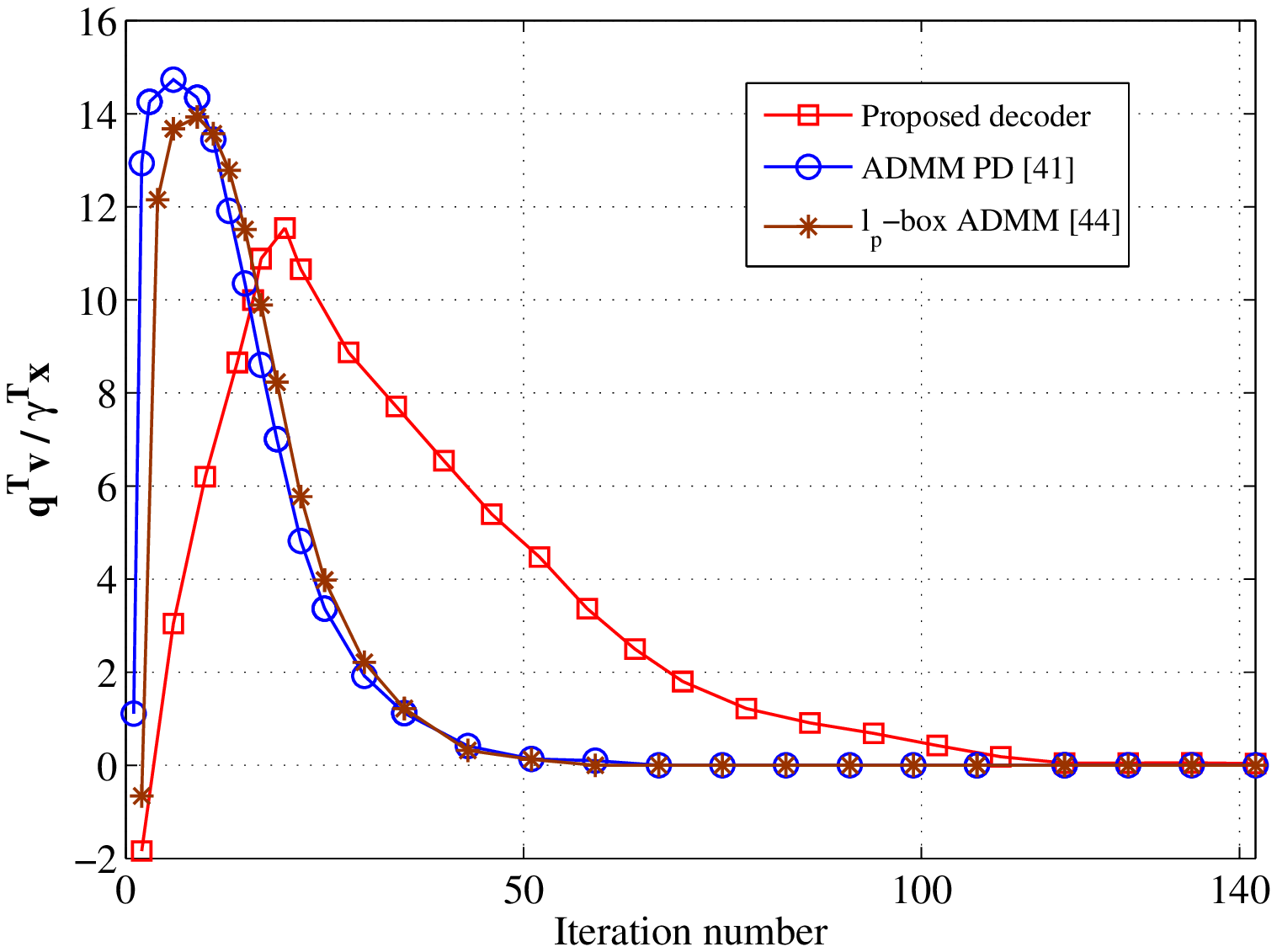}
            \label{conver999}
    \end{minipage}%
    }
        \subfigure[Iteration number distribution of $\mathcal{C}_{3}$.]{
    \begin{minipage}{9cm}
    \centering
        \includegraphics[width=3.5in,height=2.7in]{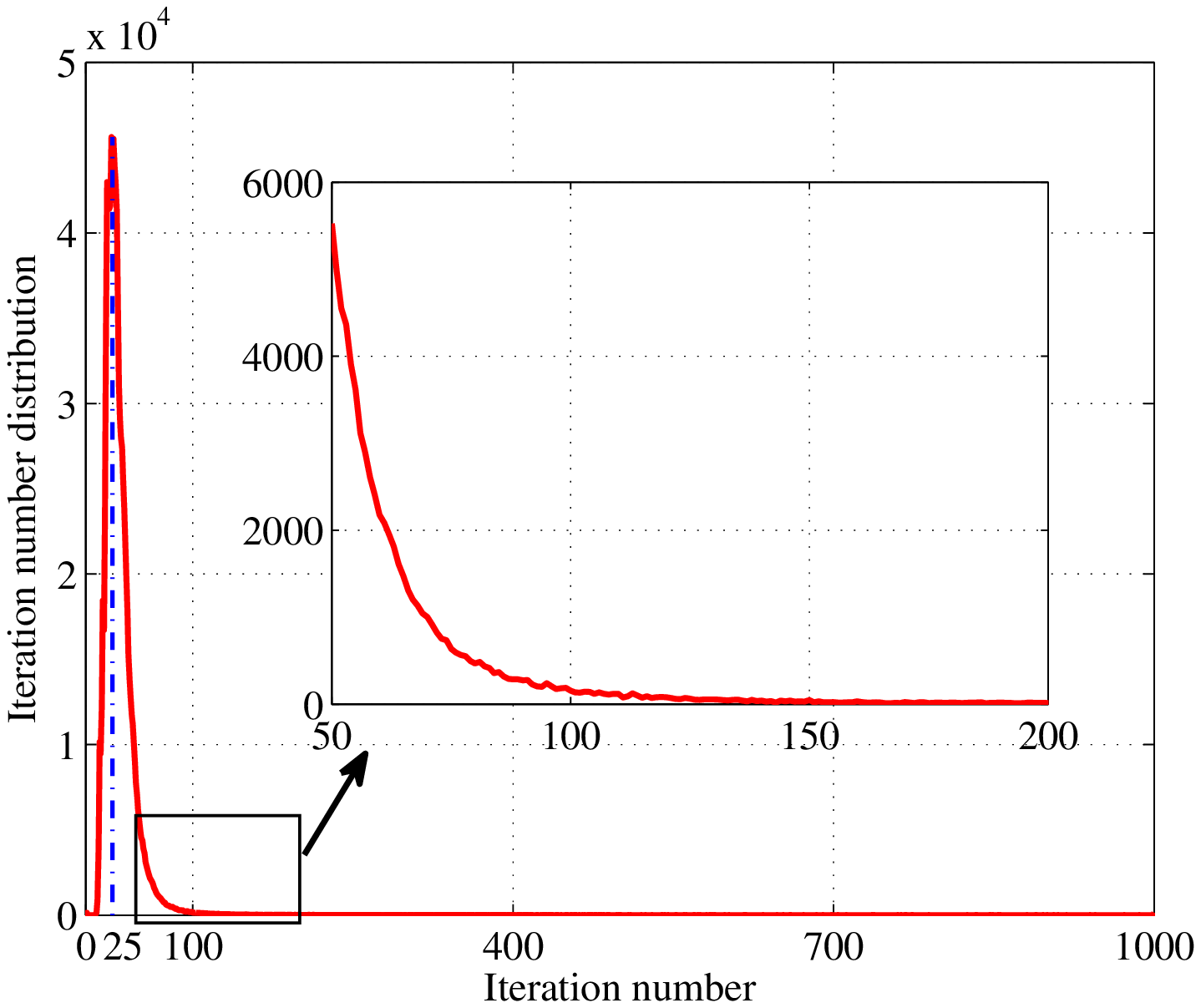}
            \label{ite999}
    \end{minipage}%
    }\
 \centering
 \caption{ Average decoding time comparison, convergence characteristic comparison at $E_b/N_0$=5dB and iteration number distribution at $E_b/N_0$=5dB of $\mathcal{C}_{3}$, where $\mathcal{C}_{3}$ denotes the rate-0.89 MacKay (999,888) LDPC code. All-zeros codeword is transmitted.}
 \label{time3}
 \end{figure}
 \begin{figure}[htbp]
  \subfigure[Average decoding time of $\mathcal{C}_{4}$.]{
    \begin{minipage}{9cm}
    \centering
        \includegraphics[width=3.5in,height=2.7in]{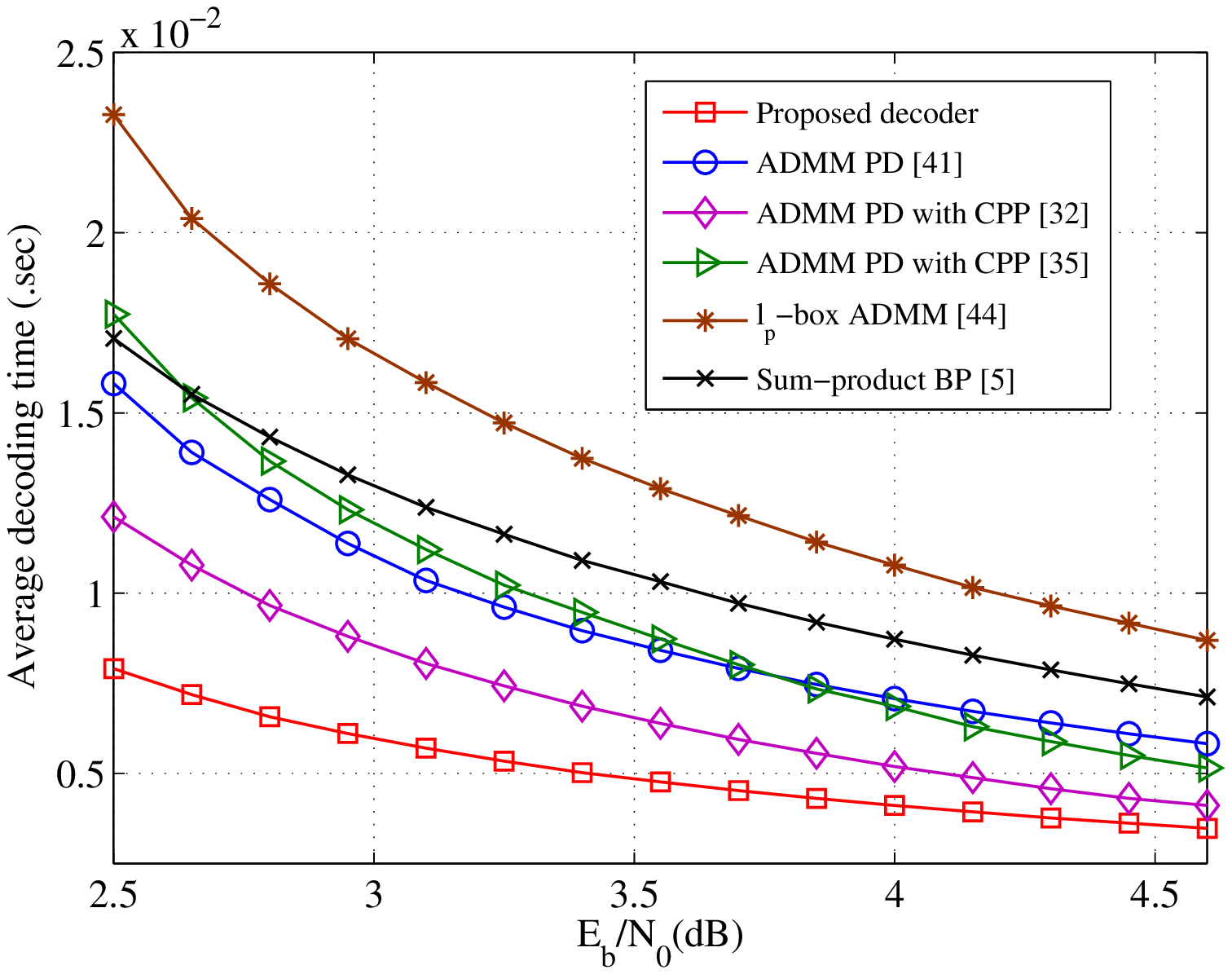}
            \label{time576}
    \end{minipage}%
    }
    \subfigure[Convergence characteristic of $\mathcal{C}_{4}$.]{
    \begin{minipage}{9cm}
    \centering
        \includegraphics[width=3.5in,height=2.7in]{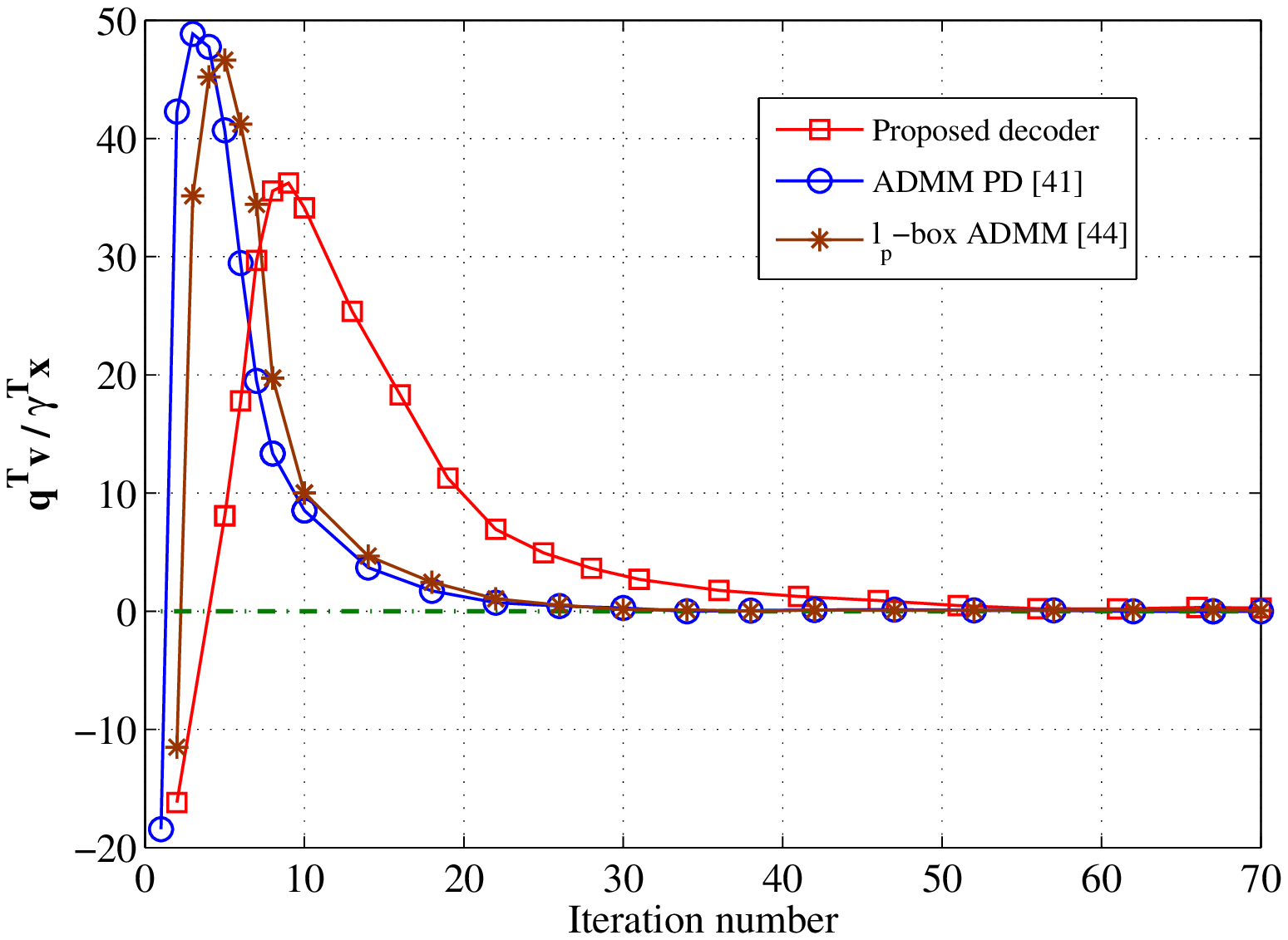}
            \label{conver576}
    \end{minipage}%
    }
        \subfigure[Iteration number distribution of $\mathcal{C}_{4}$.]{
    \begin{minipage}{9cm}
    \centering
        \includegraphics[width=3.5in,height=2.7in]{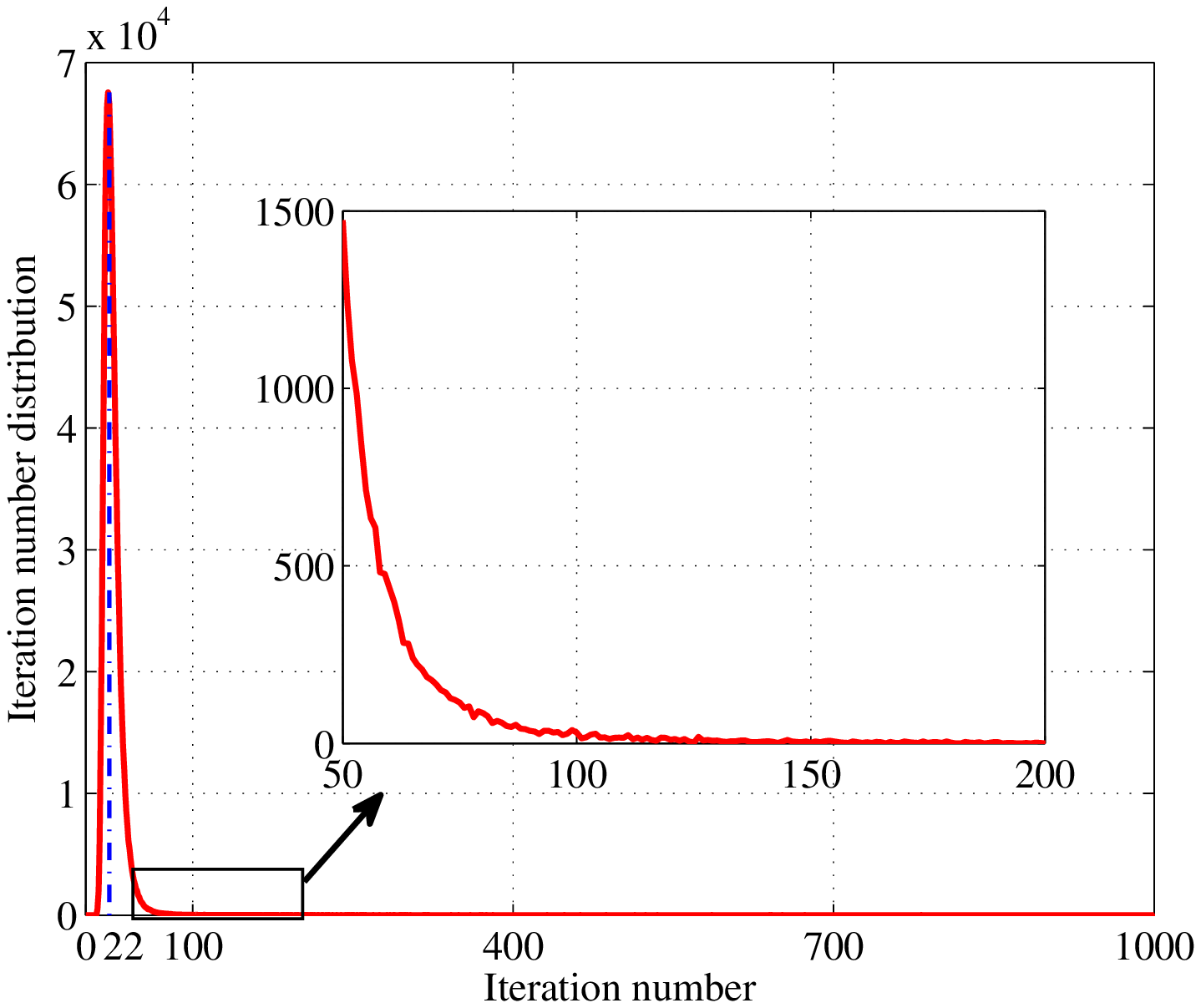}
            \label{ite576}
    \end{minipage}%
    }\
 \centering
 \caption{ Average decoding time comparison, convergence characteristic comparison at $E_b/N_0$=3dB and iteration number distribution at $E_b/N_0$=3dB of $\mathcal{C}_{4}$, where $\mathcal{C}_{4}$ denotes the (576,288) irregular LDPC code from IEEE 802.16e standard. All-zeros codeword is transmitted.}
 \label{time3}
 \end{figure}
 \begin{figure}[htbp]
  \subfigure[Average decoding time of $\mathcal{C}_{5}$.]{
    \begin{minipage}{9cm}
    \centering
        \includegraphics[width=3.5in,height=2.7in]{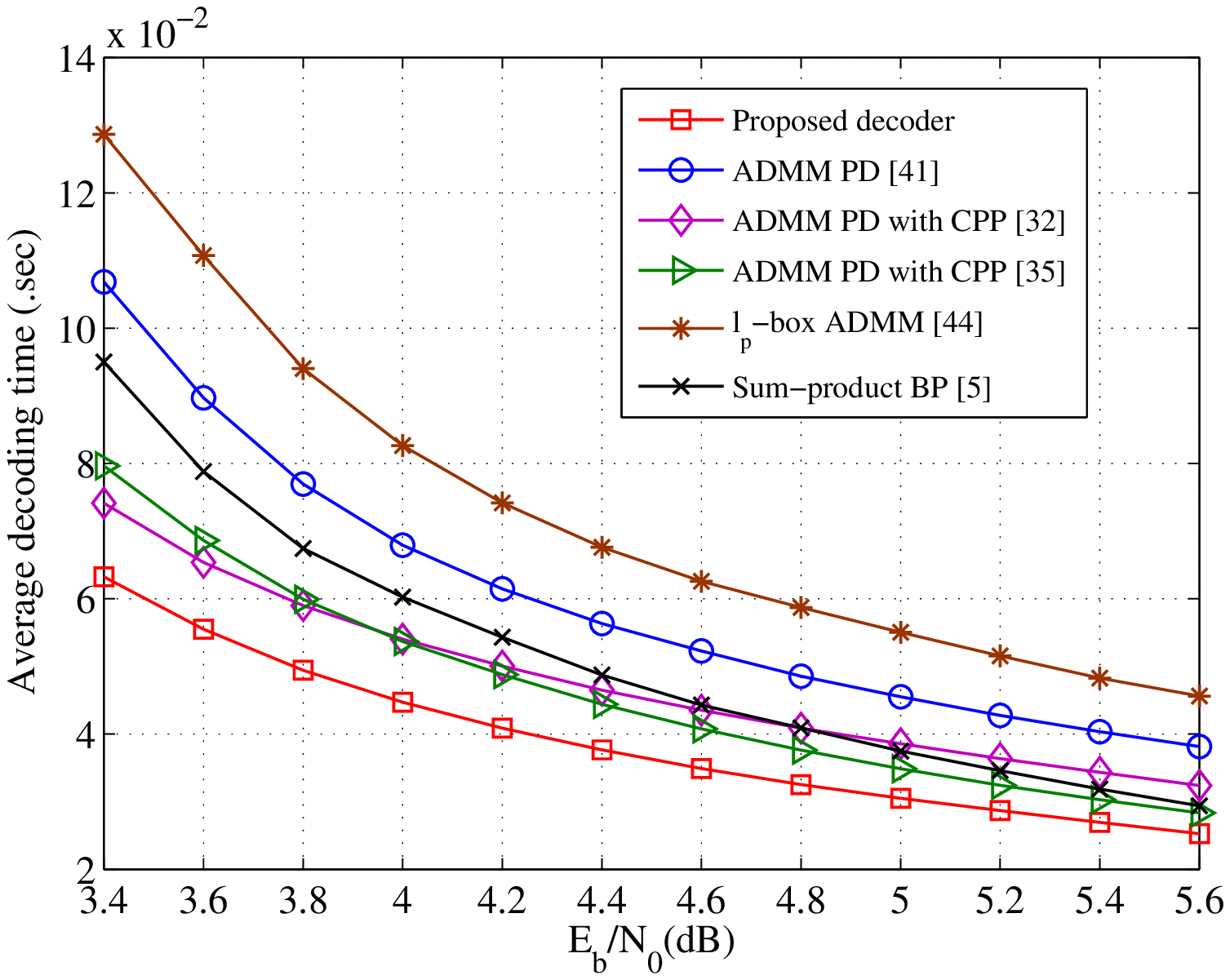}
            \label{time1152}
    \end{minipage}%
    }
    \subfigure[Convergence characteristic of $\mathcal{C}_{5}$.]{
    \begin{minipage}{9cm}
    \centering
        \includegraphics[width=3.5in,height=2.7in]{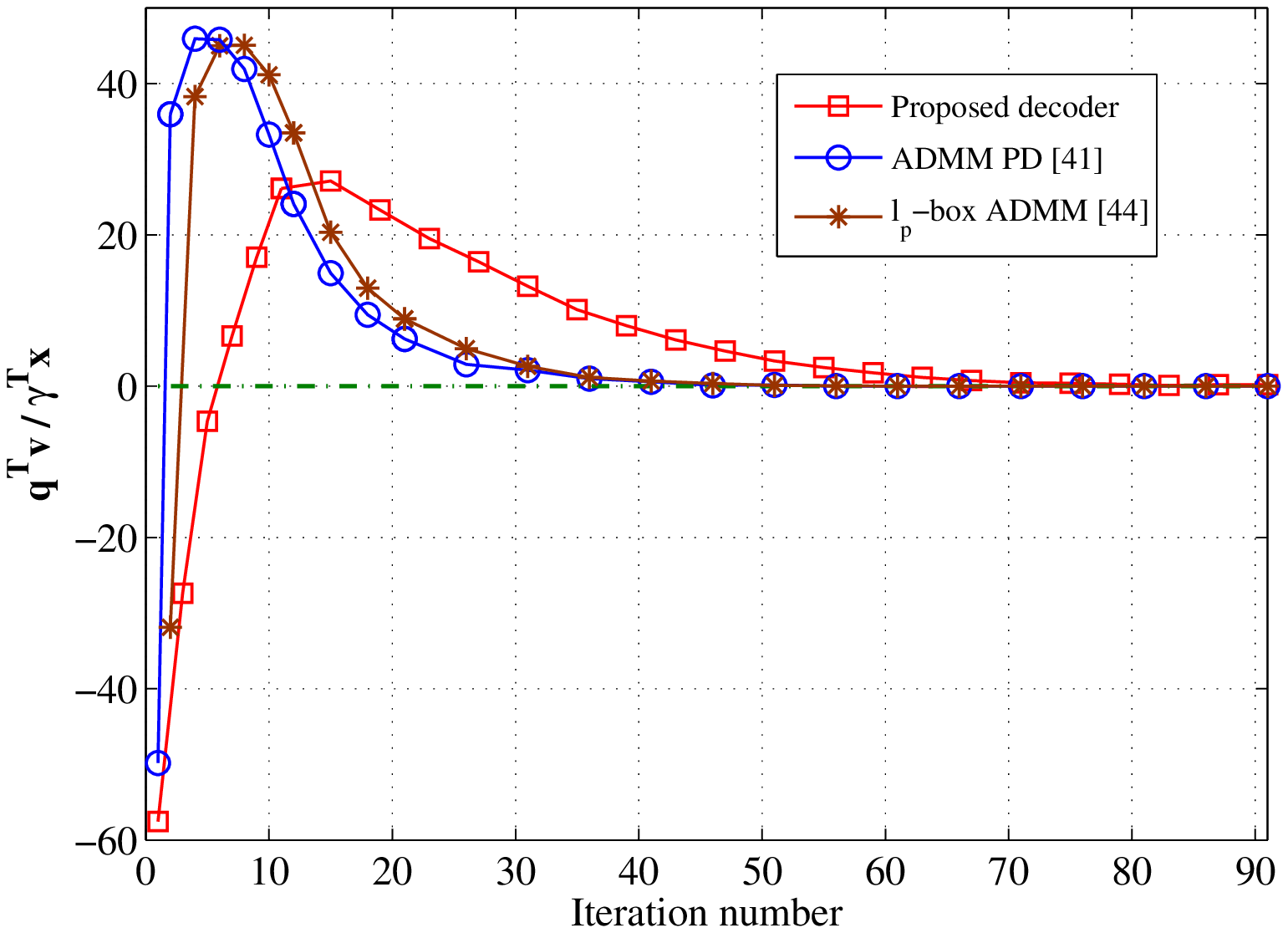}
            \label{conver1152}
    \end{minipage}%
    }
        \subfigure[Iteration number distribution of $\mathcal{C}_{5}$.]{
    \begin{minipage}{9cm}
    \centering
        \includegraphics[width=3.5in,height=2.7in]{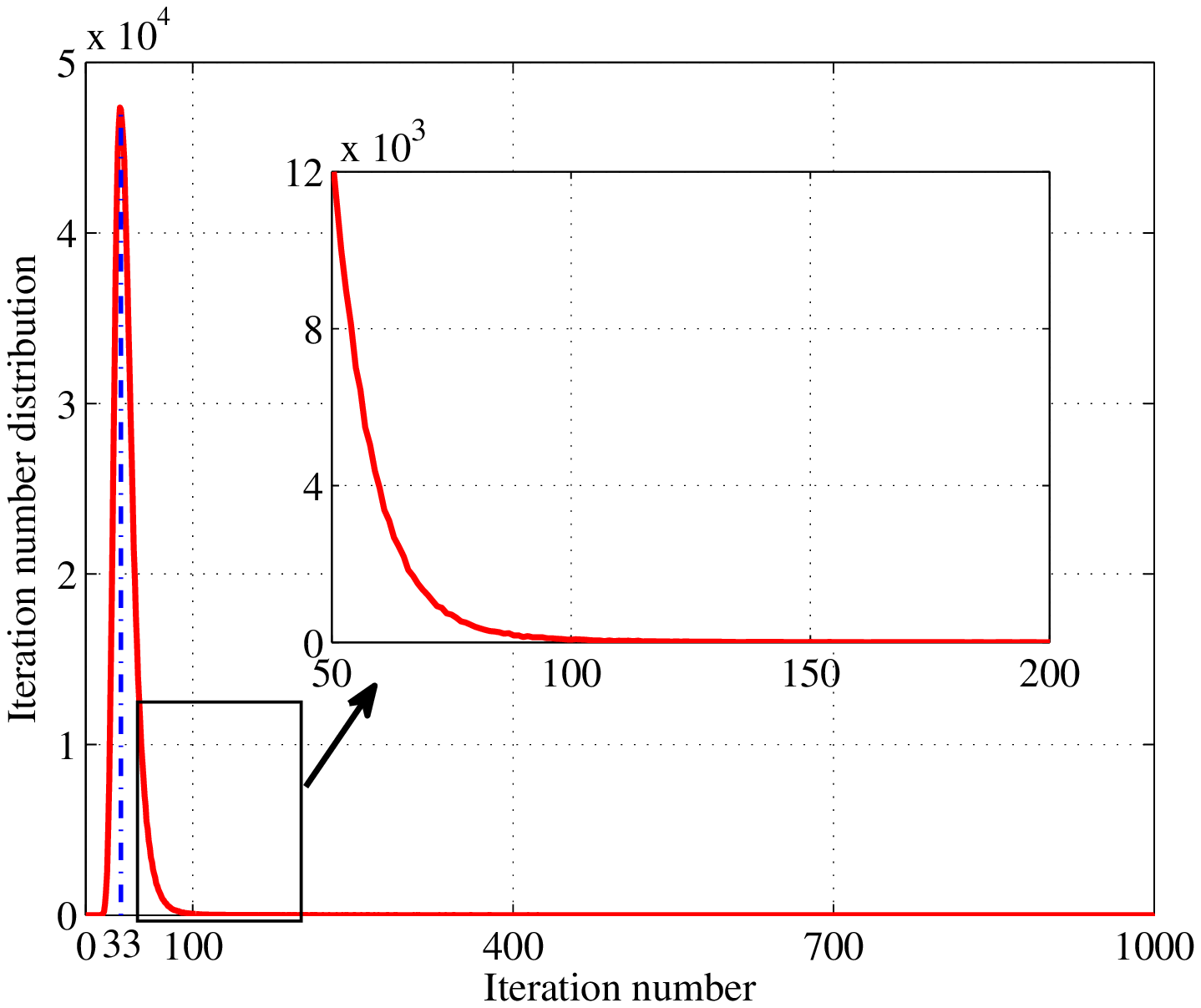}
            \label{ite1152}
    \end{minipage}%
    }\
 \centering
 \caption{ Average decoding time comparison, convergence characteristic comparison at $E_b/N_0$=3.8dB and iteration number distribution at $E_b/N_0$=3.8dB of $\mathcal{C}_{5}$, where $\mathcal{C}_{5}$ denotes the (1152,864) irregular LDPC code from IEEE 802.16e standard. All-zeros codeword is transmitted.}
 \label{time3}
 \end{figure}

  \begin{figure}[htbp]
  \subfigure[{Average decoding time of $\mathcal{C}_{6}$.}]{
    \begin{minipage}{9cm}
    \centering
        \includegraphics[width=3.5in,height=2.7in]{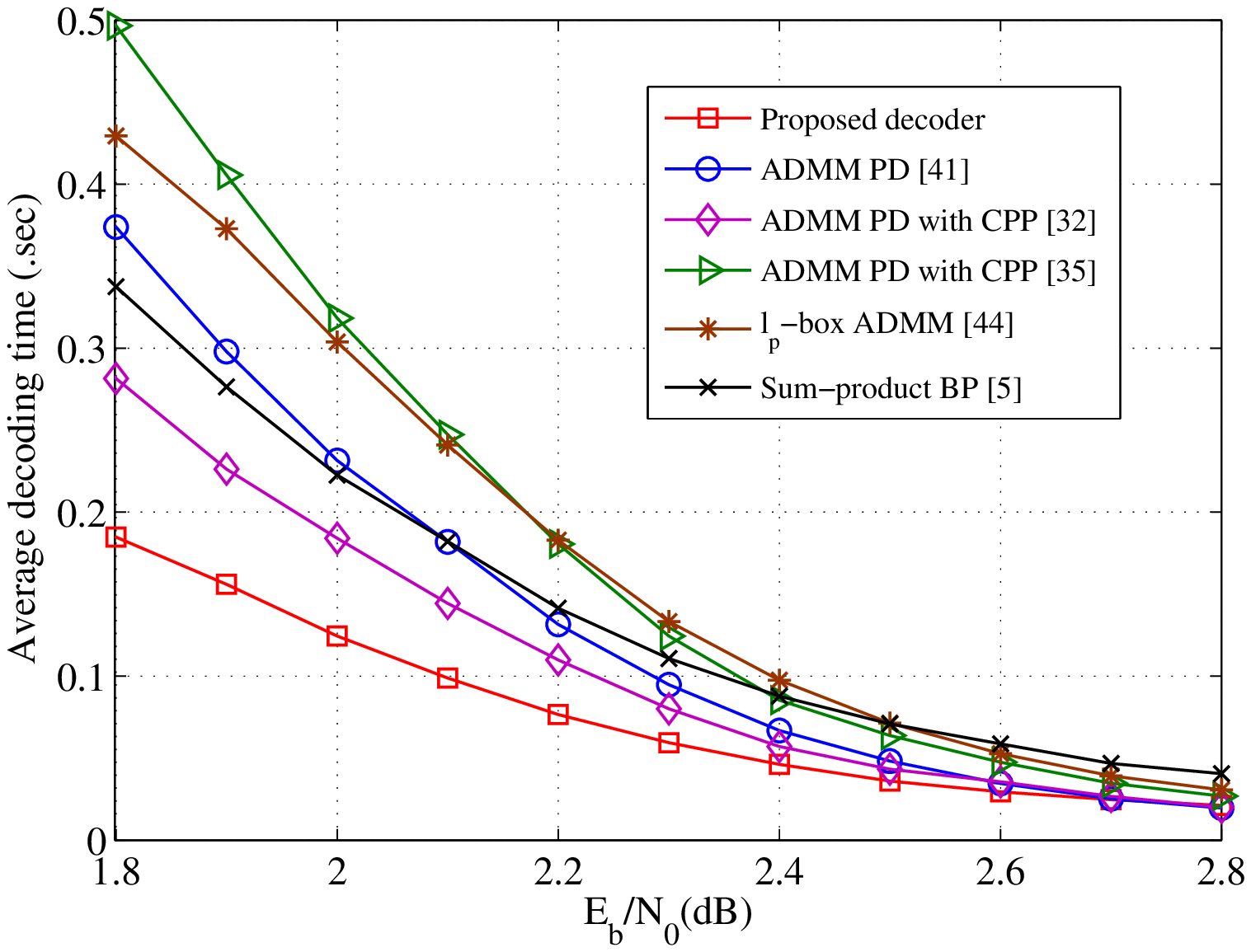}
            \label{time1296}
    \end{minipage}%
    }
    \subfigure[{Convergence characteristic of $\mathcal{C}_{6}$.}]{
    \begin{minipage}{9cm}
    \centering
        \includegraphics[width=3.5in,height=2.7in]{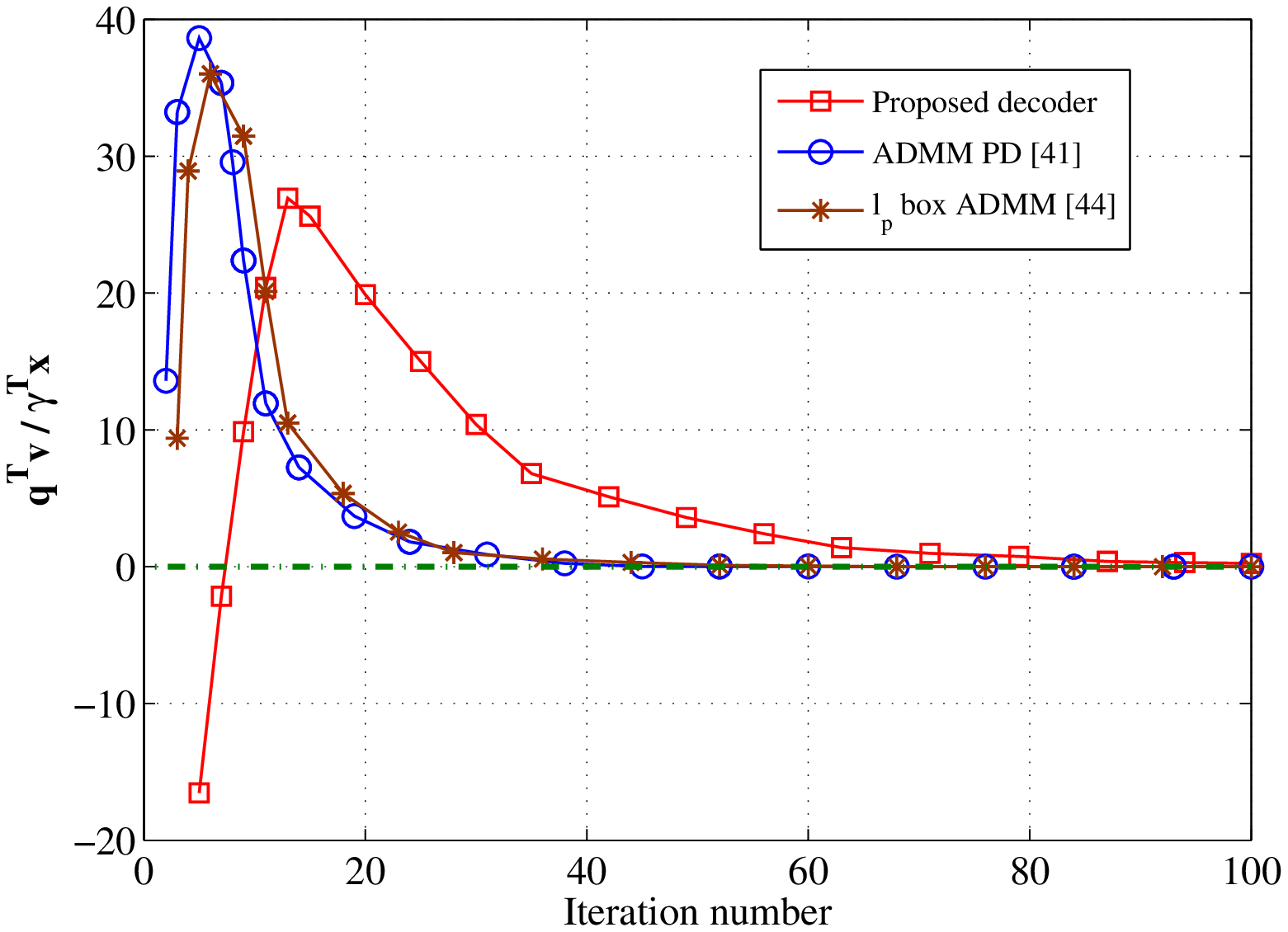}
            \label{conver1296}
    \end{minipage}%
    }
        \subfigure[{Iteration number distribution of $\mathcal{C}_{6}$.}]{
    \begin{minipage}{9cm}
    \centering
        \includegraphics[width=3.5in,height=2.7in]{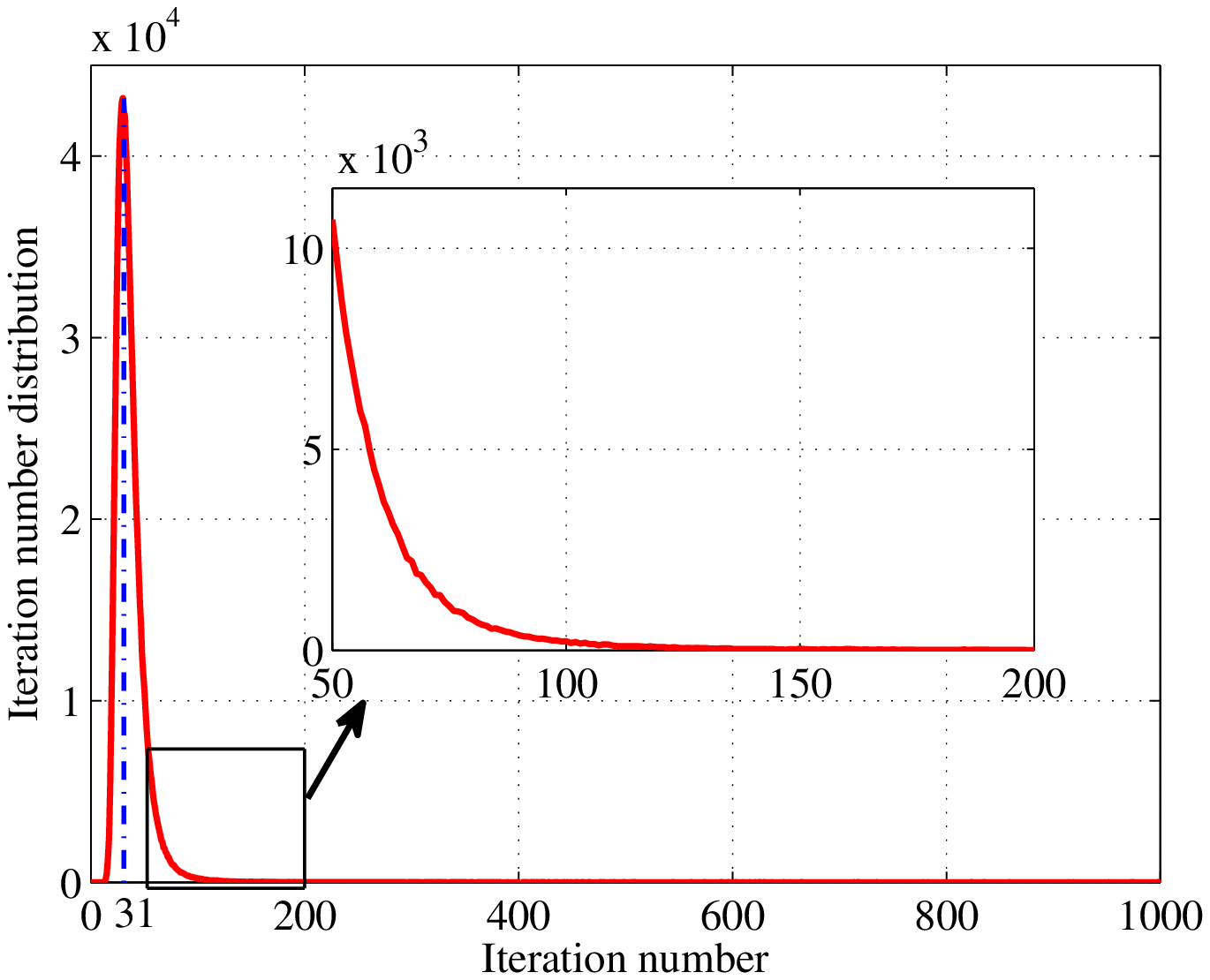}
            \label{ite1296}
    \end{minipage}%
    }\
 \centering
 \caption{{Average decoding time comparison, convergence characteristic comparison at $E_b/N_0$=3.2dB and iteration number distribution at $E_b/N_0$=3.2dB of $\mathcal{C}_{6}$, where $\mathcal{C}_{6}$ denotes the (648,432) regular LDPC code from IEEE 802.11n standard. All-zeros codeword is transmitted.}}
 \label{time3}
 \end{figure}

 Fig.\,\ref{fer_ber} shows frame error rate (FER) and bit-error-rate (BER) performance of the considered six LDPC codes. The points plotted in all FER/BER curves are based on generating at least 200 error frames, except for the last two points where 50 error frames are observed due to limited computational resources.
  In Fig.\,\ref{fer2640}, one can see that BER curves of all the decoders have a similar changing trend to their corresponding FER curves.
  Moreover, the proposed QP decoder displays comparable FER and BER performance to the sum-product BP decoder at low SNRs and continues to drop in a waterfall manner in relatively high SNR regions.
  However, the sum-product BP decoder suffers from {\it error floor effects} in these high SNRs, which can also be observed in \cite{channel-code}\cite{Barman-ADMM}\cite{penalty-decoder}.
  In addition, {our proposed QP decoder shows similar error-correction performance to the ADMM-based PD method \cite{penalty-decoder} and the $\ell_{p}$-box ADMM-based decoder \cite{lp-box-decoder}}, and outperform the LP decoders \cite{Barman-ADMM}{\cite{our-admm-lp}}.
  The FER and BER performances, of codes $\mathcal{C}_{2}$ and $\mathcal{C}_{3}$ are shown in Fig.\ref{fer13298} and Fig.\ref{fer999} respectively.
  {In Fig.\,\ref{fer13298}, we observe that, for the long block-length code $\mathcal{C}_{2}$, our proposed QP decoder and the ADMM-based decoders \cite{penalty-decoder} \cite{lp-box-decoder} achieve similar error-correction performance to the sum-product BP decoder, and display much better performance than LP decoding in terms of either FER or BER.
  Moreover, observing Fig.\,\ref{fer13298}, one can also find that FER and BER curves of our proposed decoder are quite similar to those of the ADMM-based penalized decoder in \cite{penalty-decoder} and the $\ell_{p}$-box ADMM-based decoder \cite{lp-box-decoder} at low SNRs.}
  {In Fig. \ref{fer999}, it can be seen that, for high-rate code $\mathcal{C}_{3}$, our proposed QP decoder and the ADMM-based decoders \cite{penalty-decoder} \cite{lp-box-decoder} achieves slightly better FER and BER performance than the sum-product BP decoder and the LP decoders \cite{Barman-ADMM} {\cite{our-admm-lp}} at low SNRs. Moreover, all of the ADMM-based MP decoders are superior to the sum-product BP decoder in terms of FER and BER in high SNR regions.}
  {Fig.\,\ref{fer576} shows that our proposed decoder attains similar FER performance to the ADMM-based decoders \cite{penalty-decoder} \cite{lp-box-decoder}}, and achieves better error-correction performance than the sum-product BP decoding and LP decoding. Moreover, our proposed decoder outperforms ADMM-based decoders \cite{penalty-decoder} \cite{lp-box-decoder} and the classical BP decoder in terms of BER.
  Observing Fig.\,\ref{fer1152}, one can find that FER and BER performances of our proposed decoder, ADMM-based penalized decoder \cite{penalty-decoder} and $\ell_{p}$-box ADMM-based decoder \cite{lp-box-decoder} are comparable to the sum-product BP decoder in low SNR regions and outperform the classical BP decoder in high SNR regions.
  {Fig.\ref{fer1296} shows that, for the wifi code $\mathcal{C}_{6}$, our proposed QP decoder displays similar FER/BER performance as the ADMM-based penalized decoder \cite{penalty-decoder} and the $\ell_{p}$-box ADMM-based decoder \cite{lp-box-decoder}, and provides much better error-correction performance compared to the LP decoders \cite{Barman-ADMM} \cite{our-admm-lp}. What's more, our proposed QP-ADMM decoder performs superior to the sum-product BP decoder in terms of FER and BER in high SNR regions.}

  Fig.\,\ref{time2640}-\ref{time1296} compare average decoding time between our proposed decoder and other competing ADMM-based MP decoders at different SNRs for the six codes $\mathcal{C}_{1}$-$\mathcal{C}_{5}$.
  Moreover, in each figure, we also present the average decoding time of the conventional sum-product BP decoder, {which is dominated by many expensive hyperbolic tangent operations}.
  The decoding time in each curve of Fig.\,\ref{time2640}-\ref{time1296} is averaged over one million LDPC frames.
  {From these figures, one can find that all of the compared decoders in low SNR regions spend much more decoding time than those in high SNR regions. This is because the failure of decoding usually happens at low SNRs, which usually takes the maximum number of iterations.}
  Furthermore, we can also find that our proposed ADMM-based QP decoder costs the least amount of decoding time among the ADMM-based decoders and {displays comparable decoding complexity to the conventional sum-product BP decoder for all of the involved LDPC codes $\mathcal{C}_{1}$-$\mathcal{C}_{6}$}.

  {Fig.\,\ref{conver2640}-\ref{conver1296} compare the convergence performance among our proposed ADMM-based QP decoding algorithm, ADMM-based PD algorithm \cite{penalty-decoder} and $\ell_{p}$-box ADMM-based decoder \cite{lp-box-decoder} when codes $\mathcal{C}_{1}$- $\mathcal{C}_{6}$ are all considered.
  Each curve in Fig.\,\ref{conver2640}-\ref{conver1296} is averaged over one hundred LDPC frames. Since our proposed QP decoder and the ADMM PD method in \cite{penalty-decoder} both satisfy the property of the \emph{all-zeros assumption}, we use the all-zeros codeword as the transmitted code.
  From \eqref{QP a}, one can find that, if the above three decoders converge to a correct codeword, their linear terms should be ``0''.
  Therefore, in order to present the convergence curve clearly, we choose the linear terms as the vertical axis.}  {Here, it should be noted that, although we allow up to 1000 iterations in the simulations, most of the decoding procedures do not reach the pre-set maximum number of iterations. In fact, the majority of the iteration numbers is less than 100.
  To illustrate this fact clearly, we present distribution curves of iteration numbers for all of the considered LDPC codes in Fig.\,\ref{ite2640}-\ref{ite1296} by running $10^6$ frames. Observing these figures, one can see that the numbers of iterations required to converge for most decoding procedures are less than 100 and only a very small ones take more than 100 iterations.}

     \begin{figure}[tp]
   \subfigure[FER performance vs. $\mu$.]{
    \begin{minipage}{9cm}
    \centering
        \includegraphics[width=3.5in,height=2.7in]{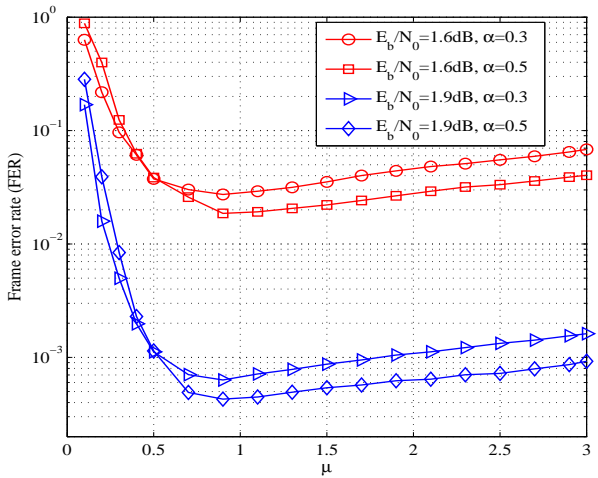}
            \label{miu_fer}
    \end{minipage}%
    }
    ~~~~~~~~~~~~~
    \subfigure[Iteration number vs. $\mu$.]{
    \begin{minipage}{9cm}
    \centering
        \includegraphics[width=3.5in,height=2.7in]{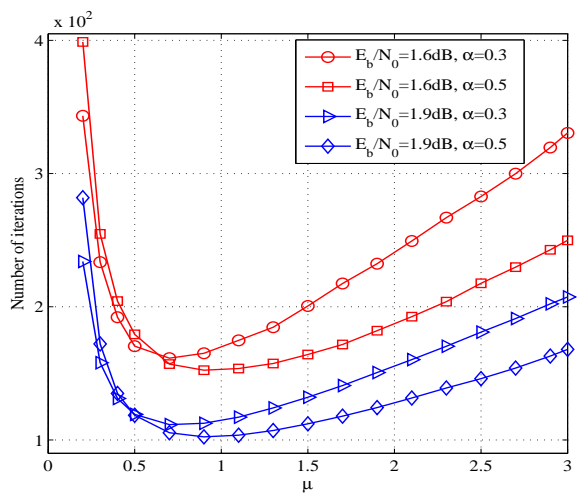}
            \label{miu_time}
    \end{minipage}%
    }
 \centering
 \caption{FER performance and average iteration number of the (2640,1320) ``Margulis'' code $\mathcal{C}_{1}$ plotted as a function of $\mu$.}
 \label{miu_fer_time}
\end{figure}
\begin{figure}[tp]
   \subfigure[FER performance vs. $\alpha$.]{
    \begin{minipage}{9cm}
    \centering
        \includegraphics[width=3.5in,height=2.7in]{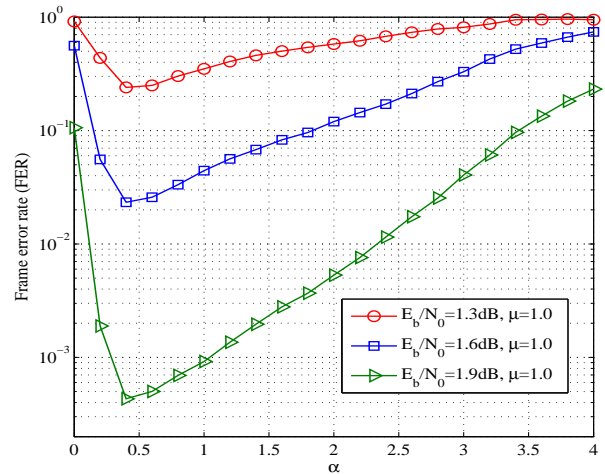}
            \label{alpha_fer}
    \end{minipage}%
    }
    ~~~~~~~~~~~~~
    \subfigure[Iteration number vs. $\alpha$.]{
    \begin{minipage}{9cm}
    \centering
        \includegraphics[width=3.5in,height=2.7in]{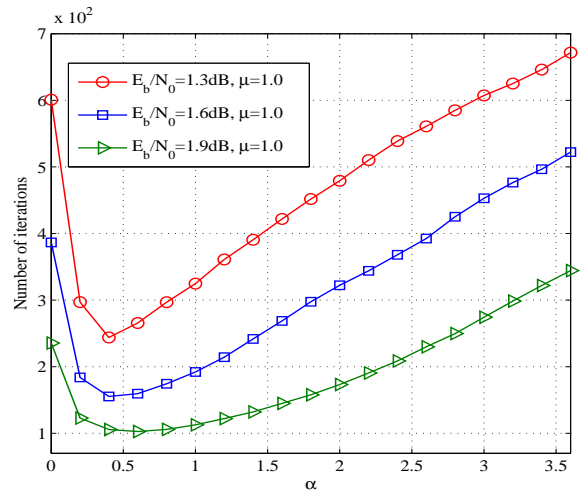}
            \label{alpha_time}
    \end{minipage}%
    }
 \centering
 \caption{FER performance and average iteration number of the (2640,1320) ``Margulis'' code $\mathcal{C}_{1}$ plotted as a function of $\alpha$.}
 \label{alpha_fer_time}
\end{figure}

\subsection{Choosing proper parameters $\mu$ and $\alpha$}\label{parameter-choices-alg1-simulation-result}

   Parameters $\mu$ and $\alpha$ can significantly affect decoding performance.
   In this subsection, we show several simulations on the effects of parameters $\mu$ and $\alpha$. From them, we can choose proper $\mu$ and $\alpha$ to obtain good decoding performance in terms of the error-correction performance and decoding efficiency.

   In Fig.\,\ref{miu_fer_time}, $\mathcal{C}_{1}$'s FER performance and the average number of iterations as a function of the penalty parameter $\mu$ with different SNRs and $\alpha$ are presented.
   We collect 200 error LDPC frames for every data point in the figures. From them, one can observe that, first, all of the curves of the FER performance in Fig.\,\ref{miu_fer} display a similar changing trend as $\mu$ increases with different $\alpha$ at both $E_{b}/N_{0}$ = 1.6\,dB and $E_{b}/N_{0}$ = 1.9\,dB, as do the curves of the average number of iterations in Fig.\,\ref{miu_time}. Second, in comparison with the FER performance, iteration number is sensitive to parameter $\mu$. Moreover, one can see that $\mu \in [0.7,1.1]$ can lead to better FER performance and fewer iterations.

   In Fig.\,\ref{alpha_fer_time}, we study the effects of the penalty parameter $\alpha$ on FER performance and iteration number. Similarly, 200 error LDPC frames are collected for each data point in Fig.\,\ref{alpha_fer_time}. From it, one can have the following observations. First, all of the curves show a similar tendency that FER performance and the average number of iterations change over the given range of parameter $\alpha$. Second, parameter $\alpha$ has a significant effect on both FER performance and iteration number. In addition, FER and the average number of iterations perform well when $\alpha \in [0.3, 1]$.

  \section{Conclusion} \label{Conclusion}
  In this paper, we propose an efficient QP-ADMM decoder for binary LDPC codes.
  Compared with the existing ADMM-based MP decoders, our proposed decoder eliminates the check polytope projection and each updated step can be implemented in parallel.
  We prove that the proposed decoding algorithm satisfies the favorable \emph{all-zeros assumption} property.
  Moreover, by exploiting the inside structure of the formulated QP decoding model, we show that the decoding complexity of the proposed ADMM algorithm in each iteration is linear in terms of LDPC code length.
  Simulation results demonstrate its effectiveness of the proposed LDPC decoder. Speeding up the convergence rate, {improving decoding performance by dynamically adjusting parameters $\alpha$ and $\mu$, addressing theoretical-guaranteed convergence issue, and extending to nonbinary LDPC code are interesting research directions in future.}

\appendices

\section{Proof of Theorem \ref{conver-theorem}}\label{convergence-proof}
In this appendix, we show that \emph{Theorem \ref{conver-theorem}} is valid.
\begin{proof}
   Under the assumption that \emph{Algorithm \ref{ADMM-penalized-alg}} converges, we have
   \begin{equation}\label{assum}
   \underset{k \rightarrow \infty}\lim \mathbf{v}^{k}= \mathbf{v}^{*}, ~~\underset{k \rightarrow \infty}\lim \mathbf{z}^{k} = \mathbf{z}^{*},~~ \underset{k \rightarrow \infty}\lim \mathbf{y}^{k} = \mathbf{y}^{*}.
   \end{equation}
  Moreover, observing \eqref{lamda-update-ori-penalized}, we obtain that, when $k \rightarrow \infty$, the following equality holds
    \begin{equation}\label{v-w-equation}
       \mathbf{Av}^{*}+\mathbf{z}^{*}-\mathbf{b}=\textbf{0}.
   \end{equation}
   Since all of the elements in $\mathbf{v}$ and $\mathbf{z}$ are implemented with the projection operations onto the interval $[0,1]$ and the nonnegative quadrant $[0,\infty)$ in \emph{Algorithm \ref{ADMM-penalized-alg}} respectively, this implies that variables $v_{i}^{*}$ and $z_{j}^{*}$ must satisfy $v_{i}^{*} \in [0,1]$ and $z_{j}^{*} \in [0,\infty)$ respectively. Thus, we obtain
    \begin{equation}\label{v-w-inequ-constraint}
        \mathbf{v}^{*} \in [0,1]^{n+\Gamma_{a}}, ~ \textrm{and} ~ \mathbf{z}^{*} \in [0,\infty)^{4\Gamma_{c}}.
   \end{equation}
  Combining \eqref{v-w-equation} with \eqref{v-w-inequ-constraint}, we conclude that $(\mathbf{v}^{*},\mathbf{z}^{*},\mathbf{y}^{*})$ lies in the feasible region of model \eqref{QP 2}.

  In the following, we verify that $\mathbf{v}^{*}$ is a stationary point of the original non-convex QP problem \eqref{QP}.
  Letting $p(\mathbf{v})=-\frac{\alpha}{2}\|\mathbf{v}-0.5\|_{2}^{2}$ and $f(\mathbf{v})=\mathbf{q}^T\mathbf{v}-\frac{\alpha}{2}\|\mathbf{v}-0.5\|_2^2$, we obtain for any $\mathbf{v} \in \mathcal{X}$
{\setlength\abovedisplayskip{2pt}
 \setlength\belowdisplayskip{2pt}
  \setlength\jot{2pt}
     \begin{equation}\label{stationary-point-inequation}
      \begin{split}
       &(\mathbf{v}-\mathbf{v}^{*})^{T}\nabla_{\mathbf{v}}f(\mathbf{v}) \\
       =&(\mathbf{v}-\mathbf{v}^{*})^{T}\Big(\mathbf{q}+\nabla_{\mathbf{v}}p(\mathbf{v}^{*})\Big)\\
       =&(\mathbf{v}\!-\!\mathbf{v}^{*})^{T}\Big(\mathbf{q}\!+\!\nabla_{\mathbf{v}}p(\mathbf{v}^{*})\!+\!\mathbf{A}^{T}\mathbf{y}^{*}\Big)
               \!\!-\!\!(\mathbf{y}^{*})^{T}\mathbf{A}(\mathbf{v}\!-\!\mathbf{v}^{*}).    \end{split}
     \end{equation}
  It} is obvious that if we can show $(\mathbf{v}-\mathbf{v}^{*})^{T}\big(\mathbf{q}+\nabla_{\mathbf{v}}p(\mathbf{v}^{*})+\mathbf{A}^{T}\mathbf{y}^{*}\big)\geq0$ and $(\mathbf{y}^{*})^{T}\mathbf{A}(\mathbf{v}-\mathbf{v}^{*})\leq0$, then $(\mathbf{v}-\mathbf{v}^{*})^{T}\nabla_{\mathbf{v}}f(\mathbf{v})\geq0$.

  For the first term, we have the following proofs.
 By \eqref{vi-update-2fanshu}, the $v^{*}_{i}$-update can be expressed as $v_{i}^{*}=\underset{[0,1]}\Pi\overline{v}_{i}^{*}$ where
{\setlength\abovedisplayskip{2pt}
 \setlength\belowdisplayskip{2pt}
  \setlength\jot{2pt}
     \begin{equation}\label{x-hat-expression}
      \overline{v}_{i}^{*}=\frac{1}{e_{i}-\frac{\alpha}{\mu}}\bigg(\hat{\mathbf{a}}_i^T\Big(\mathbf{b}- \mathbf{z}^{*}-\frac{\mathbf{y}^{*}}{\mu}\Big)-\frac{q_{i}+0.5\alpha}{\mu}\bigg),
    \end{equation}
  which} can be further transformed to
{\setlength\abovedisplayskip{6pt}
 \setlength\belowdisplayskip{6pt}
  \setlength\jot{5pt}
    \begin{equation}\label{x-hat-expression2}
      \overline{v}_{i}^{*}=\frac{1}{e_{i}}\bigg(\hat{\mathbf{a}}_i^T\Big(\mathbf{b}- \mathbf{z}^*-\frac{\mathbf{y}^*}{\mu}\Big)-\frac{q_{i}+\nabla_{\mathbf{v}}p(\overline{v}_{i}^{*})}{\mu}\bigg),
    \end{equation}
  where} $\nabla_{\mathbf{v}}p(\overline{v}_{i}^{*})=-\alpha(\overline{v}_{i}^{*}-0.5)$.

  Writing all of the variables $\overline{v}_{i}^{*}$, $i = 1,\cdots,n+\Gamma_{a}$, in a vector manner, we get
  {\setlength\abovedisplayskip{2pt}
 \setlength\belowdisplayskip{2pt}
  \setlength\jot{2pt}
  \begin{equation}\label{x-hat-expression3}
     \begin{split}
     \hspace{-0.3cm} \overline{\mathbf{v}}^{*}&\!\!=\!\!\!\left[\!\!\!\!\begin{array}{ccc}
                                    e_{1} &  &  \\
                                    & \ddots &  \\
                                    &  & e_{n+\Gamma_{a}} \\
                                     \end{array}
                                     \!\!\!\right]
         ^{-1}\!\!\!\!\!\!\!\!\bigg(\!\!\mathbf{A}^T\!\Big(\mathbf{b}\!-\! \mathbf{z}^*\!-\!\frac{\mathbf{y}^*}{\mu}\!\Big)\!\! -\!\!\frac{\mathbf{q}\!+\!\nabla_{\mathbf{v}}p(\overline{\mathbf{v}}^{*})}{\mu}\!\!\bigg)\\
        &\!=\! (\mathbf{A}^{T}\mathbf{A})^{-1}\bigg(\mathbf{A}^T\Big(\mathbf{b}- \mathbf{z}^{*}-\frac{\mathbf{y}^{*}}{\mu}\Big)-\frac{\mathbf{q}+\nabla_{\mathbf{v}}p(\overline{\mathbf{v}}^{*})}{\mu}\bigg),
     \end{split}
  \end{equation}
  where} $\nabla_{\mathbf{v}}p(\overline{\mathbf{v}}^{*})=-\alpha(\overline{\mathbf{v}}^{*}-0.5)$.
  Then, by \eqref{x-hat-expression3}, we further obtain
{\setlength\abovedisplayskip{2pt}
 \setlength\belowdisplayskip{2pt}
  \setlength\jot{2pt}
    \begin{align}\label{transfer1}
      \mathbf{q}+\nabla_{\mathbf{v}}p(\overline{\mathbf{v}}^{*})+\mathbf{A}^{T}\mathbf{y}^{*}+\mu \mathbf{A}^{T}(\mathbf{A}\overline{\mathbf{v}}^{*}+\mathbf{z}^{*}-\mathbf{b})=\textbf{0},
   \end{align}
  which} together with \eqref{v-w-equation} leads to the following equality
{\setlength\abovedisplayskip{6pt}
 \setlength\belowdisplayskip{6pt}
  \setlength\jot{3pt}
    \begin{equation}\label{transfer2}
     \begin{split}
        & \mathbf{q}+\nabla_{\mathbf{v}}p(\mathbf{v}^{*})+\mathbf{A}^{T}\mathbf{y}^{*} \\
         = &  \mu \mathbf{A}^{T}\mathbf{A}\big(\mathbf{v}^{*}-\overline{\mathbf{v}}^{*}\big)+\nabla_{\mathbf{v}}p(\mathbf{v}^{*})-\nabla_{\mathbf{v}}p(\overline{\mathbf{v}}^{*}),
     \end{split}
    \end{equation}
   where} $\mathbf{v}^{*}=\underset{[0,1]^{^{n+\Gamma_{a}}}}\Pi(\overline{\mathbf{v}}^{*})$.

  As a result, by \eqref{transfer2}, the first term of \eqref{stationary-point-inequation} can be calculated as for any $\mathbf{v} \in \mathcal{X}$
{\setlength\abovedisplayskip{6pt}
 \setlength\belowdisplayskip{6pt}
  \setlength\jot{4pt}
    \begin{equation}\label{first term}
     \begin{split}
        & (\mathbf{v}-\mathbf{v}^{*})^{T}\big(\mathbf{q}+\nabla_{\mathbf{v}}p(\mathbf{v})+\mathbf{A}^{T}\mathbf{y}^{*}\big)         \\
          \overset{a}  = & (\mathbf{v}-\mathbf{v}^{*})^{T}(\mu \mathbf{A}^{T}\mathbf{A}-\alpha \mathbf{I})(\mathbf{v}^{*}-\overline{\mathbf{v}}^{*}) \\
          \overset{b}  =  & \sum_{i=1}^{n+\Gamma_{a}}(\mu e_{i}-\alpha)(v_{i}-v_{i}^{*})(v_{i}^{*}-\overline{v}_{i}^{*})   \\
          \overset{c}  = & \sum_{\overline{v}_{i}^{*} \geq 1} (\mu e_{i}\!-\!\alpha)(v_{i}\!-\!1)(1\!-\!\overline{v}_{i}^{*})\!-\!\sum_{\overline{v}_{i}^{*} \leq 0}(\mu e_{i}\!-\!\alpha)v_{i}\overline{v}_{i}^{*}  \\
          \overset{d}  \geq &  0.
     \end{split}
  \end{equation}
  In} \eqref{first term}, the equality ``$\overset{a}=$'' holds since $\nabla_{\mathbf{v}}p(\mathbf{v}^{*})-\nabla_{\mathbf{v}}p(\overline{\mathbf{v}}^{*})=-\alpha(\mathbf{v}^{*}
  -\overline{\mathbf{v}}^{*})$; the equality ``$\overset{b}=$'' follows from the column orthogonality of matrix $\mathbf{A}$; the equality ``$\overset{c}=$'' holds because we have $v_{i}^{*}-\overline{v}_{i}^{*} = 0$ when $\overline{v}_{i}^{*} \in [0,1]$; $v_{i}^{*}=0$ when $\overline{v}_{i}^{*}<0$; $v_{i}^{*}=1$ when $\overline{v}_{i}^{*}>1$; and the inequality ``$\overset{d}\geq$'' holds since $\mu e_{i}-\alpha\geq0$ and $v_i\in[0,1]$.

  Next, we verify $(\mathbf{y}^{*})^{T}\mathbf{A}(\mathbf{v}-\mathbf{v}^{*})\leq0$. First, since $z_{j}^{*}\geq0$, we have
{\setlength\abovedisplayskip{6pt}
 \setlength\belowdisplayskip{6pt}
  \setlength\jot{4pt}
    \begin{equation}\label{y-z-multi1}
    \begin{split}
          (\mathbf{y}^{*})^{T}\mathbf{z}^{*}
          &= \sum_{z_{j}^{*}>0}y_{j}^{*}z_{j}^{*} \\
          &=  \sum_{z_{j}^{*}>0}y_{j}^{*}\Big(b_{j}-\mathbf{a}_{j}^T\mathbf{v}^{*}-\frac{y_{j}^{*}}{\mu}\Big).
    \end{split}
   \end{equation}
  Second}, \eqref{v-w-equation} implies that $(\mathbf{y}^{*})^{T}(\mathbf{Av}^{*}+\mathbf{z}^{*}-\mathbf{b})=0$. Then, we further obtain
     \begin{equation}\label{y-z-multi2}
     \begin{split}
        (\mathbf{y}^{*})^{T}\mathbf{z}^{*}& =(\mathbf{y}^{*})^{T}(\mathbf{b}-\mathbf{Av}^{*}) \\
        & = \sum_{z_{j}^{*}>0}y_{j}^{*}(b_{j}-\mathbf{a}_{j}^T\mathbf{v}^{*})+\sum_{z_{j}^{*}=0}y_{j}^{*}(b_{j}-\mathbf{a}_{j}^{T}\mathbf{v}^{*})  \\
        & = \sum_{z_{j}^{*}>0}y_{j}^{*}(b_{j}-\mathbf{a}_{j}^T\mathbf{v}^{*}).
     \end{split}
    \end{equation}
  In \eqref{y-z-multi2}, the second equality follows from $b_{j}-\mathbf{a}_{j}^{T}\mathbf{v}^{*}=0$ when $z_{j}^{*}=0$. Combining \eqref{y-z-multi1} with \eqref{y-z-multi2}, we get $\sum_{z_{j}^{*}>0}\Big \{\frac{(y_{j}^{*})^{2}}{\mu}\Big\} = 0$, which leads to
    \begin{equation}\label{y-property1}
    y_{j}^{*}=0, ~ \textrm{if} ~ z_{j}^{*}>0.
   \end{equation}
  Plugging \eqref{y-property1} into \eqref{y-z-multi2}, we obtain
     \begin{equation}\label{y-z-multi}
      (\mathbf{y}^{*})^{T}\mathbf{z}^{*}=0,
     \end{equation}
  which with \eqref{v-w-equation} together leads to the following equality
     \begin{equation}\label{y-Ax-b}
      (\mathbf{y}^{*})^{T}(\mathbf{A}\mathbf{v}^{*}-\mathbf{b})=0.
     \end{equation}

  Now consider the case where $z_{j}^{*}=0$. \eqref{zi update} implies that $b_{j}-\mathbf{a}_{j}^{T}\mathbf{v}^{*}-\frac{y_{j}^{*}}{\mu} < 0$. Hence, we get
    \begin{equation}\label{y-property2-1}
     y_{j}^{*}/\mu \geq b_{j}-\mathbf{a}_{j}^{T}\mathbf{v}^{*} \geq 0,
   \end{equation}
  where the second inequality follows from \eqref{v-w-equation} and \eqref{v-w-inequ-constraint}. Clearly, by \eqref{y-property2-1} we have
     \begin{equation}\label{y-property2}
      y_{j}^{*}\geq0, ~ \textrm{if} ~ z_{j}^{*}=0.
     \end{equation}
  Combining \eqref{y-property1} with \eqref{y-property2}, we obtain
     \begin{equation}\label{y-nonnegative}
      \mathbf{y}^{*}\succeq 0,
     \end{equation}
  which implies that for any $\mathbf{v} \in \mathcal{X}$
     \begin{equation}\label{second-term}
      (\mathbf{y}^{*})^{T}(\mathbf{A}\mathbf{v}-\mathbf{b})\leq 0,
     \end{equation}
  where $\mathbf{A}\mathbf{v}-\mathbf{b} \preceq \mathbf{0}$ holds because $\mathbf{v} \in \mathcal{X}$.
  Furthermore, combining \eqref{y-Ax-b} with \eqref{second-term}, we get
     \begin{equation}\label{second-term2}
      (\mathbf{y}^{*})^{T}\mathbf{A}(\mathbf{v}-\mathbf{v}^{*}) = (\mathbf{y}^{*})^{T}(\mathbf{A}\mathbf{v}-\mathbf{b})\leq 0.
     \end{equation}

  Finally, \eqref{second-term2} and \eqref{first term} together imply that the value of \eqref{stationary-point-inequation} is no less than 0, i.e.,
     \begin{equation}\label{x-derivation}
      (\mathbf{v}-\mathbf{v}^{*})^{T}\nabla_{\mathbf{v}}f(\mathbf{v}^{*}) \geq 0,
    \end{equation}
  where $\mathbf{v} \in \mathcal{X}$.
\end{proof}

\section{Proof of Theorem \ref{all zero assumption}}\label{all-zero-assumption-proof}

In this appendix, we present that \emph{Theorem \ref{all zero assumption}} is valid. Before showing its proof, we give one definition and one key lemma in advance.

\emph{Definition 1: } we say that ${\rm R}_{\mathbf{v}}(\pmb{\beta})$ is the {\it relative vector} of $\pmb{\beta}$ with respect to the binary vector $\mathbf{v}$ if it satisfies
{\setlength\abovedisplayskip{2pt}
   \setlength\belowdisplayskip{2pt}
    \setlength\jot{2pt}
     \begin{equation}\label{relative_vetor}
        ({\rm R}_{\mathbf{v}}(\pmb{\beta}))_{i} = \left\{
          \begin{split}
            &\beta_{i}, ~~~~~ & v_{i}=0, \\
            &1- \beta_{i}, ~~~ & v_{i}=1. \\
            \end{split}
           \right.
    \end{equation}
 Vector} $\mathbf{v}=[\mathbf{c}; \mathbf{u}]$, where $\mathbf{c}$ is the transmitted LDPC codeword and $\mathbf{u}$ is the corresponding binary auxiliary variable.

\begin{lemma} \label{admm-stop}
let $\hat{\mathbf{v}}$ and $\hat{\mathbf{v}}^{0}$ denote the output of \emph{Algorithm \ref{ADMM-penalized-alg}} when $\mathbf{r}$ and $\mathbf{r}^{0}$ are received respectively. Here $\mathbf{r}$ and $\mathbf{r}^{0}$ are received vectors when {codeword $\mathbf{v}$ and the all-zeros codeword $0^{n+\Gamma_{a}}$ are transmitted over the channel}. Then, there exists $\hat{\mathbf{v}}^{0} = \mathrm{R}_{\mathbf{v}}(\hat{\mathbf{v}})$.
\end{lemma}

\begin{proof}
   {see Appendix \ref{proof-lemma1}}.
\end{proof}

Based on the above definition and lemma, we can show the proof for {\it Theorem \ref{all zero assumption}}.

\begin{proof}
 let $B(\mathbf{v})$ denote the set of received vectors $\mathbf{r}$ that would cause decoding failure when the transmitted codeword $\mathbf{v}$ is transmitted. Then, we need to verify
     \begin{equation}\label{err-equation}
      \sum_{\mathbf{r}\in B(\mathbf{v})}\rm{Pr}[\mathbf{r} \mid \mathbf{v}]=\sum_{\mathbf{r}^{0}\in B(0^{n+\Gamma_{a}})}\rm{Pr}[\mathbf{r}^{0} \mid 0^{n+\Gamma_{a}}].
     \end{equation}
It is obvious that \eqref{err-equation} holds if and only if the following two statements hold.
\begin{itemize}
  \item (a) $\rm{Pr}[\mathbf{r} \mid \mathbf{v}]=\rm{Pr}[\mathbf{r}^{0} \mid 0^{n+\Gamma_{a}}]$.
  \item (b) $\mathbf{r}\in B(\mathbf{v})$ if and only if $\mathbf{r}^{0}\in B(0^{n+\Gamma_{a}})$.
\end{itemize}

Statement (a) is easy to prove according to the one-by-one mapping of $\mathbf{r}$ and $\mathbf{r}^{0}$ when $\mathbf{v}$ and $0^{n+\Gamma_{a}}$ are sent over to the symmetrical channel.
The detailed derivations can be found in \emph{Theorem 6} in \cite{FeldmanLP}.

As for statement (b), by \emph{Lemma \ref{admm-stop}}, we have ${\hat{\mathbf{v}}}^{0} = \mathrm{R}_{\mathbf{v}}(\hat{\mathbf{v}})$, which means that ${{\mathbf{r}}}^0$  is decoded correctly only when ${\mathbf{r}}$ is decoded correctly. Hence, we can conclude that  $\mathbf{r} \in B(\mathbf{v})$ if and only if $\mathbf{r}^{0} \in B(0^{n+\Gamma_{a}})$. Thus, the second statement (b) holds.
\end{proof}

\section{Proof of Lemma 1}\label{proof-lemma1}

Before proving {\it Lemma 1}, we present {\it Definition 2} and {\it Lemma 2}.

In \eqref{QP 2}, we introduce auxiliary variable $\mathbf{z}$ to change  inequality constraints \eqref{QP b} to equality constraints \eqref{QP 2 b}. Observing \eqref{degree-3 linear inequalities}, we can see that four auxiliary variables are required for transforming inequality constraints in \eqref{degree-3 linear inequalities} to equalities. Based on these observations, we define the following \emph{mapping operators} $\mathcal{M}_{\mathbf{v}_{\tau}}(\cdot)$ and $\mathcal{T}_{\mathbf{v}}(\cdot)$ for the introduced auxiliary variables. Both of them can also be extended to the Lagrangian multipliers, which are related to the corresponding equality constraints.

\emph{Definition 2:} let $\mathbf{z}_{\tau}=[z_{\tau_{1}}\ z_{\tau_{2}}\ z_{\tau_{3}}\ z_{\tau_{4}}]^{T}$ be the auxiliary variables which change the inequality constraints corresponding to the $\tau th$ three-variables {parity-check} equation to the equality constraints and denote $\mathbf{v_{\tau}}=\mathbf{Q_{\tau}\mathbf{v}}$ (see eq.(10)), $\tau=1,\dotsb,\Gamma_c$, be the corresponding binary variables.
Then, we define the following \emph{mapping operator} for $\mathbf{z}_{\tau}$
       \begin{equation}\label{relative_vetor}
        \mathcal{M}_{\mathbf{v}_{\tau}}(\mathbf{z}_{\tau}) = \left\{
          \begin{split}
            &[z_{\tau_{2}}\ z_{\tau_{1}}\ z_{\tau_{4}}\ z_{\tau_{3}}]^{T}, ~~~ & \mathbf{v}_{\tau}=[1\ 1\ 0]^{T}, \\
            &[z_{\tau_{3}}\ z_{\tau_{4}}\ z_{\tau_{1}}\ z_{\tau_{2}}]^{T}, ~~~ & \mathbf{v}_{\tau}=[1\ 0\ 1]^{T}, \\
            &[z_{\tau_{4}}\ z_{\tau_{3}}\ z_{\tau_{2}}\ z_{\tau_{1}}]^{T}, ~~~ & \mathbf{v}_{\tau}=[0\ 1\ 1]^{T}, \\
            &\mathbf{z}_{\tau},                                   \hspace{1.2cm} & \mathbf{v}_{\tau}=[0\ 0\ 0]^{T}.
          \end{split}
          \right.
       \end{equation}
Moreover, we also define
     \begin{equation}
      \begin{split}
        &\mathcal{T}_{\mathbf{v}}(\mathbf{z})\!=\![\mathcal{M}_{\mathbf{v}_{1}}(\mathbf{z}_{1});\cdots;\mathcal{M}_{\mathbf{v}_{\tau}}(\mathbf{z}_{\tau});\cdots;\mathcal{M}_{\mathbf{v}_{\Gamma_{c}}}(\mathbf{z}_{\Gamma_{c}})]^T.
      \end{split}
     \end{equation}

\begin{lemma} \label{admm-iter}
 let tuple $\{\mathbf{v}^{k}, \mathbf{z}^{k}, \mathbf{y}^{k}\}$ and ${\{{\mathbf{v}^{0}}^k, {\mathbf{z}^{0}}^k, {\mathbf{y}^{0}}^k}\}$ be the updated variables in the $kth$ iteration of \emph{Algorithm \ref{ADMM-penalized-alg}} when codeword $\mathbf{c}$ and the all-zeros codeword are transmitted over to the channel respectively.
 If ${\mathbf{v}^{0}}^k = {\rm R}_{\mathbf{v}}(\mathbf{v}^{k})$, ${\mathbf{z}^{0}}^k = \mathcal{T}_{\mathbf{v}}(\mathbf{z}^{k})$, and ${\mathbf{y}^{0}}^k = \mathcal{T}_{\mathbf{v}}(\mathbf{y}^{k})$, there exist
 \[
   \begin{split}
   &{\mathbf{v}^{0}}^{k+1} = \mathrm{R}_{\mathbf{v}}(\mathbf{v}^{k+1}),\  {\mathbf{z}^{0}}^{k+1} = \mathcal{T}_{\mathbf{v}}(\mathbf{z}^{k+1}), \ {\mathbf{y}^{0}}^{k+1} = \mathcal{T}_{\mathbf{v}}(\mathbf{y}^{k+1}).
   \end{split}
 \]
\end{lemma}

\begin{proof}
   {see Appendix \ref{proof-lemma2}}.
\end{proof}

Now it is ready to prove {\it Lemma \ref{admm-stop}}.

\begin{proof}
 let $\{{\mathbf{z}^0}^0, {\mathbf{y}^0}^0\}$ and $\{{\mathbf{z}^0}, \mathbf{y}^0\}$ be the initial values when we decode the all-zeros codeword and general codeword $\mathbf{c}$ using {\it Algorithm \ref{ADMM-penalized-alg}}. We set them to satisfy ${\mathbf{z}^{0}}^{0}=\mathcal{T}_{\mathbf{v}}(\mathbf{z}^{0})$ and ${\mathbf{y}^{0}}^{0}=\mathcal{T}_{\mathbf{v}}(\mathbf{y}^{0})$.
 Then, in the first iteration, by {\it Lemma \ref{admm-iter}}, we have  ${\mathbf{v}^{0}}^{1} = \rm{R}_{\mathbf{v}}(\mathbf{v}^{1})$, ${\mathbf{z}^{0}}^{1}=\mathcal{T}_{\mathbf{v}}(\mathbf{z}^{1})$ and ${\mathbf{y}^{0}}^{1}=\mathcal{T}_{\mathbf{v}}(\mathbf{y}^{1})$. Continuing with the iterations, we can get ${\mathbf{z}^{0}}^{k} = \rm{R}_{\mathbf{v}}(\mathbf{v}^{k})$,  ${\mathbf{z}^{0}}^{k}=\mathcal{T}_{\mathbf{v}}(\mathbf{z}^{k})$ and ${\mathbf{z}^{0}}^{k}=\mathcal{T}_{\mathbf{v}}(\mathbf{y}^{k})$.

Consider the termination criterion in {\it Algorithm \ref{ADMM-penalized-alg}}. We have the following derivations
{\setlength\abovedisplayskip{6pt}
   \setlength\belowdisplayskip{6pt}
    \setlength\jot{2pt}
     \begin{equation}\label{termination1}
      \begin{split}
        \hspace{0.3cm} \parallel \mathbf{Av}^{k}+\mathbf{z}^{k}-\mathbf{b} \parallel_{2}^{2} &= \mathop{\sum}\limits_{\tau =1}^{\Gamma_{c}}\parallel\mathbf{T}\mathbf{v}_{\tau}^{k}+\mathbf{z}_{\tau}^{k}-\mathbf{w}\parallel_{2}^{2} \\
       & = \mathop{\sum}\limits_{\tau =1}^{\Gamma_{c}}\parallel \mathbf{T}{\mathbf{v}_{\tau}^{0}}^k+{\mathbf{z}_{\tau}^{0}}^k-\mathbf{w}\parallel_{2}^{2} \\ &=  \parallel \mathbf{A}{\mathbf{v}^{0}}^k+{\mathbf{z}^{0}}^k-\mathbf{b} \parallel_{2}^{2},
      \end{split}
     \end{equation}
where} the second equality holds since ${\mathbf{v}^{0}}^{k}=\mathrm{R}_{\mathbf{v}}(\mathbf{v}^{k})$ and ${\mathbf{z}_{\tau}^{0}}^k=\mathcal{M}_{\mathbf{v}_{\tau}}(\mathbf{z}_{\tau}^k)$ and we can find that vector $\mathbf{T}\mathbf{v}_{\tau}^{k}+\mathbf{z}_{\tau}^{k}-\mathbf{w}$ can be obtained from the permutation of vector $\mathbf{T}{\mathbf{v}_{\tau}^{0}}^k+{\mathbf{z}_{\tau}^{0}}^k-\mathbf{w}$.
It means that both of the decoding procedures in {\it Algorithm \ref{ADMM-penalized-alg}}  for $\mathbf{r}$ and $\mathbf{r}^{0}$ should always be terminated at the same time. Therefore, we can conclude $\hat{\mathbf{v}}^{0} = \mathrm{R}_{\mathbf{v}}(\hat{\mathbf{v}})$.
\end{proof}

\section{Proof of Lemma \ref{admm-iter}}\label{proof-lemma2}
Before proving \emph{Lemma \ref{admm-iter}}, we give the following corollary based on \emph{Definition 2} in appendix \ref{proof-lemma1}.

\emph{Corollary 1:} the $\tau th$ three-variables {parity-check} equation, $\mathbf{z}_{\tau}$ and its \emph{mapping operator} $\mathcal{M}_{\mathbf{v}_{\tau}}(\mathbf{z}_{\tau})$ have the following property
{\setlength\abovedisplayskip{5pt}
   \setlength\belowdisplayskip{5pt}
    \setlength\jot{3pt}
     \begin{equation*}\label{tz-tz0}
       \mathbf{t}_{\ell}^{T}\mathcal{M}_{\mathbf{v}_{\tau}}(\mathbf{z}_{\tau}) = \left\{
        \begin{split}
          &~~~\mathbf{t}_{\ell}^{T}\mathbf{z}_{\tau}, ~~~~~ & v_{\tau_{\ell}}=0, \\
          &-\mathbf{t}_{\ell}^{T}\mathbf{z}_{\tau}, ~~~ & v_{\tau_{\ell}}=1, \\
        \end{split}
        \right.
     \end{equation*}
where} $\mathbf{t}_{\ell}$, $\ell=1,2,3$, denotes the $\ell th$ column of matrix $\mathbf{T}$.
\begin{proof}
   {see Appendix \ref{proof-corollary1}}.
\end{proof}

Now we present the proof of \emph{Lemma \ref{admm-iter}}.

\begin{proof}
to simplify the derivations, we denote \eqref{vi-update-2fanshu} as $v_{i}^{k+1} = \underset{[0,1]}\Pi(\overline{v}_{i}^{k+1})$, where
{\setlength\abovedisplayskip{6pt}
   \setlength\belowdisplayskip{6pt}
    \setlength\jot{2pt}
     \begin{equation}
       \overline{v}_{i}^{k+1} = \frac{1}{e_{i}-\frac{\alpha}{\mu}}\bigg(\mathbf{a}_{i}^{T}\Big(\mathbf{b}-\mathbf{z}^{k}-\frac{\mathbf{y}^{k}}{\mu}\Big)-\frac{2q_{i}+ \alpha}{2\mu} \bigg).
     \end{equation}

We} first consider the case that $v_i=1$ is transmitted. Then, there exist
{\setlength\abovedisplayskip{6pt}
   \setlength\belowdisplayskip{6pt}
    \setlength\jot{6pt}
      \begin{equation}\label{v-component-update2}
       \begin{split}
        1-\overline{v}_{i}^{k+1} & \overset{a}= \frac{1}{e_{i}-\frac{\alpha}{\mu}}\bigg( \mathbf{a}_{i}^{T}\Big(\mathbf{b}+\mathbf{z}^{k}+\frac{\mathbf{y}^{k}}{\mu}\Big)+\frac{2q_{i}- \alpha}{2\mu}\bigg) \\
        & \overset{b}= \frac{1}{e_{i}\!-\!\frac{\alpha}{\mu}}\bigg(\mathbf{a}_{i}^{T}\Big(\mathbf{b}\!-\!{\mathbf{z}^{0}}^k\!-\!\frac{{\mathbf{y}^{0}}^k}{\mu}\Big)\!-\!\frac{2q_{i}^{0}\!+\! \alpha}{2\mu}\bigg),
       \end{split}
      \end{equation}
where} $q_{i}^0\in \mathbf{q}^{0}=[\pmb{\gamma}^{0}; \mathbf{0}]$ and $\pmb{\gamma}^{0}$ is the LLR vector when the all-zeros codeword is transmitted.

In \eqref{v-component-update2}, the equality ``$\overset{a}=$'' follows from $\mathbf{a}_{i}^{T}\mathbf{b}=2d_{i}$ and $e_{i}=4d_{i}$, where $d_{i}$ denotes the nonzero number in the $ith$ column of the {parity-check} matrix $\mathbf{H}$. The equality ``$\overset{b}=$'' holds if $q_{i}^{0}=-q_{i}$, $\mathbf{a}_i^T\mathbf{z}^k=-\mathbf{a}_i^T{\mathbf{z}^0}^k$, and $\mathbf{a}_i^T\mathbf{y}^{k}=-\mathbf{a}_i^T{\mathbf{y}^{0}}^k$. It is obvious that $q_{i}^{0}=-q_{i}$ when $v_{i}=1$ since the channel is assumed to be symmetrical.
Now we need to verify $\mathbf{a}_i^T\mathbf{z}^k=-\mathbf{a}_i^T{\mathbf{z}^0}^k$, and $\mathbf{a}_i^T\mathbf{y}^{k}=-\mathbf{a}_i^T{\mathbf{y}^{0}}^k$.
Observing \eqref{Abq}, one can find that there exist $\lfloor \mathbf{t}_{\ell}\rfloor$ sub-vectors same as $\mathbf{t}_{\ell}$ in $\mathbf{a}_{i}$, where $\lfloor \pmb{\theta} \rfloor$ is the number of the vectors same as $\pmb{\theta}$. Moreover, we denote $\mathbf{z}_{s_{\ell}}$ as the length-4 auxiliary variables corresponding to the $s_{\ell}th$, $s_{\ell}=1,\cdots,\lfloor\mathbf{t}_{\ell}\rfloor$, sub-vector as the same as $\mathbf{t}_{\ell}$. Based on the above statements, when $v_{i}=1$, we have
{\setlength\abovedisplayskip{6pt}
   \setlength\belowdisplayskip{6pt}
    \setlength\jot{6pt}
     \begin{equation}\label{tz-tz0-1}
      \begin{split}
        & \hspace{0cm} \mathbf{a}_i^T\mathbf{z}^k = \sum_{s_{1}=1}^{\lfloor \mathbf{t}_{1}\rfloor} \mathbf{t}_{1}^{T}\mathbf{z}_{s_{1}}^{k}+\sum_{s_{2}=1}^{\lfloor \mathbf{t}_{2}\rfloor} \mathbf{t}_{2}^{T}\mathbf{z}_{s_{2}}^{k}+\sum_{s_{3}=1}^{\lfloor \mathbf{t}_{3}\rfloor} \mathbf{t}_{3}^{T}\mathbf{z}_{s_{3}}^{k} \\
        & \hspace{0.8cm} = -\sum_{s_{1}=1}^{\lfloor \mathbf{t}_{1}\rfloor} \mathbf{t}_{1}^{T}\mathbf{z}_{s_{1}}^{0^{k}}-\sum_{s_{2}=1}^{\lfloor \mathbf{t}_{2}\rfloor} \mathbf{t}_{2}^{T}\mathbf{z}_{s_{2}}^{0^{k}}-\sum_{s_{3}=1}^{\lfloor \mathbf{t}_{3}\rfloor} \mathbf{t}_{3}^{T}\mathbf{z}_{s_{3}}^{0^{k}} \\
        & \hspace{0.8cm} = -\mathbf{a}_i^T\mathbf{z}^{0^{k}}.
      \end{split}
     \end{equation}
In} \eqref{tz-tz0-1}, the second equality holds since $\mathbf{t}_{\ell}^{T}\mathbf{z}_{s_{\ell}}^{k} = -\mathbf{t}_{\ell}^{T}\mathbf{z}_{s_{\ell}}^{0^{k}}$, $\ell=1,2,3$, when $v_{i}=1$ based on \emph{Corollary 1}.
With similar derivation of \eqref{tz-tz0-1}, we can also obtain that $\mathbf{a}_i^T\mathbf{y}^{k}=-\mathbf{a}_i^T{\mathbf{y}^{0}}^k$ when $v_{i}=1$. Thus, equation \eqref{v-component-update2} holds, which means that ${\overline{v}_{i}^{0}}^{k+1} = 1-\overline{v}_{i}^{k+1}$, when $v_{i}=1$.
Moreover, since $\underset{[0,1]}\Pi({\overline{v}_{i}^{0}}^{k+1})=1-\underset{[0,1]}\Pi(\overline{v}_{i}^{k+1})$, we have
{\setlength\abovedisplayskip{2pt}
   \setlength\belowdisplayskip{2pt}
    \setlength\jot{2pt}
\begin{equation}\label{v-update-result}
{v_{i}^{0}}^{k+1} = 1-v_{i}^{k+1}.
\end{equation}

For} the case when $v_i=0$ is transmitted over to the channel, there exist
{\setlength\abovedisplayskip{10pt}
   \setlength\belowdisplayskip{10pt}
    \setlength\jot{6pt}
     \begin{equation}\label{v-component-update3}
       \overline{v}_{i}^{k+1} = \frac{1}{e_{i}-\frac{\alpha}{\mu}}\bigg(\mathbf{a}_{i}^{T}\Big(\mathbf{b}-{\mathbf{z}^{0}}^k-\frac{{\mathbf{y}^{0}}^k}{\mu}\Big)-\frac{2q_{i}^{0}+ \alpha}{2\mu}\bigg),
     \end{equation}
where} the equality holds because $q_{i}^{0}=q_{i}$ when $v_{i}=0$ under the symmetrical channel, and based on \emph{Corollary 1}, we have
{\setlength\abovedisplayskip{6pt}
   \setlength\belowdisplayskip{6pt}
    \setlength\jot{6pt}
     \begin{equation}\label{tz-tz0-2}
       \begin{split}
         & \hspace{0cm} \mathbf{a}_i^T\mathbf{z}^k = \sum_{s_{1}=1}^{\lfloor \mathbf{t}_{1}\rfloor} \mathbf{t}_{1}^{T}\mathbf{z}_{s_{1}}^{k}+\sum_{s_{2}=1}^{\lfloor \mathbf{t}_{2}\rfloor} \mathbf{t}_{2}^{T}\mathbf{z}_{s_{2}}^{k}+\sum_{s_{3}=1}^{\lfloor \mathbf{t}_{3}\rfloor} \mathbf{t}_{3}^{T}\mathbf{z}_{s_{3}}^{k} \\
         & \hspace{0.8cm} = \sum_{s_{1}=1}^{\lfloor \mathbf{t}_{1}\rfloor} \mathbf{t}_{1}^{T}\mathbf{z}_{s_{1}}^{0^{k}}+\sum_{s_{2}=1}^{\lfloor \mathbf{t}_{2}\rfloor} \mathbf{t}_{2}^{T}\mathbf{z}_{s_{2}}^{0^{k}}+\sum_{s_{3}=1}^{\lfloor \mathbf{t}_{3}\rfloor} \mathbf{t}_{3}^{T}\mathbf{z}_{s_{3}}^{0^{k}} \\
         & \hspace{0.8cm} = \mathbf{a}_i^T\mathbf{z}^{0^{k}},
       \end{split}
     \end{equation}
when} $v_{i}=0$. With a similar derivation of \eqref{tz-tz0-2}, we also can get that $\mathbf{a}_i^T\mathbf{y}^{k}=\mathbf{a}_i^T{\mathbf{y}^{0}}^k$ when $v_{i}=0$.

Observing \eqref{v-component-update3}, we can find that, when $v_i=0$,
{\setlength\abovedisplayskip{8pt}
   \setlength\belowdisplayskip{8pt}
    \setlength\jot{3pt}
     \begin{equation}\label{v-update-result0}
      {v_{i}^{0}}^{k+1} = v_{i}^{k+1}.
     \end{equation}
Combining} \eqref{v-update-result} with \eqref{v-update-result0}, we get
{\setlength\abovedisplayskip{8pt}
   \setlength\belowdisplayskip{8pt}
    \setlength\jot{3pt}
     \begin{equation}\label{v-result}
      {\mathbf{v}^{0}}^{k+1} = \mathrm{R}_{\mathbf{v}}(\mathbf{v}^{k+1}).
     \end{equation}

Next,} we verify that ${\mathbf{z}^{0}}^{k+1} = \mathcal{T}_{\mathbf{v}}(\mathbf{z}^{k+1})$.
We first consider the relationship between $\mathbf{z}_{\tau}^{k+1}$ and ${\mathbf{z}_{\tau}^{0}}^{k+1}$.   Since ${\mathbf{y}_{\tau}^{0}}^k = \mathcal{M}_{\mathbf{v}_{\tau}}(\mathbf{y}_{\tau}^{k})$, there exist
${\mathbf{y}_{\tau}^{0}}^{k} = [y_{\tau_{2}}^{k}\ y_{\tau_{1}}^{k}\ y_{\tau_{4}}^{k}\ y_{\tau_{3}}^{k}]^{T}$ when $\mathbf{v}_{\tau}=[1\ 1\ 0]^{T}$.
Thus, when $\mathbf{v}_{\tau}=[1 \ 1\ 0]^T$, by \eqref{z-update-dual} and \eqref{v-result}, we obtain
{\setlength\abovedisplayskip{6pt}
   \setlength\belowdisplayskip{20pt}
    \setlength\jot{6pt}
\begin{equation}\label{z-tau}
\begin{split}
{\mathbf{z}_{\tau}^{0}}^{k+1} & = \underset{{[0,+\infty)}}\Pi\bigg(\mathbf{w}-\mathbf{T}{\mathbf{v}_{\tau}^{0}}^{k+1}-\frac{{\mathbf{y}_{\tau}^{0}}^{k}}{\mu}\bigg) \\
& = \underset{[0,+\infty)}\Pi\left(\mathbf{w}-\mathbf{T}\left[\!\!
                                                                 \begin{array}{c}
                                                                   1-v_{\tau_{1}}^{k+1} \\
                                                                   1-v_{\tau_{2}}^{k+1} \\
                                                                   v_{\tau_{3}}^{k+1} \\
                                                                 \end{array}
                                                               \!\!\right]
-\frac{1}{\mu}\left[\!\!
   \begin{array}{c}
     y_{\tau_{2}}^{k} \\
     y_{\tau_{1}}^{k} \\
     y_{\tau_{4}}^{k} \\
     y_{\tau_{3}}^{k} \\
   \end{array}
 \!\!\right] \right)  \\
& = \underset{[0,+\infty)}\Pi\left[\!\!
                     \begin{array}{c}
                       {v_{\tau_{1}}^{k+1}}-v_{\tau_{2}}^{k+1}+v_{\tau_{3}}^{k+1}-\frac{y_{\tau_{2}}^{k}}{\mu} \\
                       -v_{\tau_{1}}^{k+1}+v_{\tau_{2}}^{k+1}+v_{\tau_{3}}^{k+1}-\frac{y_{\tau_{1}}^{k}}{\mu} \\
                       2-\big(v_{\tau_{1}}^{k+1}+v_{\tau_{2}}^{k+1}+v_{\tau_{3}}^{k+1}\big)-\frac{y_{\tau_{4}}^{k}}{\mu} \\
                       v_{\tau_{1}}^{k+1}+v_{\tau_{2}}^{k+1}-v_{\tau_{3}}^{k+1}-\frac{y_{\tau_{3}}^{k}}{\mu} \\
                     \end{array}
                   \!\!\right].
\end{split}
\end{equation}
See} $\mathbf{T}$ and $\mathbf{w}$ in \eqref{t F matrix}. Then,  ${\mathbf{z}_{\tau}^{0}}^{k+1}$ can be further derived as
{\setlength\abovedisplayskip{10pt}
   \setlength\belowdisplayskip{16pt}
    \setlength\jot{16pt}
\[
\begin{split}
 {\mathbf{z}_{\tau}^{0}}^{k+1} &= \underset{_{[0,+\infty)}}\Pi\left[
                     \begin{array}{c}
                       w_2-\big(-{v_{\tau_{1}}^{k+1}}+v_{\tau_{2}}^{k+1}-v_{\tau_{3}}^{k+1}\big)-\frac{y_{\tau_{2}}^{k}}{\mu} \\
                       w_1-\big(v_{\tau_{1}}^{k+1}-v_{\tau_{2}}^{k+1}-v_{\tau_{3}}^{k+1}\big)-\frac{y_{\tau_{1}}^{k}}{\mu} \\
                       w_4-\big(v_{\tau_{1}}^{k+1}+v_{\tau_{2}}^{k+1}+v_{\tau_{3}}^{k+1}\big)-\frac{y_{\tau_{4}}^{k}}{\mu} \\
                       w_3-\big(-v_{\tau_{1}}^{k+1}-v_{\tau_{2}}^{k+1}+v_{\tau_{3}}^{k+1}\big)-\frac{y_{\tau_{3}}^{k}}{\mu} \end{array}
                   \right] \\
 \hspace{0.9cm} &= \Big[z_{\tau_{2}}^{k+1}~z_{\tau_{1}}^{k+1}~z_{\tau_{4}}^{k+1}~z_{\tau_{3}}^{k+1}\Big], \end{split}
\]
i.e.,}
    \begin{equation}\label{z0-z-k-1}
     \mathbf{z}_{\tau}^{{0}^{k+1}} = \mathcal{M}_{\mathbf{v}_{\tau}}(\mathbf{z}_{\tau}^{k+1}).
    \end{equation}

Through the above similar derivations, we can see \eqref{z0-z-k-1} also holds when $\mathbf{v}_{\tau}=[1 \ 0\ 1]^T$, $[0 \ 1\ 1]^T$, and $[0 \ 0\ 0]^T$. Since $\mathbf{z}^{k+1}$ is cascaded by $\mathbf{z}_{\tau}^{k+1}$, $\tau=1,\cdots, \Gamma_{c}$, and $\mathbf{z}^{{0}^{k+1}}$ is also cascaded by ${\mathbf{z}_{\tau}^{0}}^{k+1}$ in the same way, we get
     \begin{equation}\label{z_results}
      \mathbf{z}^{{0}^{k+1}} = \mathcal{T}_{\mathbf{v}}(\mathbf{z}^{k+1}).
     \end{equation}

Finally, it remains to verify that ${\mathbf{y}^{0}}^{k+1} = \mathcal{T}_{\mathbf{v}}(\mathbf{y}^{k+1})$.
We first consider the relationship between $\mathbf{y}_{\tau}^{k+1}$ and ${\mathbf{y}_{\tau}^{0}}^{k+1}$.
Since ${\mathbf{z}_{\tau}^{0}}^{k+1} = \mathcal{M}_{\mathbf{v}_{\tau}}(\mathbf{z}_{\tau}^{k+1})$ based on \eqref{z_results}, there exist
${\mathbf{z}_{\tau}^{0}}^{k+1} = [z_{\tau_{2}}^{k+1}\ z_{\tau_{1}}^{k+1}\ z_{\tau_{4}}^{k+1}\ z_{\tau_{3}}^{k+1}]^{T}$ when $\mathbf{v}_{\tau}=[1\ 1\ 0]^{T}$. Thus, when $\mathbf{v}_{\tau}=[1 \ 1\ 0]^T$, from \eqref{ADMM update_QP} and \eqref{v-result} we have
{\setlength\abovedisplayskip{10pt}
 \setlength\belowdisplayskip{15pt}
  \setlength\jot{15pt}
\begin{equation}\label{y-tau}
\begin{split}
{\mathbf{y}_{\tau}^{0}}^{k+1} & = {\mathbf{y}_{\tau}^{0}}^{k}+\mu\left(\mathbf{T}{\mathbf{v}_{\tau}^{0}}^{k+1}+\mathbf{z}_{\tau}^{{0}^{k+1}}-\mathbf{w}\right) \\
& = \left[\!\!
   \begin{array}{c}
     y_{\tau_{2}}^{k} \\
     y_{\tau_{1}}^{k} \\
     y_{\tau_{4}}^{k} \\
     y_{\tau_{3}}^{k} \\
   \end{array}
\!\! \right]
\!\! +\!\!\mu \left(\!\!\mathbf{T}\left[\!\!
      \begin{array}{c}
      1-v_{\tau_{1}}^{k+1} \\
      1-v_{\tau_{2}}^{k+1} \\
      v_{\tau_{3}}^{k+1} \\
      \end{array}
     \!\! \right]
\!\! +\!\! \left[\!\!
     \begin{array}{c}
       z_{\tau_{2}}^{k+1} \\
       z_{\tau_{1}}^{k+1} \\
       z_{\tau_{4}}^{k+1} \\
       z_{\tau_{3}}^{k+1} \\
     \end{array}
   \!\!\right]
 \!\!-\!\! \mathbf{w} \!\!\right)\\
 & =\left[\!\!\!\!
    \begin{array}{c}
      y_{\tau_{2}}^{k}\!+ \!\mu  \Big(\!-{v_{\tau_{1}}^{k+1}}\!+\!v_{\tau_{2}}^{k+1}\!-\!v_{\tau_{3}}^{k+1}\!+\!z_{\tau_{2}}^{k+1}\Big) \\
      y_{\tau_{1}}^{k}\!+\! \mu  \Big({v_{\tau_{1}}^{k+1}}\!-\!v_{\tau_{2}}^{k+1}\!-\!v_{\tau_{3}}^{k+1}\!+\!z_{\tau_{1}}^{k+1}\Big) \\
      y_{\tau_{4}}^{k}\!+\! \mu  \Big({v_{\tau_{1}}^{k+1}}\!+\!v_{\tau_{2}}^{k+1}\!+\!v_{\tau_{3}}^{k+1}\!+\!z_{\tau_{4}}^{k+1}\!-\!2\Big) \\
      y_{\tau_{3}}^{k}\!+ \!\mu  \Big(\!-\!{v_{\tau_{1}}^{k+1}}\!-\!v_{\tau_{2}}^{k+1}\!+\!v_{\tau_{3}}^{k+1}\big)\!+\!z_{\tau_{3}}^{k+1}\Big) \\
    \end{array}
  \!\!\!\!\right].
\end{split}
\end{equation}
Since} $\mathbf{w}=[0\ 0\ 0\ 2]^T$, ${\mathbf{y}_{\tau}^{0}}^{k+1}$ can be further derived as
{\setlength\abovedisplayskip{6pt}
 \setlength\belowdisplayskip{6pt}
  \setlength\jot{2pt}
\begin{equation}\label{y-tau}
\begin{split}
 \hspace{0.0cm} {\mathbf{y}_{\tau}^{0}}^{k+1} &\!\!=\!\!\left[\!\!\!\!
    \begin{array}{c}
      y_{\tau_{2}}^{k}\!\!+\! \mu  \Big(\!\big(\!-\!{v_{\tau_{1}}^{k+1}}\!+\!v_{\tau_{2}}^{k+1}\!-\!v_{\tau_{3}}^{k+1}\big)\!+\!z_{\tau_{2}}^{k+1}\!\!-\!w_{2}\Big) \\
      y_{\tau_{1}}^{k}\!\!+\! \mu  \Big(\!\big({v_{\tau_{1}}^{k+1}}\!-\!v_{\tau_{2}}^{k+1}\!-\!v_{\tau_{3}}^{k+1}\big)\!+\!z_{\tau_{1}}^{k+1}\!\!-\!w_{1}\Big) \\
      y_{\tau_{4}}^{k}\!\!+\! \mu  \Big(\!\big({v_{\tau_{1}}^{k+1}}\!+\!v_{\tau_{2}}^{k+1}\!+\!v_{\tau_{3}}^{k+1}\big)\!+\!z_{\tau_{4}}^{k+1}\!\!-\!w_{4}\Big) \\
      y_{\tau_{3}}^{k}\!\!+\! \mu  \Big(\!\big(\!-\!{v_{\tau_{1}}^{k+1}}\!-\!v_{\tau_{2}}^{k+1}\!+\!v_{\tau_{3}}^{k+1}\big)\!+\!z_{\tau_{3}}^{k+1}\!\!-\!w_{3}\Big) \\
    \end{array}
  \!\!\!\!\right]  \\
  &  =\!\! \left[y_{\tau_{2}}^{k+1}~y_{\tau_{1}}^{k+1}~y_{\tau_{4}}^{k+1}~y_{\tau_{3}}^{k+1}\right]^{T},
\end{split}
\end{equation}
i.e.,}
     \begin{equation}\label{y0-y-k-1}
       \mathbf{y}_{\tau}^{{0}^{k+1}} = \mathcal{M}_{\mathbf{v}_{\tau}}(\mathbf{y}_{\tau}^{k+1}),
     \end{equation}
when $\mathbf{v}_{\tau}=[1 \ 1\ 0]^T$.

With the above similar derivations, we can see \eqref{y0-y-k-1} also holds when $\mathbf{v}_{\tau}=[1 \ 0\ 1]^T$, $[0 \ 1\ 1]^T$, and $[0 \ 0\ 0]^T$. Since $\mathbf{y}^{k+1}$ and $\mathbf{y}^{{0}^{k+1}}$,$\tau=1,\cdots, \Gamma_{c}$, are cascaded by $\mathbf{z}_{\tau}^{k+1}$ and ${\mathbf{z}_{\tau}^{0}}^{k+1}$ respectively, we can get
\begin{equation}\label{y_results}
  \mathbf{y}^{{0}^{k+1}} = \mathcal{T}_{\mathbf{v}}(\mathbf{y}^{k+1}).
\end{equation}

This ends the proof of \emph{Lemma \ref{admm-iter}}.
\end{proof}

\section{Proof of Corollary 1} \label{proof-corollary1}

\begin{proof}
based on \emph{Definition 2}, there exist four cases for the relationship between $\mathbf{z}_{\tau}$ and its \emph{mapping operator} $\mathcal{M}_{\mathbf{v}_{\tau}}(\mathbf{z}_{\tau})$.
When $\mathbf{v}_{\tau}=[1 \ 1\ 0]^T$, we have $\mathcal{M}_{\mathbf{v}_{\tau}}(\mathbf{z}_{\tau})=[z_{\tau_{2}}\ z_{\tau_{1}}\ z_{\tau_{4}}\ z_{\tau_{3}}]^{T}$.
Thus, when $\mathbf{v}_{\tau_{1}}=1$, $\mathbf{v}_{\tau_{2}}=1$, and $\mathbf{v}_{\tau_{3}}=0$ respectively, we can obtain \eqref{case1-Mz-z-1}-\eqref{case1-Mz-z-3} as follows:
     \begin{equation}\label{case1-Mz-z-1}
      \begin{split}
       \mathbf{t}_{1}^{T}\mathcal{M}_{\mathbf{v}_{\tau}}(\mathbf{z}_{\tau})  & = -z_{\tau_{1}}+z_{\tau_{2}}+z_{\tau_{3}}-z_{\tau_{4}} \\
       & = -\mathbf{t}_{1}^{T}\mathbf{z}_{\tau},
      \end{split}
     \end{equation}
      \begin{equation}\label{case1-Mz-z-2}
       \begin{split}
          \mathbf{t}_{2}^{T}\mathcal{M}_{\mathbf{v}_{\tau}}(\mathbf{z}_{\tau}) & = z_{\tau_{1}}-z_{\tau_{2}}+z_{\tau_{3}}-z_{\tau_{4}} \\
          &  = -\mathbf{t}_{2}^{T}\mathbf{z}_{\tau},
       \end{split}
      \end{equation}
and
     \begin{equation}\label{case1-Mz-z-3}
      \begin{split}
       \mathbf{t}_{3}^{T}\mathcal{M}_{\mathbf{v}_{\tau}}(\mathbf{z}_{\tau}) & = -z_{\tau_{1}}-z_{\tau_{2}}+z_{\tau_{3}}+z_{\tau_{4}} \\
        & = \mathbf{t}_{3}^{T}\mathbf{z}_{\tau}.
      \end{split}
     \end{equation}
which can be further written as
     \begin{equation}\label{tz-tz0}
       \mathbf{t}_{\ell}^{T}\mathcal{M}_{\mathbf{v}_{\tau}}(\mathbf{z}_{\tau}) = \left\{
        \begin{split}
         &~~~\mathbf{t}_{\ell}^{T}\mathbf{z}_{\tau}, ~~~~~ & v_{\tau_{\ell}}=0, \\
         &-\mathbf{t}_{\ell}^{T}\mathbf{z}_{\tau}, ~~~ & v_{\tau_{\ell}}=1, \\
        \end{split}
        \right.
     \end{equation}
where $\ell=1,2,3$.

Moreover, through similar derivations, we can prove that \eqref{tz-tz0} also holds when $\mathbf{v}_{\tau}=[1 \ 0\ 1]^T$, $[0 \ 1\ 1]^T$, and $[0 \ 0\ 0]^T$. This ends the proof.
\end{proof}

%
%

\ifCLASSOPTIONcaptionsoff
  \newpage
\fi



%
%
%




\end{document}